\documentclass[a4paper,UKenglish]{lipics-v2016}
 
\usepackage{microtype}
\usepackage{enumitem}
\usepackage{booktabs}
\usepackage{enumitem}

\newcommand{\ietal}{\textit{et al.}}
\newcommand{\floor}[1]{\lfloor #1 \rfloor}
\title{Cache Oblivious Algorithms for Computing the Triplet Distance Between Trees\footnote{Research supported by the Danish National Research Foundation, grant DNRF84, Center for Massive Data Algorithmics (MADALGO).}}
\titlerunning{Cache Oblivious Algorithms for Computing the Triplet Distance Between Trees} 

\author[1]{Gerth Stølting Brodal}
\author[2]{Konstantinos Mampentzidis}
\affil[1]{Department of Computer Science, Aarhus University, Aarhus, Denmark\\
  \texttt{gerth@cs.au.dk}}
\affil[2]{Department of Computer Science, Aarhus University, Aarhus, Denmark\\
  \texttt{kmampent@cs.au.dk}}
\authorrunning{G.\,S. Brodal and K. Mampentzidis} 

\Copyright{Gerth Stølting Brodal and Konstantinos Mampentzidis}


\ArticleNo{1} 

\begin{document}

\maketitle

\begin{abstract}
We consider the problem of computing the triplet distance between two rooted unordered trees with~$n$ labeled leaves. Introduced by Dobson in 1975, the triplet distance is the number of leaf triples that induce different topologies in the two trees. The current theoretically fastest algorithm is an~$\mathrm{O}(n \log n)$ algorithm by Brodal~\ietal~(SODA 2013). Recently Jansson and Rajaby proposed a new algorithm that, while slower in theory, requiring~$\mathrm{O}(n \log^3 n)$ time, in practice it outperforms the theoretically faster $\mathrm{O}(n \log n)$ algorithm. Both algorithms do not scale to~external~memory.

We present two cache oblivious algorithms that combine the best of both worlds. The first algorithm is for the case when the two input trees are binary trees, and the second is a generalized algorithm for two input trees of arbitrary degree. Analyzed in the RAM model, both algorithms require $\mathrm{O}(n \log n)$ time, and in the cache oblivious model $\mathrm{O}(\frac{n}{B} \log_{2} \frac{n}{M})$ I/Os. Their relative simplicity and the fact that they scale to external memory makes them achieve the best practical performance. We note that these are the first algorithms that scale to external memory, both in theory and in practice, for this problem.
\end{abstract}

\section{Introduction}
\label{sec:introduction}

Trees are data structures that are often used to represent relationships. For example in the field of Biology, a tree can be used to represent evolutionary relationships, with the leaves corresponding to species that exist today, and internal nodes to ancestor species that existed in the past. For a fixed set of~$n$ species, different data (e.g., DNA, morphological) or construction methods (e.g., Q*~\cite{Qmethod}, neighbor joining~\cite{NJmethod}) can lead to trees that look structurally different. An interesting question that arises then is, given two trees $T_{1}$ and $T_{2}$ over $n$ species, how different are they? An answer to this question could potentially be used to determine whether the difference is statistically significant or not, which in turn could help with evolutionary inferences. 

Several distance measures have been proposed in the past to compare two trees that are \emph{unordered}, i.e., trees in which the order of the siblings is not taken into account. A class of them includes distance measures that are based on how often certain features are different in the two trees. Common distance measures of this kind are the Robinson-Foulds distance~\cite{RobinsonFoulds}, the triplet distance~\cite{Dobson} for rooted trees and the quartet distance~\cite{Quartet} for unrooted trees. The Robinson-Foulds distance counts how many leaf bipartitions are different, where a bipartition in a given tree is generated by removing a single edge from the tree. The triplet distance is only defined for rooted trees, and counts how many leaf triples induce different topologies in the two trees. The counterpart of the triplet distance for unrooted trees, is the quartet distance, which counts how many leaf quadruples induce different topologies in the two trees. 

Algorithms exist that can efficiently compute these distance measures. The Robinson-Foulds distance can be optimally computed in $\mathrm{O}(n)$ time~\cite{DaysAlgorithm}. The triplet distance can be computed in~$\mathrm{O}(n \log n)$ time~\cite{TripletQuartetSoda13}. The quartet distance can be computed in $\mathrm{O}(dn\log n)$ time~\cite{TripletQuartetSoda13}, where $d$ is the maximal degree of any node in the two input trees. 

The above bounds are in the RAM model. Previous work did not consider any other models, for example external memory models like the I/O model~\cite{IOmodel} and the cache oblivious model~\cite{COmodel}. Typically when hearing about algorithms for external memory models, one might (sometimes incorrectly) think of only algorithms that have to deal with large amounts of data. Hence, any practical improvement that comes from an algorithm that scales to external memory compared to an equivalent that does not, can only be noticed if the inputs are large. However, this is not necessarily the case for cache oblivious algorithms. A cache oblivious algorithm, if built and implemented correctly, can take advantage of the L1, L2, and L3 caches that exist in the vast majority of computers and give a significant performance improvement even for small~inputs.

A trivial modification of the algorithm in~\cite{DaysAlgorithm}, can give a cache oblivious algorithm for computing the Robinson-Foulds distance that achieves the sorting bound, by requiring~$\mathrm{O}(\frac{n}{B}\log_{\frac{M}{B}} \frac{n}{M})$ I/Os instead of $\mathrm{O}(n)$ I/Os for the standard implementation.  For the triplet and quartet distance measures, no such trivial modifications exist. 

In this paper we focus on the triplet distance computation and present the first non-trivial algorithms for computing the triplet distance between two rooted trees, that for the first time for this problem, also scale to external~memory.

\subparagraph*{Problem Definition.}

\begin{figure}
\captionsetup[subfigure]{justification=centering}
    \centering
    \begin{subfigure}{0.22\textwidth}
        \centering
        \includegraphics{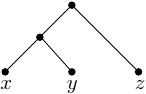}
        \caption{$xy|z$}
    \end{subfigure}%
    ~ 
    \begin{subfigure}{0.22\textwidth}
        \centering
        \includegraphics{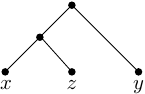}
        \caption{$xz|y$}
    \end{subfigure}%
    ~ 
    \begin{subfigure}{0.22\textwidth}
        \centering
        \includegraphics{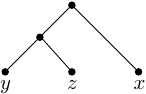}
        \caption{$yz|x$}
    \end{subfigure}%
    ~
     \begin{subfigure}{0.22\textwidth}
        \centering
        \includegraphics{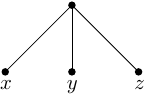}
        \caption{$xyz$}
    \end{subfigure}%
    \caption{All possible topologies of a triplet with leaves $x$, $y$, and $z$.}
    \label{fig:tripletTopologies}
\end{figure}

For a given rooted unordered tree $T$ where each leaf has a unique label, a~\emph{triplet} is defined by a set of three leaf labels $x$, $y$, and $z$ and their induced topology in~$T$. The four possible topologies are illustrated in Figure~\ref{fig:tripletTopologies}. The notation~$xy|z$ is used to describe a triplet where the lowest common ancestor of $x$ and $y$ is at a lower depth than the lowest common ancestor of $z$ with either~$x$ or $y$. Note that the triplet $xy|z$ is the same as the triplet $yx|z$ because $T$ is considered to be unordered. Similarly, notation $xyz$ is used to describe a triplet for which every pair of leaves has the same lowest common ancestor. This triplet can only appear if we allow nodes with degree three or larger in~$T$. From here on, when using the word ``tree'' we imply a ``rooted unordered tree''. 

For two given trees $T_{1}$ and $T_{2}$ that are built on $n$ identical leaf labels, the~\emph{triplet distance} $D(T_{1}, T_{2})$ is the number of triplets the leaves of which induce different topologies in $T_{1}$ and $T_{2}$. Let $S(T_{1}, T_{2})$ be the number of~\emph{shared} triplets in the two trees, i.e., leaf triples with identical topologies in the two trees.  We then have the relationship $D(T_{1}, T_{2}) + S(T_{1},T_{2}) = \binom{n}{3}$.

Previous and new results for computing the triplet distance are shown in the table below. Note that the papers \cite{Critchlow,Bansal2011,bmc13,TripletQuartetSoda13,Jansson2015} do not provide an analysis of the algorithms in the cache oblivious model, so here we provide an upper bound. From here on and unless otherwise stated, any asymptotic bound refers to time.

\enlargethispage{\baselineskip}
\begin{table}[ht]
  \centering
  \begin{tabular}{clcccc}
  	\toprule
  	Year & Reference & Time & IOs & Space & Non-Binary Trees\\
  	\midrule
    1996& Critchlow~\ietal~\cite{Critchlow} & $\mathrm{O}(n^2)$&$\mathrm{O}(n^2)$ &$\mathrm{O}(n^2)$&no\\
   2011& Bansal~\ietal~\cite{Bansal2011} & $\mathrm{O}(n^2)$ &$\mathrm{O}(n^2)$&$\mathrm{O}(n^2)$&yes\\
    2013& Sand~\ietal~\cite{bmc13} & $\mathrm{O}(n \log^{2} n)$ &$\mathrm{O}(n \log^2 n$)&$\mathrm{O}(n)$&no\\
    2013&Brodal~\ietal~\cite{TripletQuartetSoda13} & $\mathrm{O}(n \log n$) &$\mathrm{O}(n \log n$)&$\mathrm{O}(n \log n)$&yes\\
    2015&Jansson and Rajaby~\cite{Jansson2015} & $\mathrm{O}(n \log^{3} n)$ &$\mathrm{O}(n \log^{3} n)$ &$\mathrm{O}(n \log n)$&yes\\
    2017&new & $\mathrm{O}(n \log n)$ & $\mathrm{O}(\frac{n}{B} \log_{2} \frac{n}{M})$ &$\mathrm{O}(n)$&yes\\\hline
  \end{tabular}
  \label{tab:table1}
\end{table}

\subparagraph*{Related Work.}

The triplet distance was first suggested as a method of comparing the shapes of trees by Dobson in 1975 \cite{Dobson}. The first non-trivial algorithmic result dates back to~1996, when \mbox{Critchlow \ietal~\cite{Critchlow}} proposed an $\mathrm{O}(n^2)$ algorithm that however works only for binary trees. Bansal~\ietal~\cite{Bansal2011} introduced an~$\mathrm{O}(n^2)$ algorithm that works for general (binary and non-binary) trees. Both of these algorithms use $\mathrm{O}(n^2)$ space. Sand~\ietal~\cite{bmc13} introduced a new $\mathrm{O}(n^2)$ algorithm using only~$\mathrm{O}(n)$ space for the case of binary trees, that they showed how to optimize to reduce the time to~$\mathrm{O}(n\log^2 n)$. This algorithm was also implemented and shown to be the most efficient in practice. Soon after, Brodal~\ietal~\cite{TripletQuartetSoda13} managed to extend the~$\mathrm{O}(n\log^2 n)$ algorithm to work for general trees, and at the same time brought the time down to $\mathrm{O}(n\log n)$ but now with the space increased to~$\mathrm{O}(n \log n)$. The space for binary trees was still $\mathrm{O}(n)$. The algorithms from \cite{bmc13} and \cite{TripletQuartetSoda13} were implemented and added to the library~tqDist~\cite{bioinformatics14}. Interestingly, it was shown in \cite{alenex14} that for binary trees the~$\mathrm{O}(n\log^2 n)$ algorithm had a better practical performance than the~$\mathrm{O}(n\log n)$ algorithm. \mbox{Jansson and Rajaby~\cite{Jansson2015,JR_17}} showed that an even slower theoretically algorithm requiring worst case $\mathrm{O}(n\log^3 n)$ time and~$\mathrm{O}(n \log n)$ space could give the best practical performance, both for binary and non-binary trees. A survey of previous results until~2013~can~be~found~in~\cite{biology13}.

\subparagraph*{Contribution.}

The common main bottleneck with all previous approaches is that the data structures used rely intensively on $\Omega(n \log n)$ random memory accesses. This means that all algorithms are penalized by cache performance and thus do not scale to external memory. 
We address this limitation by proposing new algorithms for computing the triplet distance on binary and non-binary trees, that match the previous best $\mathrm{O}(n\log n)$ time and $\mathrm{O}(n)$ space bounds in the RAM model,  but for the first time also scale to external memory. More specifically, in the cache oblivious model, the total number of I/Os required is $\mathrm{O}(\frac{n}{B} \log_{2} \frac{n}{M})$. The basic idea is to essentially replace the dependency of random access to data structures by scanning contracted versions of the input trees. A careful implementation of the algorithms is shown to achieve the best performance in practice, thus essentially documenting that the theoretical results carry over to practice.

\subparagraph*{Outline of the Paper.}
In Section \ref{sec:prevApproaches} we provide an overview of previous approaches. In Section \ref{sec:binary} we describe the new algorithm for the case where $T_{1}$ and $T_{2}$ are binary trees. In Section \ref{sec:general} we extend the algorithm to also work for general trees. In Section \ref{sec:implementation} we provide some implementation details. Section \ref{sec:experiments} contains our experimental evaluation. Appendix \ref{appendix:experimentFigures} contains more experimental results. Finally, in Section~\ref{sec:conclusion} we provide our concluding remarks.

\section{Previous Approaches}
\label{sec:prevApproaches}
A naive approach would enumerate over all $\binom{n}{3}$ sets of three leaf labels and find for each set whether the induced topologies in $T_{1}$ and $T_{2}$ differ or not, giving an $\mathrm{O}(n^3)$ algorithm. This algorithm does not exploit the fact that the triplets are not completely independent. For example, the triplets $xy|z$ and~$yx|u$ share the leaves $x$ and $y$ and the fact that the lowest common ancestor of $x$ and $y$ is at a lower depth than the lowest common ancestor of $z$ with either $x$ or $y$ and the lowest common ancestor of $u$ with either $x$ or $y$. Dependencies like this can be exploited to count the number of shared triplets faster.

Critchlow \ietal~\cite{Critchlow} exploit the depth of the leaves' ancestors to achieve the first improvement over the naive approach. Bansal \ietal~\cite{Bansal2011} exploit the shared leaves between subtrees and reduce the problem to computing the intersection size (number of shared leaves) of all pairs of subtrees, one from $T_{1}$ and one from $T_{2}$, which can be solved with dynamic programming.

\subparagraph*{The $\mathrm{O}(n^2)$ Algorithm for Binary Trees in \cite{bmc13}.}

The algorithm for binary trees in \cite{bmc13} is the basis for all subsequent improvements \cite{bmc13, TripletQuartetSoda13, Jansson2015}, including ours as well, so we will describe it in more detail here. The dependency that was exploited is the same as in \cite{Bansal2011} but the procedure for counting the shared triplets is completely different. More specifically, each triplet in $T_{1}$ and $T_{2}$, defined by three leaf labels $i$, $j$, and $k$, is implicitly \emph{anchored} in the lowest common ancestor of $i$, $j$, and $k$. For two nodes $u$ in~$T_{1}$ and $v$ in $T_{2}$, let $s(u)$ and $s(v)$ be the set of triplets that are anchored in $u$ and~$v$ respectively. For the number of shared triplets~$S(T_{1},T_{2})$ we then have: $$S(T_{1},T_{2}) = \sum_{u \in T_{1}}\sum_{v \in T_{2}} \lvert {s(u) \cap s(v)} \rvert \;.$$ 

For the algorithm to be $\mathrm{O}(n^2)$ the value  $\lvert {s(u) \cap s(v)} \rvert$ must be computed in $\mathrm{O}(1)$ time. This is achieved by a leaf colouring procedure as follows: Fix an internal node $u$ in $T_{1}$ and color the leaves in the left subtree of $u$ \textit{red}, the leaves in the right subtree of $u$  \textit{blue}, let every other leaf have no color and then transfer this coloring to the leaves in $T_{2}$, i.e., identically labeled leaves get the same color. The triplets anchored at $u$ are exactly the triplets $xy|z$ where~$x$,~$y$ are blue and $z$ is red, or~$x$,~$y$ are red and $z$ is blue. To compute~$\lvert {s(u) \cap s(v)} \rvert$ we do as follows: let $l$ and $r$ be the left and right children of $v$, and let $w_{\texttt{red}}$ and~$w_{\texttt{blue}}$ be the number of red and blue leaves in a subtree rooted at a node~$w$ in $T_{2}$. We then have:
\begin{align}
\centering
\label{eq:eq1}
\lvert {s(u) \cap s(v)} \rvert = \binom{l_{\texttt{red}}}{2}r_{\texttt{blue}} + \binom{l_{\texttt{blue}}}{2}r_{\texttt{red}} + \binom{r_{\texttt{red}}}{2}l_{\texttt{blue}} + \binom{r_{\texttt{blue}}}{2}l_{\texttt{red}}\;.
\end{align}
\subparagraph*{Subquadratic Algorithms.}
To reduce the time, Sand~\ietal~\cite{bmc13} applied the \textit{smaller half trick}, which specifies a depth first order to visit the nodes $u$ of $T_{1}$, so that each leaf in~$T_{1}$  changes color at most $\mathrm{O}(\log n)$ times. To count shared triplets efficiently without scanning~$T_{2}$ completely for each node $u$ in $T_{1}$, the tree~$T_{2}$ is stored in a data structure denoted a \emph{hierarchical decomposition tree} (\emph{HDT}). This HDT of $T_{2}$ maintains for the current visited node $u$ in $T_{1}$, according to (\ref{eq:eq1}) the sum $\sum_{v \in T_{2}} \lvert {s(u) \cap s(v) \rvert}$, so that each leaf color change in $T_{1}$ can be updated efficiently in $T_{2}$. In~\cite{bmc13} the HDT is a binary tree of height $\mathrm{O}(\log n)$ and every update can be done by a leaf to root path traversal in the HDT, which in total gives $\mathrm{O}(n \log^2 n)$ time. In \cite{TripletQuartetSoda13} the HDT is generalized to also handle non-binary trees, each query operates the same, and now due to a contraction scheme of the HDT the total time is reduced to $\mathrm{O}(n \log n)$. Finally, in~\cite{Jansson2015} as an HDT the so called \emph{heavy-light tree decomposition} is used. Note that the only difference between all~$\mathrm{O}(n \operatorname{polylog} n)$ results that are available right~now~is~the~type~of~HDT~used.

In terms of external memory efficiency, every $\mathrm{O}(n \operatorname{polylog} n)$ algorithm performs $\mathrm{\Theta}(n \log n)$ updates to an HDT data structure, which means that for sufficiently large input trees every algorithm requires $\mathrm{\Omega}(n \log n)$ I/Os.

\section{The New Algorithm for Binary Trees}
\label{sec:binary}

In this section, we provide a cache oblivious algorithm that for two binary trees $T_{1}$ and $T_{2}$, built on the same leaf label set of size $n$, computes $D(T_{1}, T_{2})$ using $\mathrm{O}(n \log n)$ time and $\mathrm{O}(n)$ space in the~RAM model, and $\mathrm{O}(\frac{n}{B} \log_{2} \frac{n}{M})$ I/Os in the cache oblivious model.

\subparagraph{Overview.}

We use the $\mathrm{O}(n^2)$ algorithm from Section \ref{sec:prevApproaches} as a basis. The main difference between this algorithm and our new algorithm is in the order that we visit the nodes of~$T_{1}$, and how we process~$T_{2}$ when we count. We propose a new order of visiting the nodes of~$T_{1}$, which is found by applying a hierarchical decomposition on $T_{1}$. Every component in this decomposition corresponds to a connected part of~$T_{1}$ and a contracted version of $T_{2}$. In simple terms, if $\Lambda$ is the set of leaves in a component of $T_{1}$, the contracted version of~$T_{2}$ is a binary tree on $\Lambda$ that preserves the topologies induced by $\Lambda$ in $T_{2}$ and has size $\mathrm{O}(|\Lambda|)$. To count shared triplets, every component of $T_{1}$ has a representative node $u$ that we use to scan the corresponding contracted version of $T_{2}$ in order to find  $\sum_{v \in T_{2}} \lvert {s(u) \cap s(v) \rvert}$. Unlike previous algorithms, we do not store $T_{2}$ in a data structure. We process $T_{2}$ by contracting and counting, both of which can be done by scanning. At the same time, even though we apply a hierarchical decomposition on $T_{1}$, the only reason why we do so, is so we can find the order in which to visit the nodes of $T_{1}$. This means that we do not need to store~$T_{1}$ in a data structure either. Hence, we completely remove the need of data structures (and thereby random memory accesses), and scanning becomes the basic primitive in the algorithm. To make our algorithm I/O efficient, all that remains to be done is to use a proper layout to store the contracted trees in memory, so that every time we scan a tree of size $s$ we~spend~$\mathrm{O}(s/B)$~I/Os.

\subsection{Modified Centroid Decomposition}
\label{sec:mcd}

For a given binary tree $T$ let $|T|$ denote the number of nodes in $T$ (internal nodes and leaves). For a node $u$ in $T$ let $l$ and $r$ be the left and right children of $u$, and $p$ the parent of~$u$. Removing~$u$ from~$T$ partitions $T$ into three (possibly empty) \emph{connected components}~$T_{l}$,~$T_{r}$, and $T_{p}$ containing~$l$, $r$, and~$p$, respectively. A \emph{centroid} is a node~$u$ in $T$ such that $\max\{|T_{l}|, |T_{r}|, |T_{p}|\} \leq |T|/2$. A centroid always exists and can be found by starting from the root of $T$ and iteratively visiting the child with a largest subtree, eventually we will reach a centroid. Finding the size of every subtree and identifying~$u$ takes~$O(|T|)$ time in the RAM model. By recursively finding centroids in each of the three components, we in the end get a ternary tree of centroids, which is called the \emph{centroid~decomposition} of~$T$, denoted~$\mathit{CD}(T)$. We can generate a level of~$\mathit{CD}(T)$ in $\mathrm{O}(|T|)$ time, given the decomposition of~$T$ into components by the previous level. Since $\mathit{CD}(T)$ has at most $1+\log_{2}(|T|)$ levels, the total time required to build $\mathit{CD}(T)$ is~$\mathrm{O}(|T|\log |T|)$, thus we get Lemma \ref{lemma1}.

\begin{lemma}
For any given binary tree $T$ with $n$ leaves, there exists an algorithm that builds~$\mathit{CD}(T)$ using $\mathrm{O}(n \log n)$ time and $\mathrm{O}(n)$ space in the RAM model.
\label{lemma1}
\end{lemma}

\begin{figure}
\captionsetup[subfigure]{justification=centering}
    \centering
    \begin{subfigure}{0.50\textwidth}
        \centering
        \includegraphics[width=0.95\textwidth]{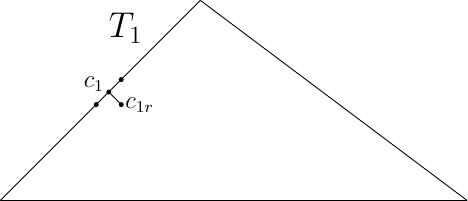}
       \caption{Picking the first centroid $c_{1}$ of $T_{1}$.\\\hspace{\textwidth}}
    \end{subfigure}%
    \begin{subfigure}{0.50\textwidth}
        \centering
        \includegraphics[width=0.95\textwidth]{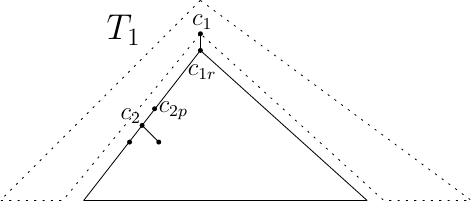}
        \caption{Recursing to component defined by $c_{1r}$ and picking the centroid $c_{2}$ of that component.}
    \end{subfigure}\\
     \begin{subfigure}{0.50\textwidth}
        \centering
        \includegraphics[width=0.95\textwidth]{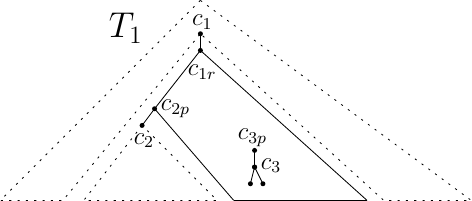}
        \caption{Recursing to component defined by $c_{2p}$ and picking the centroid $c_{3}$ of that component.}
    \end{subfigure}%
     \begin{subfigure}{0.50\textwidth}
        \centering
        \includegraphics[width=0.95\textwidth]{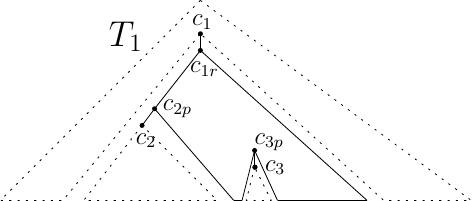}
        \caption{Recursing to component defined by $c_{3p}$.\\\hspace{\textwidth}}
     	\label{fig:CDExampleD}
    \end{subfigure}\\
 
    \caption{Generating a component in $\mathit{CD}(T_{1})$ that has two edges from below. The black polygon is the component.}
    \label{fig:CDExample}
\end{figure}

A component in a centroid decomposition $\mathit{CD}(T)$, might have several edges crossing its boundaries (connecting nodes inside and outside the component). An example of creating a component that has two edges from below can be found in Figure \ref{fig:CDExample}. It is trivial to see that by following the same pattern of generating components as depicted in Figure \ref{fig:CDExampleD}, $\mathit{CD}(T)$ can have a component with an arbitrary number of edges from below.
The below \emph{modified centroid decomposition}, denoted~$\mathit{MCD}(T)$, generates components with at most two edges crossing the boundary, one going towards the root and one down to exactly one subtree.

An $\mathit{MCD}(T)$ is built as follows: The first component is defined by $T$, just like in $\mathit{CD}(T)$. To find recursively the rest of the components, if a component~$C$ has no edge from below, we select the centroid $c$ of $C$ as a splitting node, just like when building $\mathit{CD}(T)$. Otherwise, let~$(x,y)$ be the edge that crosses the boundary from below, where $x$ is in $C$, $y$ is a child of $x$ and $y$ is not in $C$, and let $c$ be the centroid of~$C$ (possibly~$x=c$). As a splitting node choose the lowest common ancestor of~$x$ and $c$ (possibly~$x$ or $c$). By induction every component has at most one edge from below and one edge from above. A useful property of~$\mathit{MCD}(T)$ is captured by the following lemma:

\begin{lemma}
For any given binary tree $T$, we have $h(\mathit{MCD}(T))$ $\leq$ $2 + 2 \log_{2}|T|$, where $h(\mathit{MCD}(T))$ denotes the height of $\mathit{MCD}(T)$.
\label{lemma2}
\end{lemma}

\begin{proof}
In~$\mathit{MCD}(T)$ if a component $C$ does not have an edge from below then the centroid of~$C$ is used as a splitting node, thus generating three components~$C_{l}$,~$C_{r}$, and $C_{p}$ such that~$|C_{l}| \leq~\frac{|C|}{2}$, $|C_{r}|\leq \frac{|C|}{2}$, and $|C_{p}|\leq \frac{|C|}{2}$. Otherwise,~$C$ has one edge $(x,y)$ from below, with~$x$ being the node that is part of $C$. Let $c$ be a centroid of $C$. We have to consider the following two cases: if $c$ happens to be the lowest common ancestor of~$c$ and $x$, then our algorithm will split $C$ according to the actual centroid, so we will have that~$|C_{l}| \leq~\frac{|C|}{2}$,~$|C_{r}|\leq~\frac{|C|}{2}$, and~$|C_{p}|\leq \frac{|C|}{2}$. Otherwise, the splitting node will produce components $C_{l}$, $C_{r}$, and $C_{p}$, where $C_{l}$ contains $x$ and $C_{r}$ contains $c$, i.e., we have \mbox{$|C_{l}| + |C_{p}| \leq \frac{|C|}{2}$} and $|C_{r}| \geq \frac{|C|}{2}$. From the first inequality, we have that $|C_{l}| \leq \frac{|C|}{2}$ and~$|C_{p}| \leq \frac{|C|}{2}$.  Notice that~$C_{r}$ is going to be a component corresponding to a complete subtree of $T$, so it will have no edges from below. This means that in the next recursion level when working with~$C_{r}$ the actual centroid of~$C_{r}$ will be chosen as a splitting node, thus in the following recursion level the three components produced from~$C_{r}$ will be such that their sizes are at most half the size of $C$. From the analysis given so far, it becomes clear that when we have a component of size $|C|$ with one edge from below, then we will need at most~2 levels in $\mathit{MCD}(T)$ before producing components all of which will have a guaranteed size of at most $\frac{|C|}{2}$. Hence, the statement~follows.
\end{proof}

Since every level of $\mathit{MCD}(T)$ can be constructed in $\mathrm{O}(|T|)$ time and we have \mbox{$|T| = 2n-1$}, we obtain the following:

\begin{theorem}
For any given binary tree $T$ with $n$ leaves, there exists an algorithm that constructs~$\mathit{MCD}(T)$ using $\mathrm{O}(n \log n)$ time and $\mathrm{O}(n)$ space in the~RAM model.
\label{theorem1}
\end{theorem}

\subsection{The Main Algorithm}
\label{sec:binaryMain}

There is a preprocessing step and a counting (of shared triplets between $T_{1}$ and $T_{2}$) step.

 In the preprocessing step, first we apply a depth first traversal on $T_{1}$ to make $T_{1}$ \emph{left-heavy}, by swapping children so that for every node $u$ in $T_{1}$ the left subtree is larger than the right subtree. Second, we change the leaf labels of~$T_{1}$, which can also be done by a depth first traversal of $T_{1}$, so that the leaves are numbered $1$ to $n$ from left to right. Both steps take~$\mathrm{O}(n)$ time in the RAM model. The second step is performed to simplify the process of transferring the leaf colors between $T_{1}$ and~$T_{2}$. The coloring of a subtree in $T_{1}$ will correspond to assigning the same color to a contiguous range of leaf labels. Determining the color of a leaf in~$T_{2}$ will then require one \texttt{if-statement} to find in what range (red or blue) its label belongs to. Finally, we build $\mathit{MCD}(T_{1})$ according to the description after Lemma~\ref{lemma1}.

In the counting step, we visit the nodes of $T_{1}$, given by the depth first traversal of the ternary tree $\mathit{MCD}(T_{1})$, where the children of every node $u$ in $\mathit{MCD}(T_{1})$ are visited from left to right. For every such node $u$ we compute $\sum_{v \in T_{2}} \lvert {s(u) \cap s(v)} \rvert$. We achieve this by processing $T_{2}$ in two phases, the \emph{contraction} phase and the \emph{counting} phase.

\subparagraph{Contraction Phase of \texorpdfstring{$T_{2}$}{Lg}.} Let $L(T_{2})$ denote the set of leaves in $T_{2}$ and~$\Lambda \subseteq L(T_{2})$. In the contraction phase, $T_{2}$ is compressed into a binary tree of size $\mathrm{O}(|\Lambda|)$ whose leaf set is $\Lambda$.  The contraction is done in a way so that all the topologies induced by $\Lambda$ in $T_{2}$ are preserved in the compressed binary tree. This is achieved by the following three sequential steps:
\begin{itemize}
\item Prune all leaves of $T_{2}$ that are not in $\Lambda$,
\item Repeatedly prune all internal nodes of $T_{2}$ with no children, and
\item Repeatedly contract unary internal nodes, i.e., nodes having exactly one~child.
\end{itemize}

Let $u$ be a node of $\mathit{MCD}(T_{1})$ and $C_{u}$ the corresponding component of $T_{1}$. For every such node~$u$ we have a contracted version of $T_{2}$, from now on referred to as $T_{2}(u)$, where $L(T_{2}(u))=L(C_{u})$. The goal is to augment $T_{2}(u)$ with counters (see counting phase below), so that we can find~$\sum_{v \in T_{2}} \lvert {s(u) \cap s(v) \rvert}$ by scanning $T_{2}(u)$ instead of $T_{2}$. One can imagine~$\mathit{MCD}(T_{1})$ as being a tree where each node $u$ is augmented with $T_{2}(u)$. To generate all contractions of~$T_{2}$ for level~$i$ of $\mathit{MCD}(T_{1})$, which correspond to a set of disjoint connected components in $T_{1}$, we can reuse the contractions of $T_{2}$ at level~$i-1$ in~$\mathit{MCD}(T_{1})$. This means that we can generate the contractions of level~$i$ in $\mathrm{O}(n)$ time, thus we can generate all contractions of $T_{2}$ in $\mathrm{O}(n \log n)$ time. Note that by explicitly storing all contractions, we will also need to use $\mathrm{O}(n \log n)$ space. For our problem, because we traverse~$\mathit{MCD}(T_{1})$ in a depth first manner, we only need to store the contractions corresponding to the stack of nodes of~$\mathit{MCD}(T_{1})$ that we have to remember during the traversal of $\mathit{MCD}(T_{1})$. Since the components at every second level of $\mathit{MCD}(T_{1})$ have at most half the size of the components two levels above, Lemma \ref{lemma3} states that the size of this stack is always $\mathrm{O}(n)$.

\begin{lemma}
Let $T_{1}$ and $T_{2}$ be two binary trees with $n$ leaves and $u_{1}, u_{2},\dots,u_{k}$ a root to leaf path of $\mathit{MCD}(T_{1})$. For the sizes of the corresponding contracted versions $T_{2}(u_{1})$, $T_{2}(u_{2})$,~$\dots$,~$T_{2}(u_{k})$ we have that~$\sum_{i=1}^{k} \lvert T_{2}(u_{i}) \rvert = \mathrm{O}(n)$.
\label{lemma3}
\end{lemma}

\begin{proof}
For the root $u_{1}$ we have $T_{2}(u_{1}) = T_{2}$, thus $|T_{2}(u_{1})| \leq 2n$. From the proof of Lemma~\ref{lemma2} we have that for every component of size $x$, we need at most two levels in $\mathit{MCD}(T_{1})$ before producing components all of which have a size of at most $\frac{x}{2}$. This means that $\sum_{i=1}^{k} \lvert T_{2}(u_{i}) \rvert \leq 2n + 2n + \frac{2n}{2} + \frac{2n}{2} + \frac{2n}{4} + \frac{2n}{4} + \cdots + \frac{2n}{2^{i}} + \frac{2n}{2^{i}} + \cdots = 2\sum_{j=0}^{\infty}\frac{2n}{2^{j}} \leq 8n = \mathrm{O}(n)$.
\end{proof}

\subparagraph{Counting Phase of \texorpdfstring{$T_{2}$}{Lg}.} In the counting phase, we find the value of $\sum_{v \in T_{2}} \lvert {s(u) \cap s(v) \rvert}$ by scanning $T_{2}(u)$ instead of $T_{2}$. This makes the total time of the algorithm in the RAM model~$\mathrm{O}(n \log n)$, with the space being $\mathrm{O}(n)$ because of Lemma \ref{lemma3}. We consider the following two~cases:

\begin{itemize}[listparindent=\parindent]

\item \textit{$C_{u}$ has no edges from below.}

In this case $C_{u}$ corresponds to a complete subtree of $T_{1}$. We act exactly like in the~$\mathrm{O}(n^2)$ algorithm (Section \ref{sec:prevApproaches}) but now instead of scanning $T_{2}$ we scan $T_{2}(u)$.

\item \textit{$C_{u}$ has one edge from below.} 

\begin{figure}[ht]
\captionsetup[subfigure]{justification=centering}
    \centering
    \begin{subfigure}{0.33\textwidth}
        \centering
        \includegraphics[width=0.95\textwidth]{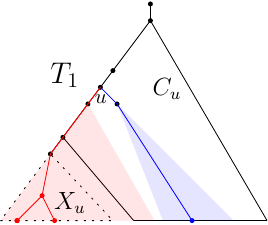}
    \end{subfigure}%
    \begin{subfigure}{0.33\textwidth}
        \centering
        \includegraphics[width=0.95\textwidth]{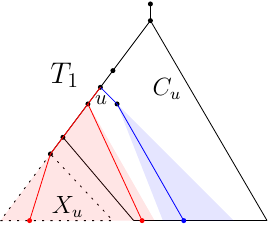}
    \end{subfigure}%
     \begin{subfigure}{0.33\textwidth}
        \centering
        \includegraphics[width=0.95\textwidth]{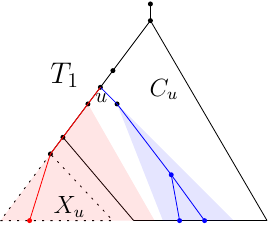}
    \end{subfigure}
    \caption{$\mathit{MCD}(T_{1}$): Triplets (red and blue) that can be anchored in $u$ with the leaves not being in the component~$C_{u}$.}
    \label{fig:tripletsWithXu}
\end{figure}

\begin{figure}[ht]
\captionsetup[subfigure]{justification=centering}
    \centering
    \begin{subfigure}{0.50\textwidth}
        \centering
        \includegraphics[width=1\textwidth]{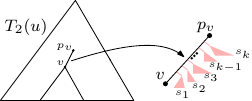}
    \end{subfigure}%
    \begin{subfigure}{0.50\textwidth}
        \centering
        \includegraphics[width=1\textwidth]{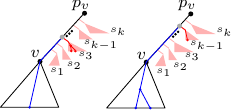}
    \end{subfigure}%
    \caption{Contracted subtrees on edges in $T_{2}(u)$ and shared triplets rooted at contracted nodes.}
    \label{fig:T2uEdges}
\end{figure}

In this case $C_{u}$ does not correspond to a complete subtree of $T_{1}$, since the edge from below~$C_{u}$, points to a subtree $X_{u}$, that is located outside of $C_{u}$ (see Figure \ref{fig:tripletsWithXu}). Note that because in the preprocessing step $T_{1}$ was made to be left-heavy, $X_{u}$ is always rooted at a node on the leftmost path from $u$. The leaves in $X_{u}$ are important because they can be used to form triplets that are anchored in $u$. Acting in the exact same manner as in the previous case is not sufficient because we need to count these triplets as well.

To address this problem, every edge $(p_{v}, v)$ in $T_{2}(u)$ between a node~$v$ and its parent~$p_{v}$, is augmented with some counters about the leaves from~$X_{u}$ that were contracted away in~$T_{2}$. If~$v$ is the root of $T_{2}(u)$, we add an extra edge to store this information. For every such edge~$(p_{v},v)$, let $s_{1},s_{2},\dots,s_{k}$ be the contracted subtrees rooted at the edge (see Figure~\ref{fig:T2uEdges}). Every such subtree contains either leaves with no color or leaves from $X_{u}$ that have the color red (the color cannot be blue because $T_{1}$ was made to be left-heavy). For every node $v$ in $T_{2}(u)$ the counters that we have are the following:

\begin{itemize}
\item $v_{\texttt{red}}$: total number of red leaves in the subtree of $v$ (including those coming from~$X_{u}$).
\item $v_{\texttt{blue}}$: total number of blue leaves in the subtree of $v$.
\item $v_{ts}$: total number of red leaves in $s_{1},s_{2},\dots,s_{k}$.
\item $v_{ps}$: total number of pairs of red leaves in $s_{1},s_{2},\dots,s_{k}$ such that each pair comes from the same contracted subtree, i.e., $\sum_{i=1}^{k} \binom{r_{i}}{2}$ where $r_{i}$ is the number of red leaves in~$s_{i}$.
\end{itemize}

The number of shared triplets that are anchored in a non-contracted node $v$ of $T_{2}(v)$ can be found like in the $\mathrm{O}(n^2)$ algorithm using the counters $v_{\texttt{red}}$ and $v_{\texttt{blue}}$ in (\ref{eq:eq1}). As for the number of shared triplets that are anchored in a contracted node on edge $(p_{v}, v)$, this value is exactly~$\binom{v.\texttt{blue}}{2} \cdot v_{ts} + v_{\texttt{blue}} \cdot v_{ps}$.

\end{itemize}

\subsection{Scaling to External Memory.} 
\label{sec:iobinary}

We now describe how to make the algorithm scale to external memory. The tree $T_{1}$ is stored in an array of size $2n-1$ by following a preorder layout, i.e., if a node $w$ of $T_{1}$ is stored in position $p$, the left child of~$w$ is stored in position~$p+1$ and if $x$ is the size of the left subtree of~$w$, the right child of~$w$ is stored in position $p + x + 1$. The components of $T_{1}$ are connected parts of~$T_{1}$, so they can be identified in $T_{1}$ without having to make a unique copy for each one of them. For $T_{2}$ and its contractions, we use the proof of Lemma~\ref{lemma3} to initialize a large enough array that can fit~$T_{2}$ and every contraction of $T_{2}$ that we need to remember while traversing $\mathit{MCD}(T_{1})$. This array is used as a stack that we use to push and pop the contractions of $T_{2}$. The tree $T_{2}$ and its contractions are stored in memory following a post order layout, i.e., if a node $w$ is stored in position $p$ and $y$ is the size of the right subtree of $w$, the left child of $w$ is stored in position~$p-y-1$ and the right child of~$w$ is stored in position $p-1$.

In the preprocessing step, $T_{1}$ can be made left-heavy with two depth first traversals. The first traversal computes for every node $u$ in $T_{1}$ the size of the subtree rooted at $u$. The second traversal starts from the root of $T_{1}$, recursively visits the children by first visiting a largest child, and prints all nodes visited along the way to an output array. This output array will at the end of the traversal contain the left-heavy version of $T_{1}$ in a preorder layout. From the following Lemma \ref{lemma:ioproof1} we have that both the first and second depth first traversal of~$T_{1}$ require~$\mathrm{O}(n/B)$ I/Os in the cache oblivious model, i.e., making $T_{1}$ left-heavy requires~$\mathrm{O}(n/B)$~I/Os in the cache oblivious model.

In Lemma \ref{lemma:ioproof1} we consider the I/Os required to apply a depth first traversal on a binary tree $T$ that is stored in memory following a local layout, i.e., the nodes of every subtree of~$T$ are stored consecutively in memory and every node has $\mathrm{O}(1)$ occurrences in memory. From here on, when we refer to an edge $(u,v)$, we imply that $u$ is the parent of $v$ in~$T$. During a depth first traversal of~$T$, an edge $(u,v)$ is either \emph{processed} to discover~$v$ or to backtrack from $v$ to $u$. In any case, w.l.o.g.\ we assume that when an edge is processed, both $u$ and $v$ are visited, i.e., both $u$ and $v$ are accessed in memory.

\begin{lemma}
\label{lemma:ioproof1}
Let~$T$ be a binary tree with~$n$ leaves that is stored in an array following a local layout,~i.e., the nodes of every subtree of~$T$ are stored consecutively in memory and every node has $\mathrm{O}(1)$ occurrences in memory. Any depth first traversal that starts from the root of~$T$, and in which for every internal node~$u$ in~$T$ the children of~$u$ are discovered in any order, requires~$\mathrm{O}(n/B)$ I/Os in the cache oblivious~model.\end{lemma}

\begin{proof}
\begin{figure}[ht]
    \centering
        \includegraphics{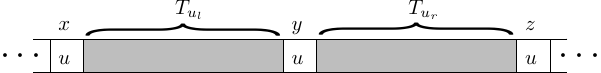}
   \caption{Position of a node $u$ in memory with respect to the two children subtrees of $u$.}
    \label{fig:1}
\end{figure}

For a node $u$ in $T$, let $T_u$ denote the set of nodes in the subtree defined by $u$. From here on, $T_u$ will be referred to as a \emph{subtree of $T$}. Let $u_{l}$ and~$u_{r}$ be the two children of $u$. In Figure \ref{fig:1} we illustrate the three possibilities for the position of~$u$ in memory with respect to~$T_{u_l}$ and~$T_{u_r}$. W.l.o.g.\ and to simplify the presentation of the proof, in our analysis we assume that~$u$ is stored in all these three possible positions, denoted~$x$,~$y$, and~$z$. This assumption is w.l.o.g.\ because in any local layout one or more of these positions is used, thus the number of I/Os is upper bounded by the number of I/Os incurred if we follow our assumption. This placement of $u$ in memory implies that when~$u$ is visited in a depth first traversal of $T$, all the three copies of~$u$ are accessed in memory. Note that according to the definition of a local layout,~$T_{u_l}$ and~$T_{u_r}$ can be interchanged in Figure~\ref{fig:1}. In the following, the~aim~is~to~bound~the~number~of~I/Os~implied.

\begin{figure}
\captionsetup[subfigure]{justification=centering}
    \centering
    \begin{subfigure}{0.6\textwidth}
        \centering
        \includegraphics[width=1\textwidth]{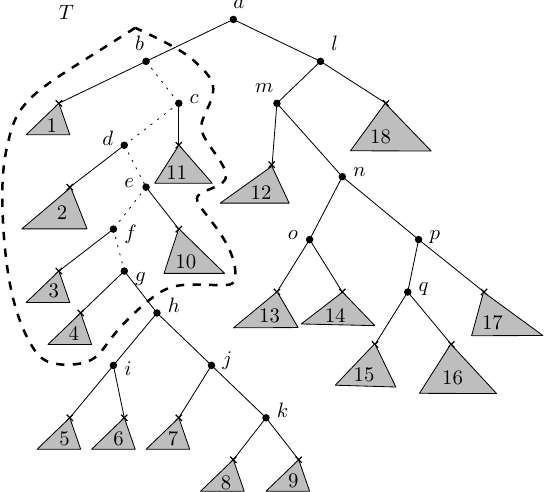}
        \caption{}
        \label{fig:2a}
    \end{subfigure}\hspace{10mm}
    \begin{subfigure}{0.25\textwidth}
        \centering
        \includegraphics[width=1\textwidth]{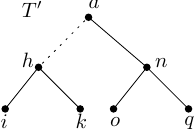}
        \caption{}
        \label{fig:2b}
    \end{subfigure}\\\par\bigskip\par\bigskip
     \begin{subfigure}{1\textwidth}
        \centering
        \includegraphics[width=1\textwidth]{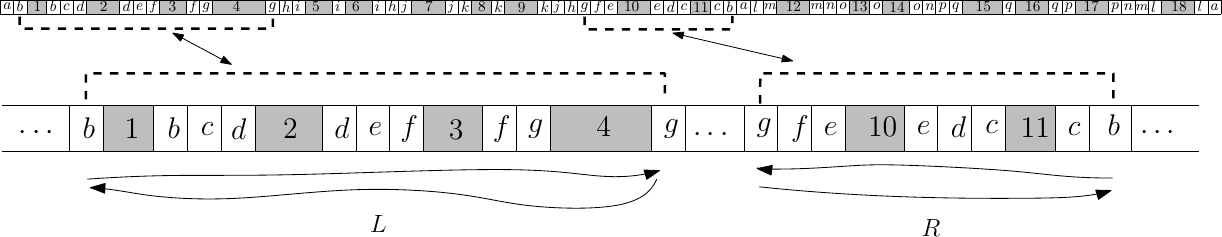}
        \caption{}
        \label{fig:2c}
    \end{subfigure}%
    \caption{(a)~A tree $T$. The gray subtrees are $B$-light subtrees and every node not in a~$B$-light subtree is a $B$-heavy node. ~(b)~The corresponding tree~$T'$ according to the proof of Lemma~\ref{lemma:ioproof1}.~(c)~How~$T$ is stored in memory, the two segments of memory (in dashed lines) that correspond to the edge~$(a,h)$ in $T'$ and how the nodes in~$P_{(a,h)}$ are visited (defined by the one directional lines) during a depth first traversal of $T$.}
    \label{fig:2}
\end{figure}

Define a node $u$ in $T$ to be \emph{$B$-light} if $3|T_u| \leq B-2$, otherwise the node is said to be~\mbox{\emph{$B$-heavy}}. Observe that the children of a $B$-light node are all~$B$-light. We consider the following disjoint sets of nodes from $T$:
\enlargethispage{\baselineskip} 
\begin{enumerate}
\item[$S_{1}$:] Every $B$-light node
\item[$S_{2}$:] Every $B$-heavy node with only $B$-light children
\item[$S_{3}$:] Every $B$-heavy node with two $B$-heavy children
\item[$S_{4}$:] Every $B$-heavy node with one $B$-heavy child and one $B$-light child.
\end{enumerate}

For a $B$-light node~$u$ in~$T$, let~$w$ be the first $B$-heavy node we reach in the path from~$u$ to the root of $T$. An I/O incurred by visiting the node $u$ in $T$ is charged to $w$. This node $w$ can be either in~$S_{2}$ or $S_{4}$. Let $w'$ be the child of~$w$ such that $T_{w'}$ contains $u$. Since~$3|T_{w'}| \leq B-2$, at most 1 I/O is sufficient to visit all nodes in $T_{w'}$. We say that $T_{w'}$ is a subtree that is~$B$-light. In Figure \ref{fig:2a} we have an example of a tree, where the gray subtrees denote~$B$-light~subtrees.

  We now argue that $|S_{2}| = \mathrm{O}(n/B)$ and $|S_{3}| = \mathrm{O}(n/B)$. Let $T'$ be the binary tree created by pruning every $B$-light node from $T$ and their incident edges, and subsequently contracting nodes with in-degree of 1 and out-degree of 1. An example for $T$ and the corresponding tree~$T'$ can be found in Figures~\ref{fig:2a} and \ref{fig:2b}. Let~$l_{1},l_{2},\dots,l_{k}$ be the leaves of~$T'$ and~$T_{l_{1}}, \dots, T_{l_{k}}$ the corresponding subtrees in $T$. Since all these subtrees are disjoint and for every~$1\leq i \leq k$ we have $|T_{l_{i}}| > \frac{B-2}{3} $, for the total number of leaves $x$ in $T'$ we have~$x \leq 3|T|/(B-2)$. Hence, we have $|S_{2}| = x = \mathrm{O}(n/B)$. By construction $T'$ is a binary tree, thus we have that~\mbox{$|T'| \leq 2x \leq 6|T|/(B-2) = \mathrm{O}(n/B)$}. Since the nodes in $S_{3}$ correspond to internal nodes in $T'$, we have $|S_{3}| = \mathrm{O}(n/B)$.

We now argue that the total number of I/Os incurred by the nodes in $S_{4}$ is $\mathrm{O}(n/B)$, thus proving the statement. Let $(u,v)$ be an edge in $T'$. This edge corresponds to a unique path, denoted~$P_{(u,v)}$ in~$T$ that contains every $B$-heavy node, except~$u$ and $v$, that is in the path from~$u$ to~$v$. For example the edge~$(a,h)$ in Figure~\ref{fig:2b} corresponds to~\mbox{$P_{(a,h)} = b \rightarrow c \rightarrow d \rightarrow e \rightarrow f \rightarrow g$}. Let $C_{(u,v)}$ contain all $B$-light and $B$-heavy nodes, except $u$ and $v$, rooted at the path from $u$ to $v$ in $T$. By the local layout followed to store~$T$ in memory, the nodes in $C_{(u,v)}$ are stored in two segments of memory (e.g., see Figure~\ref{fig:2c}). Let~$L$ be the left segment and $R$ the right segment. During a depth first traversal of~$T$, visiting all nodes in $P_{(u,v)}$ corresponds to visiting~$L$ from left to right and then from right to left, and visiting~$R$ from right to left and then from left to right. Since each of the~$B$-light subtrees in~$L$ and $R$ use at most~$B-2$ positions in memory, by accessing all three copies of a node~$w$ in~$P_{(u,v)}$ every time $w$ is visited in a depth first traversal of $T$, we guarantee that the corresponding $B$-light subtree rooted at~$w$ is in cache, i.e., it can be accessed in memory for free. Hence, the total number of I/Os that are sufficient to pay for traversing all nodes in~$C_{(u,v)}$ is~$4+\lceil 3|C_{(u,v)}|/B \rceil$, where the~$+4$ comes from the~4~I/Os we need to pay (in the worst case) to visit the first and last node of $L$ and~$R$. In total, the total number of~I/Os we need to spend for all paths of $T$ that correspond to edges of~$T'$ is~$\sum_{(u,v) \in T'}^{}(4+\lceil 3|C_{(u,v)}|/B \rceil) = \mathrm{O}(n/B)$. Together with the fact that for every node of~$T$ that corresponds to a node of~$T'$ we only spend $\mathrm{O}(1)$ I/Os and there are $\mathrm{O}(n/B)$ such nodes, the statement~follows.
\end{proof}

Changing the labels of $T_{1}$ can be done in $\mathrm{O}(\frac{n}{B}\log_{2} \frac{n}{M})$ I/Os with a cache oblivious sorting routine, e.g., with merge sort. Overall, the preprocessing step requires $\mathrm{O}(\frac{n}{B}\log_{2} \frac{n}{M})$ I/Os. 

When building $\mathit{MCD}(T_{1})$, by scanning the leftmost path that starts from the root of a component~$C_{u}$, we can find the splitting node of $C_{u}$ in {$1 + \lceil |C_{u}|/B \rceil$ I/Os. In $T_{2}(u)$ we spend $1 + \Theta(\lceil |T_{u}|/B \rceil)$ I/Os for the contraction and counting phase. Since $|T_{2}(u)| = \Theta(|C_{u}|)$, overall for a given~$(C_{u},T_{2}(u))$ pair the algorithm requires~$2 + \Theta(\lceil |C_{u}|/B \rceil)$ I/Os. However, after $\mathrm{O}(\log_{2} \frac{n}{M})$ levels in $\mathit{MCD}(T_{1})$, any~$(C_{u},T_{2}(u))$ pair will fit in a cache of size $M$. All such pairs together incur $\mathrm{O}(n/B)$ I/Os. By using a stack to store the contractions of $T_{2}$, the remaining pairs incur $\mathrm{O}(\frac{n}{B}\log_{2} \frac{n}{M})$ I/Os. Overall, the algorithm requires~$\mathrm{O}(\frac{n}{B}\log_{2} \frac{n}{M})$~I/Os in the cache oblivious model.

\section{The New Algorithm for General Trees}
\label{sec:general}

Unlike a binary tree, a general tree can have internal nodes with an arbitrary number of children. By anchoring the triplets of $T_{1}$ and $T_{2}$ in edges instead of nodes, we show that with only four colors we can count all the shared triplets between the two trees. We start by describing a new~$\mathrm{O}(n^2)$ algorithm for general trees. We then show how we can use the same ideas presented in the previous section to extend the $\mathrm{O}(n^2)$ algorithm and reduce~the~time~to~$\mathrm{O}(n \log n)$.

\subsection{Quadratic Algorithm}
\label{sec:quadraticGeneral}

For a given tree $T$, let $t$ be a triplet with leaves $i$, $j$, and $k$ that is either a resolved triplet~$ij|k$  or an unresolved triplet $ijk$, where $i$ is to the left of $j$ and for the triplet $ijk$, $k$ is also to the right of $j$. Let~$w$ be the lowest common ancestor of $i$ and $j$ and $(w,c)$ the edge from $w$ to the child $c$ whose subtree contains $j$. We anchor $t$ in edge $(w,c)$. Let $s'(w,c)$ be the set containing all triplets anchored in edge $(w,c)$. For the number of shared triplets $S(T_{1},T_{2})$ we have:
$$S(T_{1},T_{2}) = \sum_{(u,c) \in T_{1}}\sum_{(v,c') \in T_{2}} \lvert {s'(u,c) \cap s'(v,c')} \rvert \;.$$ 
For the efficient computation of $S(T_{1},T_{2})$ we use the following coloring procedure: Fix a node $u$ in~$T_{1}$ and a child $c$. Color the leaves of every child subtree of $u$ to the left of $c$ red, the leaves of the child subtree defined by~$c$ blue, the leaves of every child subtree to the right of $c$ green and give the color black to every other leaf of $T_{1}$. We then transfer this coloring to the leaves of~$T_{2}$. For the resolved triplet~$ij|k$, $i$ corresponds to the red color, $j$ corresponds to the blue color and $k$ corresponds to the black color. For the unresolved triplet~$ijk$, $i$ corresponds to the red color, $j$ corresponds to the blue color and~$k$ corresponds to the green color.

Suppose that the node $v$ in $T_{2}$ has $k$ children. We are going to compute all shared triplets that are anchored in the $k$ children edges of $v$ in $\mathrm{O}(k)$ time. This will give an $\mathrm{O}(n^2)$ total running time, because for every edge in $T_{1}$ we spend $\mathrm{O}(n)$ time in $T_{2}$ and there are $\mathrm{O}(n)$ edges in $T_{1}$. In $v$ we have the following counters:

\begin{itemize}
\item $v_{\texttt{red}}$: total number of red leaves in the subtree of $v$.
\item $v_{\texttt{blue}}$: total number of blue leaves in the subtree of $v$.
\item $v_{\texttt{green}}$: total number of green leaves in the subtree of $v$.
\item $\overline{v}_{\texttt{black}}$: total number of black leaves not in the subtree of $v$.
\end{itemize}

While scanning the $k$ children edges of $v$ from left to right, for the child $c'$ that is the~$m$-th child of $v$, we also maintain the following:

\begin{itemize}
\item $a_{\texttt{red}}$: total number red leaves from the first $m-1$ children subtrees.
\item $a_{\texttt{blue}}$: total number blue leaves from the first $m-1$ children subtrees.
\item $a_{\texttt{green}}$: total number of green leaves from the first $m-1$ children subtrees.
\item $p_{\texttt{red,green}}$: total number of pairs of leaves from the first $m-1$ children subtrees, where one is red, the other is green, and they both come from different subtrees.
\item $p_{\texttt{red,blue}}$ : total number of pairs of leaves from the first $m-1$ children subtrees, where one is red, the other is blue, and they both come from different subtrees.
\item $p_{\texttt{blue,green}}$ : total number of pairs of leaves from the first $m-1$ children subtrees, where one is blue, the other is green, and they both come from different subtrees.
\item $t_{\texttt{red,blue,green}}$: total number of leaf triples from the first $m-1$ children subtrees, where one is red, one is blue and one is green, and all three leaves come from different subtrees.
\end{itemize}
Before scanning the children edges of $v$, every variable is initialized to 0. Then for the child~$c'$ every variable is updated in $\mathrm{O}(1)$ time as follows:

\begin{itemize}
\item $a_{\texttt{red}} = a_{\texttt{red}} + c'_{\texttt{red}}$
\item $a_{\texttt{blue}} = a_{\texttt{blue}} + c'_{\texttt{blue}}$
\item $a_{\texttt{green}} = a_{\texttt{green}} + c'_{\texttt{green}}$
\item $p_{\texttt{red,green}} = p_{\texttt{red,green}} +  a_{\texttt{green}} \cdot c'_{\texttt{red}} +  a_{\texttt{red}} \cdot c'_{\texttt{green}}$
\item $p_{\texttt{red,blue}} = p_{\texttt{red,blue}} +  a_{\texttt{blue}} \cdot c'_{\texttt{red}} +  a_{\texttt{red}} \cdot c'_{\texttt{blue}}$
\item $p_{\texttt{blue,green}} = p_{\texttt{blue,green}} +  a_{\texttt{green}} \cdot c'_{\texttt{blue}} +  a_{\texttt{blue}} \cdot c'_{\texttt{green}}$
\item $t_{\texttt{red,blue,green}} = t_{\texttt{red,blue,green}} + p_{\texttt{red,green}} \cdot c'_{\texttt{blue}} + p_{\texttt{red,blue}} \cdot c'_{\texttt{green}} + p_{\texttt{blue,green}} \cdot c'_{\texttt{red}} $
\end{itemize}

After finishing scanning the $k$ children edges of $v$, we can compute the shared triplets that are anchored in every child edge of $v$ as follows: for the total number of shared resolved triplets, denoted~$\texttt{tot}_{\texttt{res}}$, we have that $\texttt{tot}_{\texttt{res}} = p_{\texttt{red,blue}} \cdot \overline{v}_{\texttt{black}}$ and for the total number of shared unresolved triplets, denoted~$\texttt{tot}_{\texttt{unres}}$, we have that $\texttt{tot}_{\texttt{unres}} = t_{\texttt{red,blue,green}}$. We are now ready to describe the $\mathrm{O}(n \log n)$ algorithm.

\subsection{Subquadratic Algorithm}
\label{sec:generalMain}
Similarly to the case of binary trees in Section \ref{sec:binary}, there is a preprocessing step and a counting step. The counting step is divided into two phases, the contraction and counting phase of $T_{2}$.

In the preprocessing step of the algorithm, we start by transforming $T_{1}$ into a binary tree, denoted~$b(T_{1})$. Let $w$ be a node of $T_{1}$ that has exactly $k$ children, where~$k>2$. The $k$ edges that connect $w$ to its children in $T_{1}$ are replaced in $b(T_{1})$ by a so called \emph{orange binary tree}. The root of this binary tree is $w$ and the leaves are the $k$ children of $w$ in $T_{1}$. Every internal node (except the root) and edge is colored orange, hence the given name. We assume that node~$w$ and its $k$ children in $T_{1}$, in $b(T_{1})$ have the color \emph{black}. This binary tree is built in a way so that every orange node is on the leftmost path that starts from~$w$, and its leftmost leaf stores the heaviest child of $w$ in $T_{1}$ (i.e., the child whose subtree is  the largest among all other children subtrees of $w$, when transformed in~$b(T_{1})$), thus making $b(T_{1})$ left-heavy. The order in which the other children of $w$ in $T_{1}$ are stored in the remaining leaves does not matter, however for the notation below to be mathematically correct, we assume that after constructing $b(T_{1})$, the left to right order of the children of $w$ in $T_{1}$ is implicitly updated, so that it matches the left to right order in which they appear in the leaves of the orange binary tree below $w$ in~$b(T_{1})$.

Let $u$ be a node in $b(T_{1})$ and $c$ its right child. By construction, $c$ must be a black node. If~$u$ is orange, then let $u_{\texttt{root}}$ be the root of the orange binary tree that $u$ is part of.  If $u$ is black, then let~$u_{\texttt{root}} = u$. Again by construction,~$u_{\texttt{root}}$ must be the parent of $c$ in $T_{1}$. For the edge~$(u,c)$ in $b(T_{1})$, we define~$s''(u,c)$ to be the set of triplets that are anchored in edge~$(u_{root},c)$ of~$T_{1}$, i.e., \mbox{$s''(u,c) = s'(u_{root},c)$}. Note that for an edge $(u',c')$ in $b(T_{1})$ connecting $u'$ with its left child $c'$, we have~$s''(u',c') = 0$.

 For the number of shared triplets we then have:
$$S(T_{1},T_{2}) = \sum_{(u,c) \in b(T_{1})}\sum_{(v,c') \in T_{2}} \lvert {s''(u,c) \cap s'(v,c')} \rvert \;.$$ 

We can capture all triplets in $T_{1}$ by coloring $b(T_{1})$ instead of $T_{1}$. For the nodes $u$ and~$c$ where~$c$ is the right child of $u$, the leaves of $b(T_{1})$ are colored according to edge $(u,c)$ as follows: the leaves in the left subtree of $u$ are colored red, the leaves in the right right subtree of $u$ are colored blue. If $u$ is an orange node, then the black leaves in the remaining subtrees of the orange binary tree that $u$ is part of are colored green. All other leaves of $b(T_{1})$ maintain their color black. 

The reason behind transforming $T_{1}$ into the binary tree $b(T_{1})$, is because now we can use exactly the same core ideas described in Section \ref{sec:binary}. The tree~$b(T_{1})$ is a binary tree, so we apply the same preprocessing step, except we do not need to make it left-heavy because by construction it already is. However, we change the labels of the leaves in $b(T_{1})$ and $T_{2}$, so that the leaves in $b(T_{1})$ are numbered~$1$ to $n$ from left to right. The order in which we visit the nodes of $b(T_{1})$ is determined by a depth first traversal of $\mathit{MCD}(b(T_{1}))$, where the children of every node $u$ in $\mathit{MCD}(b(T_{1}))$ are visited from left to right.

\begin{figure}
\captionsetup[subfigure]{justification=centering}
    \centering
    \begin{subfigure}[b]{0.50\textwidth}
        \centering
        \includegraphics[width=0.95\textwidth]{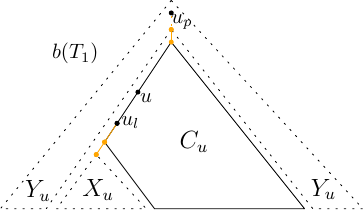}
    \end{subfigure}%
    \begin{subfigure}[b]{0.50\textwidth}
        \centering
        \includegraphics[width=0.95\textwidth]{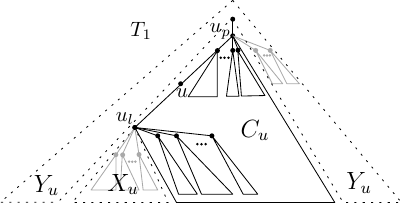}
    \end{subfigure}\\
    \caption{How a component in $b(T_{1})$ translates to a component in $T_{1}$.}
    \label{fig:componentsGeneral}
\end{figure}

In Figure \ref{fig:componentsGeneral} we see that a component $C_{u}$ of $b(T_{1})$ structurally looks like a component of~$T_{1}$ in the binary algorithm of Section \ref{sec:binary}. However, the edges crossing the boundary can now be orange edges as well, which in $T_{1}$ translates to more than one consecutive~subtrees.

Like in the case of binary input trees, while traversing  $\mathit{MCD}(b(T_{1}))$ we process $T_{2}$ in two phases, the contraction phase and the counting phase. The only difference after this point between the algorithm for binary trees and the algorithm for general trees, is in the counters that we have to maintain in the contracted versions of $T_{2}$. Otherwise, the same analysis from Section~\ref{sec:binary}~holds.

\subparagraph{Contraction Phase of \texorpdfstring{$T_{2}$}{Lg}.} The contraction of $T_{2}$ with respect to a set of leaves $\Lambda \subseteq L(T_{2})$, happens in the exact same way as described in Section~\ref{sec:binary}, i.e., we start by pruning all leaves of $T_{2}$ that are not in~$\Lambda$, then we prune all internal nodes of $T_{2}$ with no children, and finally, we contract the nodes that have exactly one child. 

Let $u$ be a node of $\mathit{MCD}(b(T_{1}))$ and $C_{u}$ the corresponding component of $b(T_{1})$. For every such node~$u$ we have a contracted version of $T_{2}$, denoted $T_{2}(u)$, where $L(T_{2}(u)) = L(C_{u})$. Like in the binary algorithm of Section \ref{sec:binary}, the goal is to augment $T_{2}(u)$ with counters, so that we can find~$\sum_{(v,c') \in T_{2}} \lvert {s''(u,c) \cap s'(v,c')} \rvert$ by scanning $T_{2}(u)$ instead of $T_{2}$. 

Because of the location where the triplets are anchored, every leaf that was contracted when constructing $T_{2}(u)$ must have a color and be stored in some way. The color of each such leaf depends on the type of the corresponding component that we have in $b(T_{1})$ and the splitting node that is used for that component. For example, in Figure \ref{fig:componentsGeneral} the contracted leaves from $X_{u}$ will have the red color because $b(T_{1})$ is left-heavy. The contracted leaves from the children subtrees of $u_{p}$ in $T_{1}$ can either have the color green or black. If~$u$ in $b(T_{1})$ happens to be orange and part of the orange binary tree that $u_{p}$ is the root of, then the color must be green, otherwise black. Finally, every leaf that is not in the subtree defined by $u_{p}$, and thus is in $Y_{u}$, must have the color black. The way we store this information is described in the counting phase~below.

\begin{figure}
    \centering
    \includegraphics[width=0.7\textwidth]{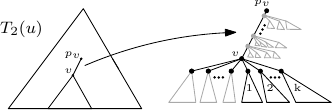}
    \caption{$T_{2}(u)$: Contracted children subtrees rooted at node $v$ and contracted subtrees rooted at contracted nodes (gray color) on the edge $(p_{v},v)$.}
    \label{fig:T2contractedNodesEdgesGeneral}
\end{figure}

\subparagraph{Counting Phase of \texorpdfstring{$T_{2}$}{Lg}.} In Figure \ref{fig:T2contractedNodesEdgesGeneral} we illustrate how a node $v$ in~$T_{2}(u)$ can look like. The contracted subtrees are illustrated with the dark gray color. Every such subtree contains some number of red, green, and black leaves. The counters that we maintain should be so that if~$v$ has~$k$ children in $T_{2}(u)$, then we can count all shared triplets that are anchored in every child edge~(including those of the contracted children subtrees) of $v$ in $\mathrm{O}(k)$ time. At the same time, in~$\mathrm{O}(1)$ time we should be able to count all shared triplets that are anchored in every child edge of every contracted node that lies on the edge $(p_{v}, v)$. Then, the time required by the counting phase becomes~$\mathrm{O}(|T_{2}(u)|)$, giving the same time bounds as in the binary algorithm of Section \ref{sec:binary}. In $v$ we have the following counters:

\begin{itemize}
\item $v_{\texttt{red}}$: total number of red leaves (including the contracted leaves) in the subtree of $v$.
\item $v_{\texttt{blue}}$: total number of blue leaves in the subtree of $v$.
\item $v_{\texttt{green}}$: total number of green leaves (including the contracted leaves) in the subtree of $v$.
\item $\overline{v}_{\texttt{black}}$: total number of black leaves (including the contracted leaves) not in the subtree of $v$.
\end{itemize}

We divide the rest of the counters into two categories. The first category corresponds to the leaves in the contracted children subtrees of $v$ and each counter is stored in a variable of the form~$v_{A.x}$. The second category corresponds to the leaves in the contracted subtrees on the edge~$(p_{v}, v)$, and each counter is stored in a variable of the form $v_{B.x}$. For the first category $A$ we have the following counters:
\begin{itemize}
\item $v_{A.\texttt{red}}$: total number of red leaves in the contracted children subtrees of~$v$.
\item $v_{A.\texttt{green}}$: total number of green leaves in the contracted children subtrees of $v$.
\item $v_{A.\texttt{black}}$: total number of black leaves in the contracted children subtrees of $v$.
\item $v_{A.\texttt{red,green}}$: total number of pairs of leaves where one is red, the other is green, and one leaf comes from one contracted child subtree of $v$ and the other leaf comes from a different contracted child subtree of $v$.
\end{itemize}

While scanning the $k$ children edges of $v$ from left to right, for the child $c'$ that is the~$m$-th child of $v$, we also maintain the following:

\begin{itemize}
\item $a_{\texttt{red}}$: total number of red leaves from the first $m-1$ children subtrees, including the contracted children subtrees.
\item $a_{\texttt{blue}}$: total number of blue leaves from the first $m-1$ children subtrees.
\item $a_{\texttt{green}}$: total number of green leaves from the first $m-1$ children subtrees, including the contracted children subtrees.
\item $p_{\texttt{red,green}}$: total number of pairs of leaves from the first $m-1$ children subtrees, including the contracted children subtrees, where one is red, the other is green, and they both come from different subtrees (one might be contracted and the other non-contracted).
\item $p_{\texttt{red,blue}}$ :  total number of pairs of leaves from the first $m-1$ children subtrees, including the contracted children subtrees, where one is red, the other is blue, and they both come from different subtrees (one might be contracted and the other non-contracted).
\item $p_{\texttt{blue,green}}$ :  total number of pairs of leaves from the first $m-1$ children subtrees, including the contracted children subtrees, where one is blue, the other is green, and they both come from different subtrees (one might be contracted and the other non-contracted).
\item $t_{\texttt{red,blue,green}}$: total number of leaf triples from the first $m-1$ children subtrees, including the contracted children subtrees, where one is red, one is blue and one is green, and all three leaves come from different subtrees (some might be contracted, some might be non-contracted).
\end{itemize}
Every variable is updated in $\mathrm{O}(1)$ time in exactly the same manner like in the $\mathrm{O}(n^2)$ algorithm of Section \ref{sec:quadraticGeneral}. The main difference is in the values of the variables before we begin scanning the children edges of $v$. Every variable is initialized as follows:
\enlargethispage{\baselineskip} 
\begin{itemize}
\item $a_{\texttt{red}} = v_{A.\texttt{red}}$
\item $a_{\texttt{blue}} = 0$
\item $a_{\texttt{green}} = v_{A.\texttt{green}}$
\item $p_{\texttt{red,green}} = v_{A.\texttt{red,green}}$
\item $p_{\texttt{red,blue}} = p_{\texttt{blue,green}} = t_{\texttt{red,blue,green}} = 0$
\end{itemize}

After finishing scanning the $k$ children edges of $v$, we can compute the shared triplets that are anchored in every child edge of $v$ (including the children edges pointing to contracted subtrees) as follows: for the total number of shared resolved triplets, denoted $\texttt{tot}_{\texttt{A.res}}$, we have that $\texttt{tot}_{\texttt{A.res}} = p_{\texttt{red,blue}} \cdot \overline{v}_{\texttt{black}}$ and for the total number of shared unresolved triplets, denoted $\texttt{tot}_{\texttt{A.unres}}$, we have that $\texttt{tot}_{\texttt{unres}} = t_{\texttt{red,blue,green}}$.
\enlargethispage{\baselineskip} 

The second category $B$ of counters help us count triplets involving leaves (contracted and non-contracted) from the subtree of $v$ and leaves from the contracted subtrees rooted at the edge~$(p_{v},v)$. We maintain the following:

\begin{itemize}
\item $v_{B.\texttt{red}}$: total number of red leaves in all contracted subtrees rooted at the edge $(p_{v},v)$.
\item $v_{B.\texttt{green}}$: total number of green leaves in all contracted subtrees rooted at the edge $(p_{v},v)$.
\item $v_{B.\texttt{black}}$: total number of black leaves in all contracted subtrees rooted at the edge $(p_{v},v)$.
\item $v_{B.\texttt{red,green}}$: total number of pairs of leaves where one is red and the other is green such that one leaf comes from a contracted child subtree of a contracted node $v'$ and the other leaf comes from a different contracted child subtree of the same contracted node $v'$.
\item $v_{B.\texttt{red,black}}$: total number of pairs of leaves where one is red and the other is black such that the red leaf comes from a contracted child subtree of a contracted node $v'$ and the black leaf comes from a contracted child subtree of a contracted node $v''$, where $v''$ is closer to $p_{v}$ than~$v'$.
\end{itemize}

For the total number of shared unresolved triplets, denoted $\texttt{tot}_{B.\texttt{unres}}$, that are anchored in the children edges of every contracted node that exists in edge $(p_{v},v)$, we have that $\texttt{tot}_{B.\texttt{unres}} = v_{\texttt{blue}} \cdot v_{B.\texttt{red,green}}$. For the total number of shared resolved triplets, denoted~$\texttt{tot}_{B.\texttt{res}}$, that are anchored in the children edges of every contracted node that exists in edge~$(p_{v},v)$, we have that~$\texttt{tot}_{B.\texttt{res}} = v_{\texttt{blue}} \cdot v_{B.\texttt{red,black}} + v_{\texttt{blue}} \cdot v_{B.\texttt{red}} \cdot (\overline{v}_{\texttt{black}} - v_{B.\texttt{black}})$.

\subsection{Scaling to External Memory.} 
\label{sec:iogeneral}

The analysis is the same as in Section \ref{sec:binary}, except for minor details. The proof of Lemma \ref{lemma3} can be trivially modified to apply to general trees as well. Finally, Lemma \ref{lemma:ioproof1} is generalized to non-binary trees in the following Lemma~\ref{lemma:ioproof2}. In Lemma \ref{lemma:ioproof2}, we consider the I/Os required to apply a depth first traversal on a non-binary tree $T$ that is stored in memory following a local layout, i.e., the nodes of every subtree of~$T$ are stored consecutively in memory and every node has $\mathrm{O}(1)$ occurrences in memory. Similarly to the assumptions we made for Lemma \ref{lemma:ioproof1}, w.l.o.g.\ we assume that when an edge $(u,v)$ of $T$ is processed in a depth first traversal of $T$, both $u$ and $v$ are visited, i.e.,~ both $u$ and $v$ are accessed in memory.

\begin{lemma}
\label{lemma:ioproof2}
Let~$T$ be a non-binary tree with~$n$ leaves that is stored in an array following a local layout, i.e., the nodes of every subtree of~$T$ are stored consecutively in memory and every node has $\mathrm{O}(1)$ occurrences in memory. Any depth first traversal that starts from the root of~$T$ and in which for every internal node~$u$ in~$T$, after the discovery of the first child of~$u$ the remaining children are discovered in order that they appear in memory from left to right, requires $\mathrm{O}(n/B)$ I/Os in the cache oblivious model.\end{lemma}

\begin{proof}

\begin{figure}[ht]
    \centering
        \includegraphics[width=1\textwidth]{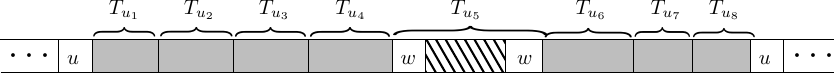}
   \caption{Position of a node $u$ in memory with respect to the 8 subtrees defined by the children of $u$, with $T_{u_{5}}$ being a largest subtree.}
    \label{fig:3}
\end{figure}

This proof can be thought of as an extension of the proof of Lemma~\ref{lemma:ioproof1}. Following the proof of Lemma \ref{lemma:ioproof1}, for a node~$u$ in~$T$, let~$T_u$ denote the set of nodes in the subtree defined by~$u$. For~$i \geq 2$, let~$u_{1},\dots,u_{i}$ be the children of~$u$ and let~$T_{u_{1}}, \dots, T_{u_{i}}$ be the corresponding subtrees. We assume that these subtrees are ordered from left to right in order that they appear in memory. In the proof of Lemma \ref{lemma:ioproof1}, we implicitly assumed that the positions of the two children of~$u$ are stored together with~$u$ in memory. For general trees, together with~$u$ we need to store a list of arbitrary size~$i \geq 2$ containing the positions in memory of every child of~$u$. To avoid complicating the presentation of the proof, we assume that we can find the position in memory of every child of~$u$ without this list, i.e., this list is not stored together with~$u$, thus finding the position of any child of~$u$ incurs no I/Os. An easy way to support this is to store in every node~$u$ in~$T$, one pointer to the first child to be discovered and one pointer to the sibling appearing next in memory. For every node $u$ in $T$, we allow a constant number of occurrences in memory. For any given placement of the copies of $u$ in memory, we add two copies of $u$ before the first child subtree and after the last child subtree. W.l.o.g.\ we assume that~$u$ is only stored before the first child subtree and after the last child subtree (see Figure \ref{fig:3} for an example).
\enlargethispage{\baselineskip} 
\begin{figure}
\captionsetup[subfigure]{justification=centering}
    \centering
    \begin{subfigure}{0.67\textwidth}
        \centering
        \includegraphics[width=1\textwidth]{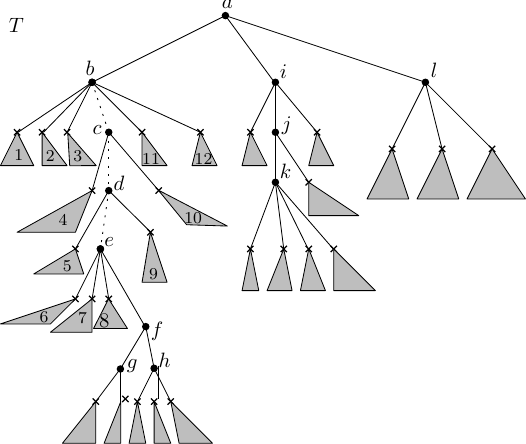}
        \caption{}
        \label{fig:4a}
    \end{subfigure}\hspace{10mm}%
    \begin{subfigure}{0.25\textwidth}
        \centering
        \includegraphics[width=1\textwidth]{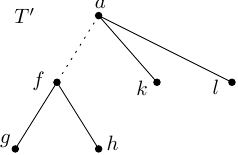}
        \caption{}
        \label{fig:4b}
    \end{subfigure}\\\par\bigskip\par\bigskip
     \begin{subfigure}{1\textwidth}
        \centering
        \includegraphics[width=1\textwidth]{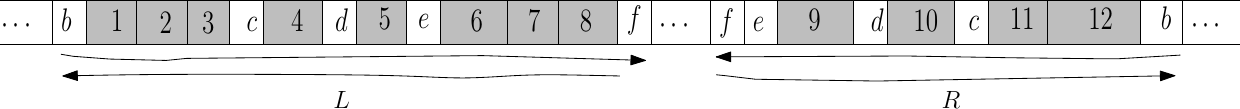}
        \caption{}
        \label{fig:4c}
    \end{subfigure}%
    \caption{(a)~A general tree $T$. The gray subtrees are $B$-light subtrees and every node not in a $B$-light subtree is a $B$-heavy node.~(b)~The corresponding tree T' according to the proof of Lemma~\ref{lemma:ioproof2}.~(c)~How~$T$ is stored in memory and the two segments of memory that correspond to the edge $(a,f)$ in $T'$.}
    \label{fig:4}
\end{figure}

Define a node $u$ in $T$ to be \emph{$B$-light} if $2|T_u| \leq B-2$, otherwise the node is said to be~\mbox{\emph{$B$-heavy}}. Observe that the children of a $B$-light node are all $B$-light. We consider the following disjoint sets of nodes from $T$:

\begin{enumerate}
\item[$S_{1}$:] Every $B$-light node 
\item[$S_{2}$:] Every $B$-heavy node with only $B$-light children
\item[$S_{3}$:] Every $B$-heavy node with at least two $B$-heavy children and an arbitrary number of \mbox{$B$-light} children
\item[$S_{4}$:] Every $B$-heavy node with exactly one $B$-heavy child and at least 1 $B$-light children.
\end{enumerate}

For a $B$-light node $u$ in $T$, let~$w$ be the first $B$-heavy node we reach in the path from~$u$ to the root of $T$. An I/O incurred by visiting the node $u$ in $T$ is charged to $w$. This node~$w$ can be either in~$S_{2}$,~$S_{3}$ or $S_{4}$. Let $w'$ be the child of~$w$ such that $T_{w'}$ contains $u$. Since~$2|T_{w'}| \leq B-2$, at most 1 I/O is sufficient to visit all nodes in $T_{w'}$. We say that $T_{w'}$ is a subtree that is~$B$-light. In Figure \ref{fig:4a} we have an example of a tree, where the gray subtrees denote $B$-light subtrees.

Similarly to the proof of Lemma \ref{lemma:ioproof1}, we have that~$|S_{2}| = \mathrm{O}(n/B)$ and $|S_{3}| = \mathrm{O}(n/B)$. Since~$T$ is non-binary, we have to argue that the number of I/Os spent traversing the~$B$-light subtrees that are rooted at every node in~$S_{2}$ and~$S_{3}$ is~$\mathrm{O}(n/B)$. For a node $u$ in~$T$, let~$G_u$ be the size of all gray subtrees rooted at~$u$. For every node~$u$ in~$S_{2}$ we spend at most~1~I/O to traverse the first chosen child subtree and~\mbox{$1 + |G_u|/B$}~I/Os to traverse the remaining subtrees, thus $2 + |G_u|/B$ I/Os in total. Since~\mbox{$|S_{2}| = \mathrm{O}(n/B)$} and the gray subtrees in~$T$ are disjoint, i.e.,~$\sum_{u \in T} |G_u| = \mathrm{O}(n)$, we spend~$\mathrm{O}(n/B)$~I/Os traversing the $B$-light subtrees rooted at every node in $S_{2}$. For every node~$u$ in~$S_{3}$, let~$d'(u)$ denote the number of $B$-heavy children of $u$. For this node~$u$, we spend at most~$1$~I/O to traverse the first chosen child subtree that could be~$B$-light and~\mbox{$1 + d'(u) + |G_u|/B$}~I/Os to traverse the remaining gray subtrees rooted at~$u$. Since~\mbox{$|S_{3}| = \mathrm{O}(n/B)$}, we have~\mbox{$\sum_{u \in T'} d'(u) = \mathrm{O}(n/B)$}. Together with the fact that~\mbox{$\sum_{u \in T} |G_u| = \mathrm{O}(n)$}, we spend $\mathrm{O}(n/B)$~I/Os traversing the $B$-light subtrees rooted at every node in $S_{3}$.

	We now argue that the total number of I/Os incurred by the nodes in $S_{4}$ is $\mathrm{O}(n/B)$, thus proving the statement. Let $T'$ be defined as in the proof of Lemma \ref{lemma:ioproof1}, as well as~$P_{(u,v)}$ and~$C_{(u,v)}$ for an edge~$(u,v)$ in $T'$. By the local layout followed to store $T$ in memory, the nodes in $C_{(u,v)}$ are stored in two segments of memory (e.g., see Figure~\ref{fig:4c}). Let $w$ be a node in $P_{(u,v)}$ and~$G_w$ be the total size of the gray subtrees rooted at $w$. We say that $w$ is~$G$-light if $2G_w \leq B-2$, otherwise~$G$-heavy. There can be $\mathrm{O}(n/B)$~$G$-heavy nodes in $T$, thus by the same argument as in the previous paragraph, scanning the gray subtrees for all $G$-heavy nodes together incurs~$\mathrm{O}(n/B)$ I/Os. For the~$G$-light nodes we follow a similar argument as in the proof of lemma~\ref{lemma:ioproof1}. Let~$L$ be the left chunk and $R$ the right chunk and w.l.o.g assume that every node in~$P_{(u,v)}$ is~$G$-light. During a depth first traversal of~$T$, visiting all nodes in~$P_{(u,v)}$ corresponds to visiting ~$L$ from left to right and then from right to left, and visiting $R$ from right to left and then from left to right. Let $c$ be the child of $w$ that is~$B$-heavy. Since for every node $w$ in~$P_{(u,v)}$ we have~$2G_w \leq B-2$, by accessing all two copies of $w$ and $c$ when~$c$ is visited in a depth first traversal of $T$, we guarantee that all the gray subtrees rooted at $w$ are in cache i.e., they can be accessed in memory for free. Hence, $\mathrm{O}(n/B)$ I/Os are sufficient to pay to traverse the gray subtrees of all $G$-light nodes. Overall, by having~$M \geq 5B$, where two blocks are used to hold copies of a node~$w$ in~$T$, two blocks are used to hold copies of a child of~$w$ and one block is used to scan gray subtrees, the statement follows.	
\end{proof}

\section{Implementation}
\label{sec:implementation}

\begin{figure}
    \centering
        \includegraphics[width=0.6\textwidth]{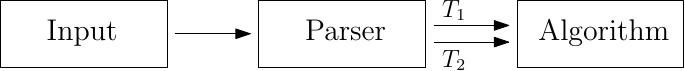}
    \caption{Implementation overview.}
    \label{fig:implOverview}
\end{figure}

The algorithms of Sections \ref{sec:binary} and \ref{sec:general} have been implemented in the C++ programming language. A high level overview of each implementation is illustrated in Figure \ref{fig:implOverview}. The source code is publicly available and can be found at \url{https://github.com/kmampent/CacheTD}.

\subsection{Input}

The two input trees $T_{1}$ and $T_{2}$ are stored in two separate text files following the Newick format. Both trees have $n$ leaves and the label of each leaf is assumed to be a number in~$\{1,2,\dots,n\}$. Two leaves cannot have the same~label.

\subsection{Parser}

The parser receives the files that store $T_{1}$ and $T_{2}$ in Newick format, and returns~$T_{1}$ and $T_{2}$ but now with $T_{1}$ stored in an array following the preorder layout and $T_{2}$ in an array following the postorder layout. The parser takes~$\mathrm{O}(n)$ time and space in the RAM model and $\mathrm{O}{(n/B)}$ I/Os in the cache oblivious~model.

\subsection{Algorithm}

Having $T_{1}$ and $T_{2}$ stored in memory following the desired layouts, we proceed with the main part of the algorithm. Both implementations (binary, general) follow the same approach. There exists an \emph{initialization} step and a \emph{distance computation} step.

\subparagraph*{Initialization.}

In the initialization step, the preprocessing parts of the algorithms are performed (see Sections \ref{sec:binaryMain} and \ref{sec:generalMain}), where the first component of $T_{1}$ is built, and the corresponding contracted version of $T_{2}$, from now on referred to as \emph{corresponding component} of~$T_{2}$, is built as well. After this step, the first component of $T_{1}$ is stored in an array (different than the one produced by the parser) following the preorder layout. Similarly, the corresponding component of $T_{2}$ is stored in an array following the postorder~layout.

\subparagraph*{Distance Computation.}

Let $\texttt{comp}(T_{1})$ and $\texttt{comp}(T_{2})$ be the component of $T_{1}$ and the corresponding component of $T_{2}$ produced by the initialization step. Having these two components available, we can begin counting shared triplets in order to compute $S(T_{1}, T_{2})$. The following steps are recursively applied:

\begin{itemize}
\item Starting from the root of $\texttt{comp}(T_{1})$ and according to Section \ref{sec:mcd}, scan the leftmost path of~$\texttt{comp}(T_{1})$ to find the splitting node $u$.
\item Scan $\texttt{comp}(T_{2})$ to compute for the binary algorithm $\sum_{v \in T_{2}} \lvert {s(u) \cap s(v) \rvert}$ (see counting phase of $T_{2}$ in Section \ref{sec:binaryMain}), or for the general algorithm $\sum_{(v,c') \in T_{2}} \lvert {s''(u,c) \cap s'(v,c')} \rvert$ (see counting phase of $T_{2}$ in Section \ref{sec:generalMain}).
\item Using the splitting node $u$, generate the next three components of $T_{1}$. Let $\texttt{comp}(T_{1}(u_{l}))$, $\texttt{comp}(T_{1}(u_{r}))$, and $\texttt{comp}(T_{1}(u_{p}))$ be the components determined by the left child, right child, and parent of~$u$ respectively. Let $\texttt{comp}(T_{2}(u_{l}))$, $\texttt{comp}(T_{2}(u_{r}))$ and $\texttt{comp}(T_{2}(u_{p}))$ be the corresponding contracted versions of~$T_{2}$ with all the necessary counters  properly maintained (see contraction phase of $T_{2}$ in Section~\ref{sec:binaryMain} for the binary case and in Section \ref{sec:generalMain} for the general case).
\item Scan and contract $\texttt{comp}(T_{2})$ to generate $\texttt{comp}(T_{2}(u_{l}))$ and then recurse on the pair defined by $\texttt{comp}(T_{1}(u_{l}))$ and $\texttt{comp}(T_{2}(u_{l}))$.
\item Scan and contract $\texttt{comp}(T_{2})$ to generate $\texttt{comp}(T_{2}(u_{r}))$ and then recurse on the pair defined by  $\texttt{comp}(T_{1}(u_{r}))$ and $\texttt{comp}(T_{2}(u_{r}))$.
\item Scan and contract $\texttt{comp}(T_{2})$ to generate $\texttt{comp}(T_{2}(u_{p}))$ and then recurse on the pair defined by $\texttt{comp}(T_{1}(u_{p}))$ and $\texttt{comp}(T_{2}(u_{p}))$.
\end{itemize}
As a final step, print $\binom{n}{3} - S(T_{1}, T_{2})$, which is equal to the triplet distance~$D(T_{1},T_{2})$.

\subparagraph*{Correctness.}

The correctness of our implementations was extensively tested by generating hundreds of thousands of random trees of varying size and varying degree and comparing the output of our implementations against the output of the implementations of the $\mathrm{O}(n \log^{3} n)$ algorithm in \cite{Jansson2015} and the~$\mathrm{O}(n \log n)$ algorithm in \cite{bioinformatics14}.

\subparagraph*{Changing the Leaf Labels.} To get the right theory bounds, changing the leaf labels of $T_{1}$ and~$T_{2}$ must be done with a cache oblivious sorting routine, e.g., merge sort. In the RAM model this approach takes $\mathrm{O}(n \log n)$ time and in the cache oblivious model $\mathrm{O}(\frac{n}{B}\log_{2} \frac{n}{M})$~I/Os. A second approach is to exploit the fact that each label is between $1$ and $n$ and use an auxiliary array that stores the new labels of the leaves in $T_{1}$, which we then use to update the leaf labels of $T_{2}$. In the RAM model this second approach takes $\mathrm{O}(n)$ time but in the cache oblivious model $\mathrm{O}(n)$ I/Os. In practice, the problem with the first approach is that the number of instructions it incurs eliminates any advantage that we expect to get due to its cache related efficiency for~$L_{1}$,~$L_{2}$, and~$L_{3}$ cache. For the input sizes tested, the array of labels easily fits into RAM, so in our implementation of both algorithms we~use~the~second~approach. 

\section{Experiments}
\label{sec:experiments}

In this section we provide an extensive experimental evaluation of the practical performance of the algorithms described in Sections \ref{sec:binary} and \ref{sec:general}.

\subsection{The Setup}

The experiments were performed on a machine with 8GB RAM, Intel(R) Core(TM) i5-3470 CPU @ 3.20GHz, 32K L1 cache, 256K L2 cache and 6144K L3 cache. The operating system was Ubuntu 16.04.2 LTS. The compilers used were g++ 5.4 and g++ 4.7, together with cmake 3.5.1. The experiments were performed in text mode, i.e., by booting into the terminal of Ubuntu, to minimize the interference from other programs running at the same time.

\subparagraph*{Generating Random Trees.}

We use two different models for generating input trees. The first model is called the \emph{random model}. A tree $T$ with $n$ leaves in this model is generated as follows: 

\begin{itemize}
\item Create a binary tree $T$ with $n$ leaves as follows: start with a binary tree~$T$ with two leaves. Iteratively pick $n-1$ times a leaf $l$ uniformly at random. Make $l$ an internal node by appending a left child node and a right child node to $l$, thus increasing the number of leaves in $T$ by exactly 1.
\item With probability $p$ contract every internal node $u$ of $T$, i.e., make the children of $u$ be the children of $u$'s parent and remove $u$.
\end{itemize}

The second model is called the \emph{skewed} model. In this model, we can control more directly the shape of the input trees. A tree $T$ with $n$ leaves in this model is generated as follows: 

\begin{itemize}
\item Create a binary tree $T$ with $n$ leaves as follows: let $0 \leq \alpha \leq 1$ be a parameter, $u$ some internal node in $T$, $l$ and $r$ the left and right children of $u$, and $T(u)$, $T(l)$, and $T(r)$ the subtrees rooted at $u$, $l$, and~$r$ respectively. Create $T$ so that for every internal node $u$ we have $\frac{|T(l)|}{|T(u)|} \approx \alpha$, i.e., if $n'$ is the number of leaves below $T(u)$, and $|\Lambda_{l}|$ and~$|\Lambda_{r}|$ are the number of leaves in $T(l)$ and $T(r)$ respectively, first choose \mbox{$|\Lambda_{l}| = \max(1, \min(\floor{\alpha\cdot n'},n'-1))$} and then let $|\Lambda_{r}| = n' - |\Lambda_{l}|$.
\item With probability $p$ contract every internal node $u$ of $T'$ like in the random~model.
\end{itemize}

In both models and after creating $T$, we shuffle the leaf labels by using \texttt{std::shuffle}\footnote{\url{http://www.cplusplus.com/reference/algorithm/shuffle/}} together with \texttt{std::default\_random\_engine}\footnote{\url{http://www.cplusplus.com/reference/random/default_random_engine/}}.

\subparagraph*{Implementations Tested.}

Let $p_{1}$ and $p_{2}$ denote the contraction probability of $T_{1}$ and $T_{2}$ respectively. When $p_{1} = p_{2} = 0$, the trees $T_{1}$ and $T_{2}$ are binary trees, so in the experiments we use the algorithm from Section \ref{sec:binary}. In every other case, the algorithm from Section \ref{sec:general} is used. Note that the algorithm from Section \ref{sec:general} can handle binary trees just fine, however there is an extra overhead (factor 1.8 slower, see Figure \ref{fig:figOverhead}) compared to the algorithm from Section \ref{sec:binary} that comes due to the additional counters that we maintain in the contractions of~$T_{2}$.

We compared our implementation with the implementations provided in \cite{Jansson2015} and \cite{bmc13, TripletQuartetSoda13}, and are available at \url{http://sunflower.kuicr.kyoto-u.ac.jp/~jj/Software/Software.html} and \url{http://users-cs.au.dk/cstorm/software/tqdist/} respectively. The implementation of the algorithm in \cite{Jansson2015} has two versions, one that uses \texttt{unordered\_map}\footnote{\url{http://en.cppreference.com/w/cpp/container/unordered_map}}, which we refer to as \texttt{CPDT}, and another that uses sparsehash\footnote{\url{https://github.com/sparsehash/sparsehash}}, which we refer to as \texttt{CPDTg}. For binary input trees the hash maps are not used, thus \texttt{CPDT} and \texttt{CPDTg} are the same. The tqdist library \cite{bioinformatics14}, which we refer to as \texttt{tqDist}, has an implementation of the binary $\mathrm{O}(n \log^2 n)$ algorithm from~\cite{bmc13} and the general  $\mathrm{O}(n \log n)$ algorithm from \cite{TripletQuartetSoda13}. If the two input trees are binary the~$\mathrm{O}(n \log^2 n)$ algorithm is used. We refer to our new algorithm as~\texttt{CacheTD}.

\subparagraph*{Statistics.}

We measured the execution time of the algorithms with the \texttt{clock\_gettime} function in C++. Due to the different parser implementations, we do not consider the time taken to parse the input trees. We used the PAPI library\footnote{\url{http://icl.utk.edu/papi/}} for statistics related to instructions, L1, L2, and~L3 cache accesses and misses. Finally, we count the space of the algorithms by considering the~\emph{Maximum resident set size} returned by \texttt{/usr/bin/time -v}.

\subsection{Results}

The experiments are divided into two parts.  In the first part, we consider the performance of the algorithms when their memory requirements do not exceed the available main memory (8G RAM).  In the second part, we consider the performance when the memory requirements exceed the available main memory (by limiting the available RAM to the operating system to be 1GB), thus forcing the operating system to start using the swap space, which in turn yields the very expensive disk I/Os. All figures can be found in Appendix \ref{appendix:experimentFigures}.

\subparagraph*{RAM experiments in the Random Model.}

In Figure \ref{fig:expTimeSmall54} we illustrate a time comparison of all implementations for trees of up to $2^{21}$ leaves ($\sim 2$ million) with varying contraction probabilities. Every experiment is run 10 times, and each time on a different tree. All 10 data points are depicted together with a line that goes over their median. The compilers used were g++~5.4 with cmake~3.5.1 for \texttt{tqDist} and g++~5.4 for \texttt{CPDT}, \texttt{CPDTg}, and \texttt{CacheTD}. In all cases, \texttt{CacheTD} achieves the best performance. We note that for the case where $p_{1} = 0.95$ and $p_{2} = 0.2$, \texttt{CPDT} behaves in a different way compared to the experiments in \cite{Jansson2015}. The same can be observed for the case where $p_{1} = 0.8$ and~$p_{2} = 0.8$. The reason is because of the differences in the implementation of \texttt{unordered\_map} that exist between the different versions of the g++ compilers.  In Figure~\ref{fig:gpp} we compare the performance of \texttt{CPDT} when compiled with g++ 4.7 and g++ 5.4. When $p_{1}$ is large, i.e., $p_{1} = 0.8$ and $p_{1} = 0.95$, we observe that the older version of g++ achieves a better performance. For all other values of $p_{1}$, the version of the compiler has no effect on the performance. In Figure \ref{fig:expTimeSmall47} we have another time comparison of all implementations but now with \texttt{CPDT} compiled in g++ 4.7. The new algorithm achieves the best performance again, but now the behaviour of \texttt{CPDT} is more stable when $p_{1}$ is large. From now on, in every RAM experiment \texttt{CPDT} is compiled in g++ 4.7.

In Figure \ref{fig:expSpaceSmall} we show the space consumption of the algorithms. \texttt{CacheTD} is the only algorithm that uses $\mathrm{O}(n)$ space for both binary and general trees. In theory we expect that the space consumption is better and this is also what we get in practice. 

In Figures \ref{fig:expPtimeSmall} and \ref{fig:expPspaceSmall} we can see how the contraction parameter affects the running time and the space consumption of the algorithms respectively. 

Finally, in Figures \ref{fig:l1misses}, \ref{fig:l2misses} and \ref{fig:l3misses} we compare the cache performance of the algorithms, i.e., how many cache misses (L1, L2 and L3 respectively) the algorithms perform for increasing input sizes and varying contraction parameters. As expected, the new algorithm achieves a significant improvement over all previous algorithms.

\subparagraph*{RAM experiments in the Skewed Model.}

The main interesting experimental results are illustrated in Figure \ref{fig:alphaFigureTime}, where we plot the alpha parameter against the execution time of the algorithms, when~$n = 2^{21}$. The alpha parameter has the least effect on \texttt{CacheTD}, with the maximum running time in every graph of Figure \ref{fig:alphaFigureTime} being only a factor of 1.15 larger than the minimum. As mentioned in Section \ref{sec:prevApproaches}, \texttt{CPDT} and \texttt{CPDTg} use the heavy light decomposition for $T_{2}$. For binary trees, when~$\alpha$ approaches 0 or 1, the number of heavy paths that have to be updated because of a leaf color change decreases, thus the total number of operations of the algorithm decreases as well. We can verify this in Figure \ref{fig:alphaFigureInstructions}, where we have the plots of the alpha parameter against the instructions. The same cannot be said for all general trees, since the contraction parameters have an effect on the shape of the trees as well. In Figures \ref{fig:alphaFigureL1misses},  \ref{fig:alphaFigureL2misses}, and~\ref{fig:alphaFigureL3misses} we have the same graphs but for L1, L2, and L3 cache misses respectively.

\begin{table}[ht]
  \centering
  \caption{Random model: Time performance when limiting the available RAM to be 1GB. For the left table we have $p_{1} = p_{2} = 0$ and for the right table~$p_{1} = p_{2} = 0.5$.}
  \begin{tabular}{ccccc}
  	\toprule
  	$n$ &\texttt{CPDT}&\texttt{tqDist}& \texttt{CacheTD}\\
  	\midrule
    $2^{15}$ & 0m:01s & 0m:01s & 0m:01s\\
    $2^{16}$ & 0m:01s  & 0m:02s & 0m:01s\\
    $2^{17}$ & 0m:01s & 0m:04s & 0m:01s\\
    $2^{18}$ & 0m:02s  & 1m:03s & 0m:01s\\
    $2^{19}$ & 0m:04s  & 1h:21m & 0m:01s\\
    $2^{20}$ & 0m:09s & \bfseries 0\% & 0m:01s\\
    $2^{21}$ & 13m:12s  & - & 0m:03s\\
    $2^{22}$ & \bfseries 0\% & - & 0m:09s\\
    $2^{23}$ & -  & - & 3m:37s\\
    $2^{24}$ & - & - & 10m:35s\\\hline
  \end{tabular}\hspace{0.05\textwidth}%
  \begin{tabular}{ccccc}
  	\toprule
  	$n$ &\texttt{CPDT}&\texttt{CPDTg}&\texttt{tqDist}& \texttt{CacheTD}\\
  	\midrule
    $2^{15}$ & 0m:01s & 0m:01s & 0m:01s& 0m:01s\\
    $2^{16}$ & 0m:01s  & 0m:01s & 0m:01s& 0m:01s\\
    $2^{17}$ &  0m:01s & 0m:01s &  0m:03s & 0m:01s\\
    $2^{18}$ & 0m:03s & 0m:03s &  0m:07s & 0m:01s\\
     $2^{19}$ & 0m:07s & 0m:07s &  5m:20s & 0m:01s\\
     $2^{20}$ &  3m:43s  &1h:13m& \bfseries 0\% & 0m:02s\\
    $2^{21}$ & \bfseries 15\%  &\bfseries 0\%& - & 0m:20s\\
     $2^{22}$ &  -  &-& - & 2m:02s\\
      $2^{23}$ & - &-& - & 10m:42s\\
     $2^{24}$ &  -  &-& - & 42m:06s\\\hline
  \end{tabular}
  \label{tab:IOtimeFINALrandom}
\end{table}

\begin{table}[ht]
  \centering
  \caption{Skewed model: Time performance when limiting the available RAM to be 1GB. For both tables we have $\alpha = 0.5$. For the left table we have $p_{1} = p_{2} = 0$ and for the right table $p_{1} = p_{2} = 0.5$.}
  \begin{tabular}{ccccc}
  	\toprule
  	$n$ &\texttt{CPDT}&\texttt{tqDist}& \texttt{CacheTD}\\
  	\midrule
    $2^{15}$ & 0m:01s & 0m:01s & 0m:01s\\
    $2^{16}$ & 0m:01s  & 0m:02s & 0m:01s\\
    $2^{17}$ & 0m:01s & 0m:05s & 0m:01s\\
    $2^{18}$ & 0m:02s  & 0m:54s & 0m:01s\\
    $2^{19}$ & 0m:05s  & 50m:38s & 0m:01s\\
    $2^{20}$ & 0m:13s & \bfseries 0\% & 0m:01s\\
    $2^{21}$ & 20m:02s  & - & 0m:03s\\
    $2^{22}$ & \bfseries 0\% & - & 0m:09s\\
    $2^{23}$ & -  & - & 3m:46s\\
    $2^{24}$ & - & - & 13m:36s\\\hline
  \end{tabular}\hspace{0.05\textwidth}%
  \begin{tabular}{ccccc}
  	\toprule
  	$n$ &\texttt{CPDT}&\texttt{CPDTg}&\texttt{tqDist}& \texttt{CacheTD}\\
  	\midrule
    $2^{15}$ & 0m:01s & 0m:01s & 0m:01s& 0m:01s\\
    $2^{16}$ & 0m:01s  & 0m:01s & 0m:01s& 0m:01s\\
    $2^{17}$ &  0m:01s & 0m:01s &  0m:03s & 0m:01s\\
    $2^{18}$ & 0m:03s & 0m:03s &  0m:06s & 0m:01s\\
     $2^{19}$ & 0m:07s & 0m:07s &  3m:21s & 0m:01s\\
     $2^{20}$ &  6m:24s &2h:31m& \bfseries 7h:51m & 0m:02s\\
    $2^{21}$ & \bfseries 12\%  & \bfseries 0\%& - & 0m:19s\\
     $2^{22}$ &  -  &-& - & 1m:58s\\
      $2^{23}$ & - &-& - & 9m:42s\\
     $2^{24}$ &  -  &-& - & 38m:19s\\\hline
  \end{tabular}
  \label{tab:IOtimeFINALskewed}
\end{table}

\subparagraph*{I/O experiments.}

In Figures \ref{fig:IOtimeRandom} and \ref{fig:IOtimeSkewed} we illustrate the time, space, and I/O performance in the random and skewed model respectively. Every implementation was compiled with g++~5.4. Every experiment is run 5 times, each on a different tree. Like in the RAM experiments, all 5 data points are displayed together with a line that goes over their median. To measure the execution time, we used the \texttt{time} function of Ubuntu and thus also took into account the time taken to parse the input trees. For the input trees of size~$2^{23}$ and $2^{24}$ we used the 128 bit implementation of the new algorithms in order to avoid overflows. 

Unlike \texttt{CacheTD}, the performance of \texttt{CPDT}, \texttt{CPDTg}, and \texttt{tqDist} deteriorates significantly from the moment they start performing disk I/Os. Only \texttt{CacheTD} managed to finish running in a reasonable amount of time for all input sizes. For every other algorithm, some data points are missing because the execution time required was too big. To get an idea of how big, in Tables~\ref{tab:IOtimeFINALrandom} and \ref{tab:IOtimeFINALskewed} we again have the time performance of the algorithms in the random and skewed models respectively. This is the exact same time performance as depicted in Figures \ref{fig:IOtimeRandom} and \ref{fig:IOtimeSkewed}, however we also include some information about how well the algorithms performed on the extra data point that is missing from the figures. We set a time limit of 10 hours, and only for one pair of input trees $T_{1}$ and $T_{2}$ we measured for how many nodes of $T_{1}$ the value of~$\sum_{v\in T_{2}}|s(u) \cap s(v)|$ was found. Some algorithms managed to process only~0\% of the total nodes in~$T_{1}$, which means that they had to spend most of the time in the preprocessing step (e.g.\ building the HDT of $T_{2}$). The only algorithm that managed to produce a result was \texttt{tqDist}, requiring close to~8 hours for trees with $2^{20}$ leaves (see Table~\ref{tab:IOtimeFINALskewed}).

\section{Conclusion}
\label{sec:conclusion}%

In this paper we presented two cache oblivious algorithms for computing the triplet distance between two rooted unordered trees, one that works for binary trees and one that works for arbitrary degree trees. Both require $\mathrm{O}(n \log n)$ time in the RAM model and $\mathrm{O}(\frac{n}{B} \log_{2} \frac{n}{M})$ I/Os in the cache oblivious model. We implemented the algorithms in C++ and showed with experiments that their performance surpasses the performance of previous implementations for this problem. In particular, our algorithms are the first to scale to external~memory.

Future work and open problems involve the following:
\enlargethispage{\baselineskip} 
\begin{itemize}
\item Could the new algorithms be improved so that in the analysis, the base of the logarithm becomes $M/B$, thus giving the sorting bound in the cache oblivious model? Would the resulting algorithm be even more efficient in practice?
\item Is it possible to compute the triplet distance in $O(n)$ time? 
\item For the quartet distance computation, could we apply similar techniques to those described in Section \ref{sec:binary} and \ref{sec:general} in order to get an algorithm with better time bounds in the RAM model that also scales to external~memory?
\end{itemize}


\bibliography{lipics-v2016-sample-article}

\newpage
\appendix
\section{Experiment Figures}
\label{appendix:experimentFigures}

\begin{figure}[ht]
    \centering
        \includegraphics[page=1,width=0.8\textwidth]{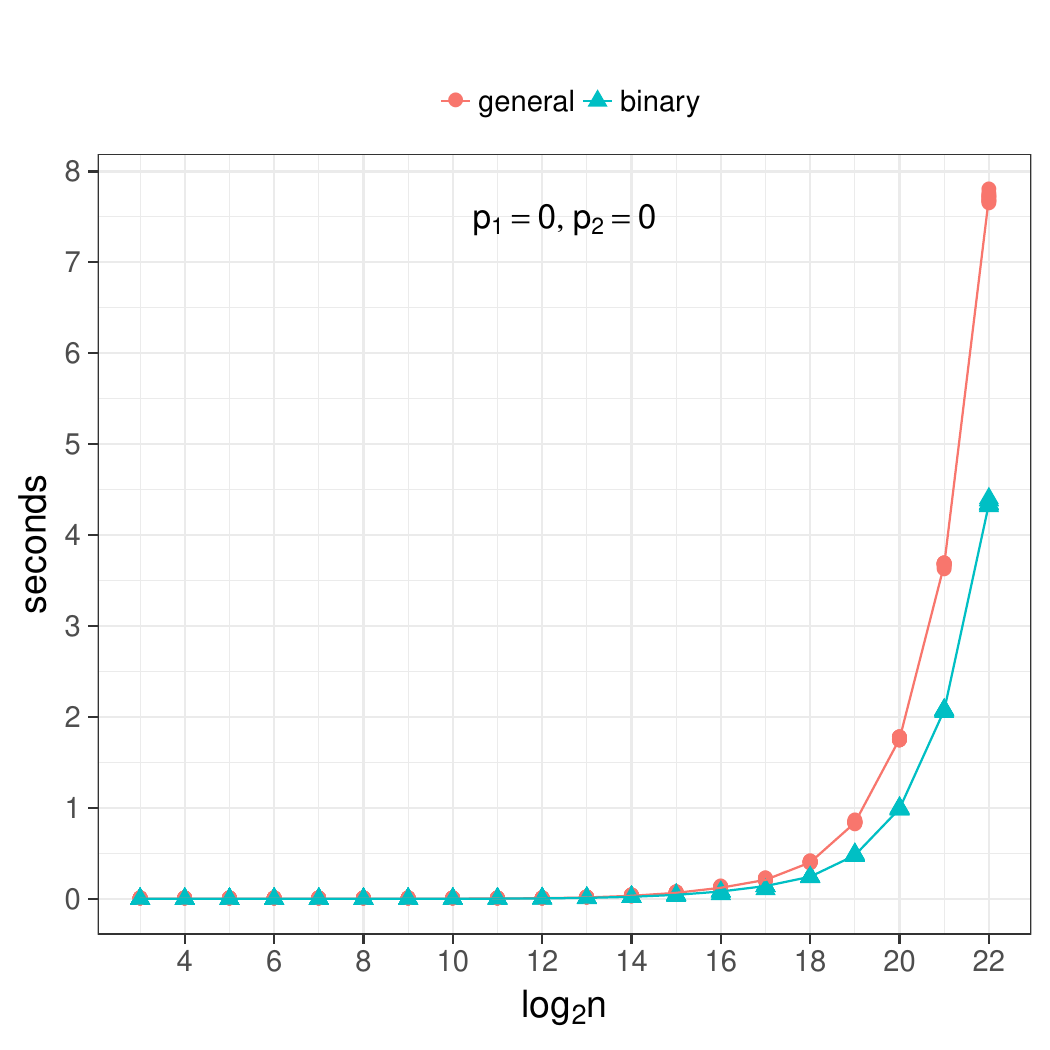}
    \caption{CacheTD: performance of binary (Section \ref{sec:binary}) and general (Section \ref{sec:general}) implementation on binary trees. All data points of the 10 runs are visible in the figure. Each run is on a different tree and the line connects the median of the  runs.}
    \label{fig:figOverhead}
\end{figure}

\begin{figure}
\captionsetup[subfigure]{justification=centering}
    \centering
    \begin{subfigure}{0.5\textwidth}
        \centering
        \includegraphics[page=1,width=1\textwidth]{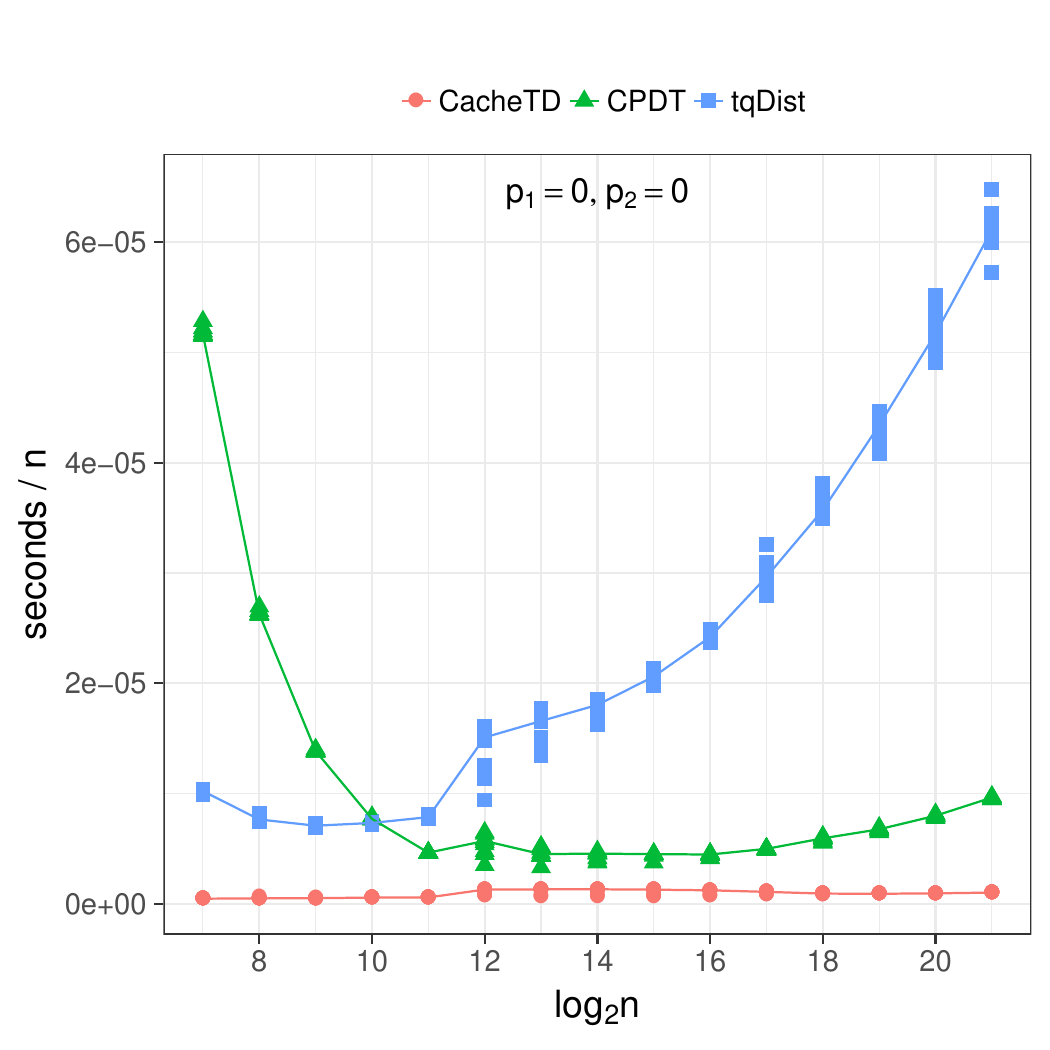}
    \end{subfigure}%
    \begin{subfigure}{0.5\textwidth}
        \centering
         \includegraphics[page=2,width=1\textwidth]{consoleFigures/experiments/random/time54.pdf}
    \end{subfigure}%
    \\
     \begin{subfigure}{0.5\textwidth}
        \centering
         \includegraphics[page=3,width=1\textwidth]{consoleFigures/experiments/random/time54.pdf}
    \end{subfigure}%
    \begin{subfigure}{0.5\textwidth}
        \centering
         \includegraphics[page=4,width=1\textwidth]{consoleFigures/experiments/random/time54.pdf}
    \end{subfigure}%
    \\
     \begin{subfigure}{0.5\textwidth}
        \centering
         \includegraphics[page=5,width=1\textwidth]{consoleFigures/experiments/random/time54.pdf}
    \end{subfigure}%
     \begin{subfigure}{0.5\textwidth}
        \centering
         \includegraphics[page=6,width=1\textwidth]{consoleFigures/experiments/random/time54.pdf}
    \end{subfigure}
    \caption{Random model: time performance, where \texttt{CPDT} is compiled in g++ version 5.4.}
    \label{fig:expTimeSmall54}
\end{figure}

\begin{figure}[ht]
\captionsetup[subfigure]{justification=centering}
    \centering
    \begin{subfigure}{0.5\textwidth}
        \centering
        \includegraphics[page=1,width=1\textwidth]{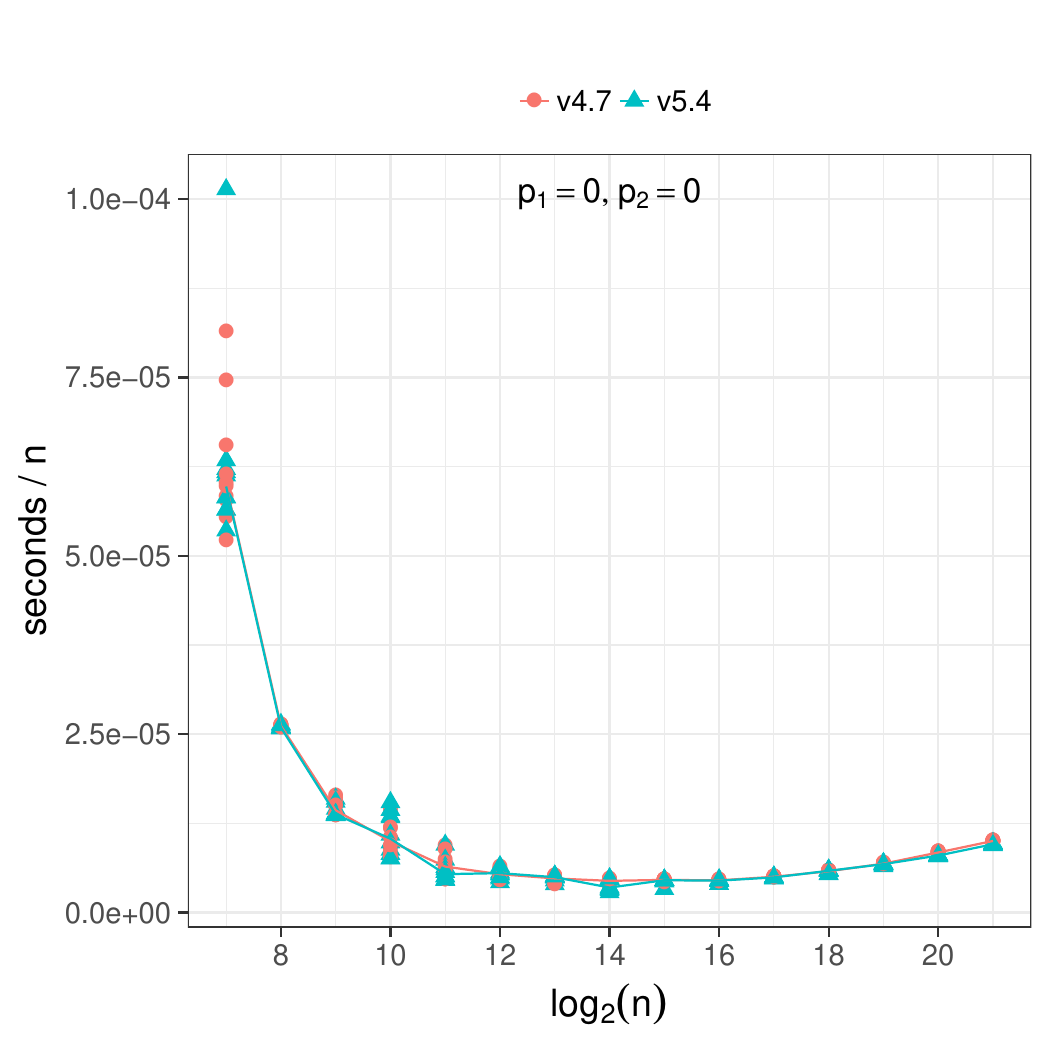}
    \end{subfigure}%
    \begin{subfigure}{0.5\textwidth}
        \centering
         \includegraphics[page=2,width=1\textwidth]{consoleFigures/experiments/random/compiler/gpp.pdf}
    \end{subfigure}%
    \\
     \begin{subfigure}{0.5\textwidth}
        \centering
         \includegraphics[page=3,width=1\textwidth]{consoleFigures/experiments/random/compiler/gpp.pdf}
    \end{subfigure}%
    \begin{subfigure}{0.5\textwidth}
        \centering
         \includegraphics[page=4,width=1\textwidth]{consoleFigures/experiments/random/compiler/gpp.pdf}
    \end{subfigure}%
    \\
     \begin{subfigure}{0.5\textwidth}
        \centering
         \includegraphics[page=5,width=1\textwidth]{consoleFigures/experiments/random/compiler/gpp.pdf}
    \end{subfigure}%
     \begin{subfigure}{0.5\textwidth}
        \centering
         \includegraphics[page=6,width=1\textwidth]{consoleFigures/experiments/random/compiler/gpp.pdf}
    \end{subfigure}
    \caption{Random model: time performance of \texttt{CPDT} when compiled with g++ 4.7 and g++ 5.4.}
    \label{fig:gpp}
\end{figure}

\begin{figure}
\captionsetup[subfigure]{justification=centering}
    \centering
    \begin{subfigure}{0.5\textwidth}
        \centering
        \includegraphics[page=1,width=1\textwidth]{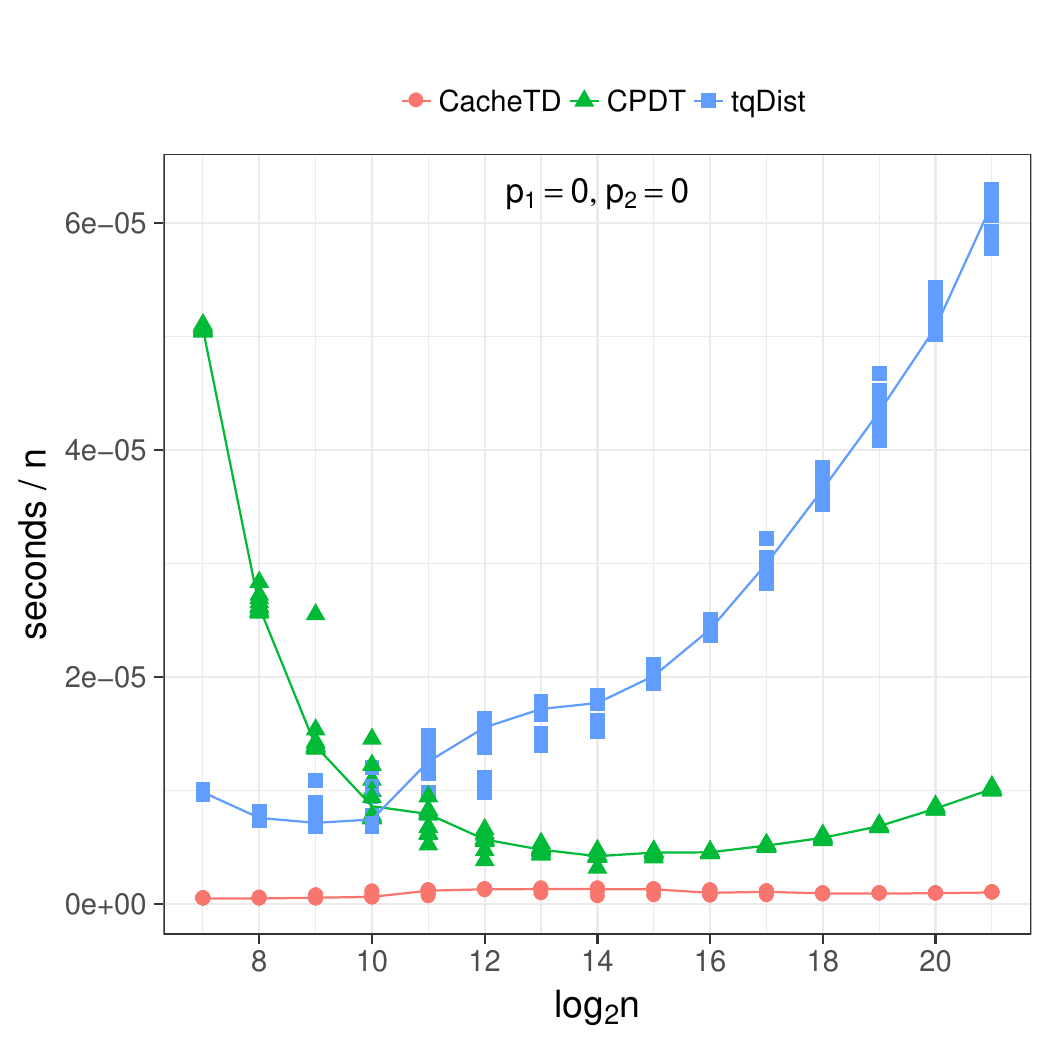}
    \end{subfigure}%
    \begin{subfigure}{0.5\textwidth}
        \centering
         \includegraphics[page=2,width=1\textwidth]{consoleFigures/experiments/random/time47.pdf}
    \end{subfigure}%
    \\
     \begin{subfigure}{0.5\textwidth}
        \centering
         \includegraphics[page=3,width=1\textwidth]{consoleFigures/experiments/random/time47.pdf}
    \end{subfigure}%
    \begin{subfigure}{0.5\textwidth}
        \centering
         \includegraphics[page=4,width=1\textwidth]{consoleFigures/experiments/random/time47.pdf}
    \end{subfigure}%
    \\
     \begin{subfigure}{0.5\textwidth}
        \centering
         \includegraphics[page=5,width=1\textwidth]{consoleFigures/experiments/random/time47.pdf}
    \end{subfigure}%
     \begin{subfigure}{0.5\textwidth}
        \centering
         \includegraphics[page=6,width=1\textwidth]{consoleFigures/experiments/random/time47.pdf}
    \end{subfigure}
    \caption{Random model: time performance, where \texttt{CPDT} is compiled in g++ version 4.7.}
    \label{fig:expTimeSmall47}
\end{figure}

\begin{figure}
\captionsetup[subfigure]{justification=centering}
    \centering
    \begin{subfigure}{0.5\textwidth}
        \centering
        \includegraphics[page=1,width=1\textwidth]{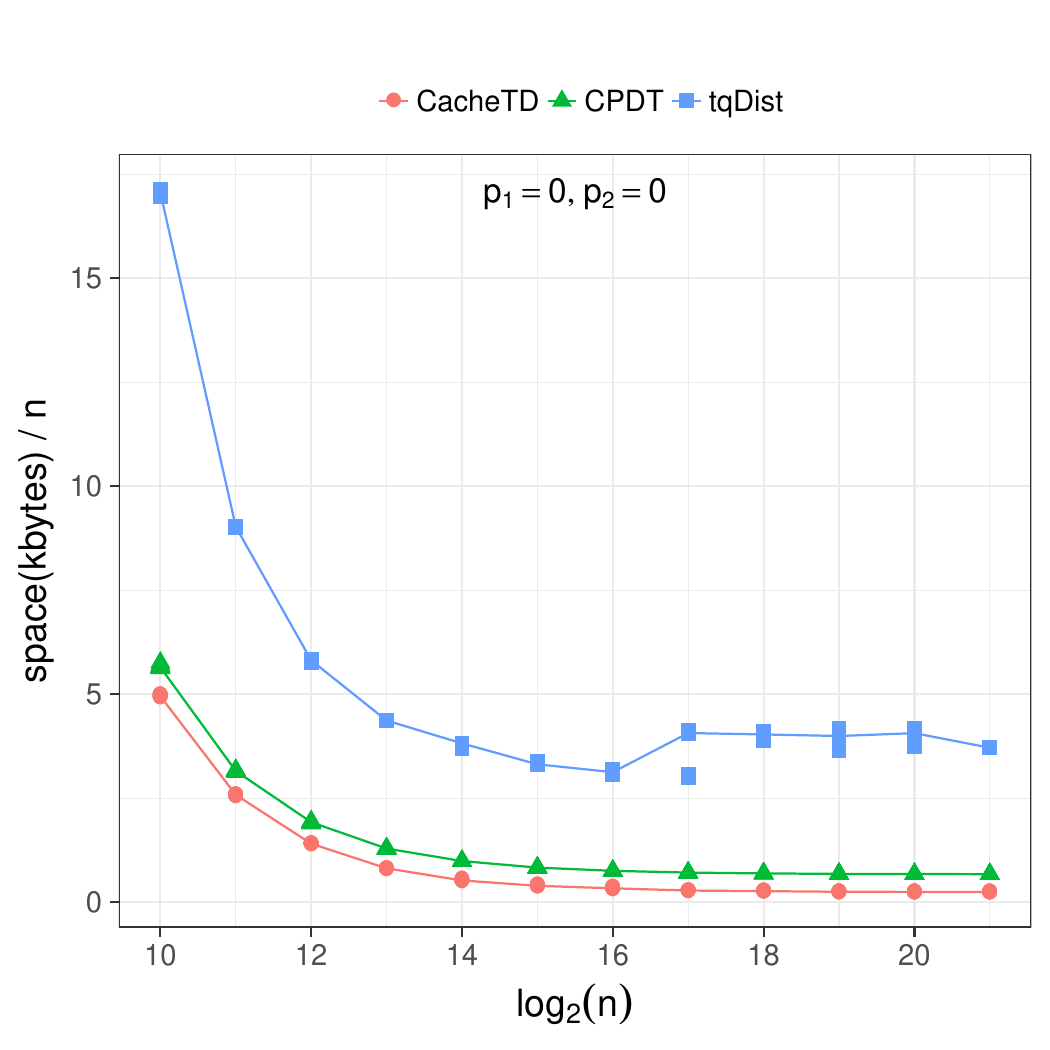}
    \end{subfigure}%
    \begin{subfigure}{0.5\textwidth}
        \centering
         \includegraphics[page=2,width=1\textwidth]{consoleFigures/experiments/random/space.pdf}
    \end{subfigure}%
    \\
     \begin{subfigure}{0.5\textwidth}
        \centering
         \includegraphics[page=3,width=1\textwidth]{consoleFigures/experiments/random/space.pdf}
    \end{subfigure}%
    \begin{subfigure}{0.5\textwidth}
        \centering
         \includegraphics[page=4,width=1\textwidth]{consoleFigures/experiments/random/space.pdf}
    \end{subfigure}%
    \\
     \begin{subfigure}{0.5\textwidth}
        \centering
         \includegraphics[page=5,width=1\textwidth]{consoleFigures/experiments/random/space.pdf}
    \end{subfigure}%
     \begin{subfigure}{0.5\textwidth}
        \centering
         \includegraphics[page=6,width=1\textwidth]{consoleFigures/experiments/random/space.pdf}
    \end{subfigure}
    \caption{Random model: space performance.}
    \label{fig:expSpaceSmall}
\end{figure}

\begin{figure}
\captionsetup[subfigure]{justification=centering}
    \centering
    \begin{subfigure}{0.5\textwidth}
        \centering
        \includegraphics[page=1,width=1\textwidth]{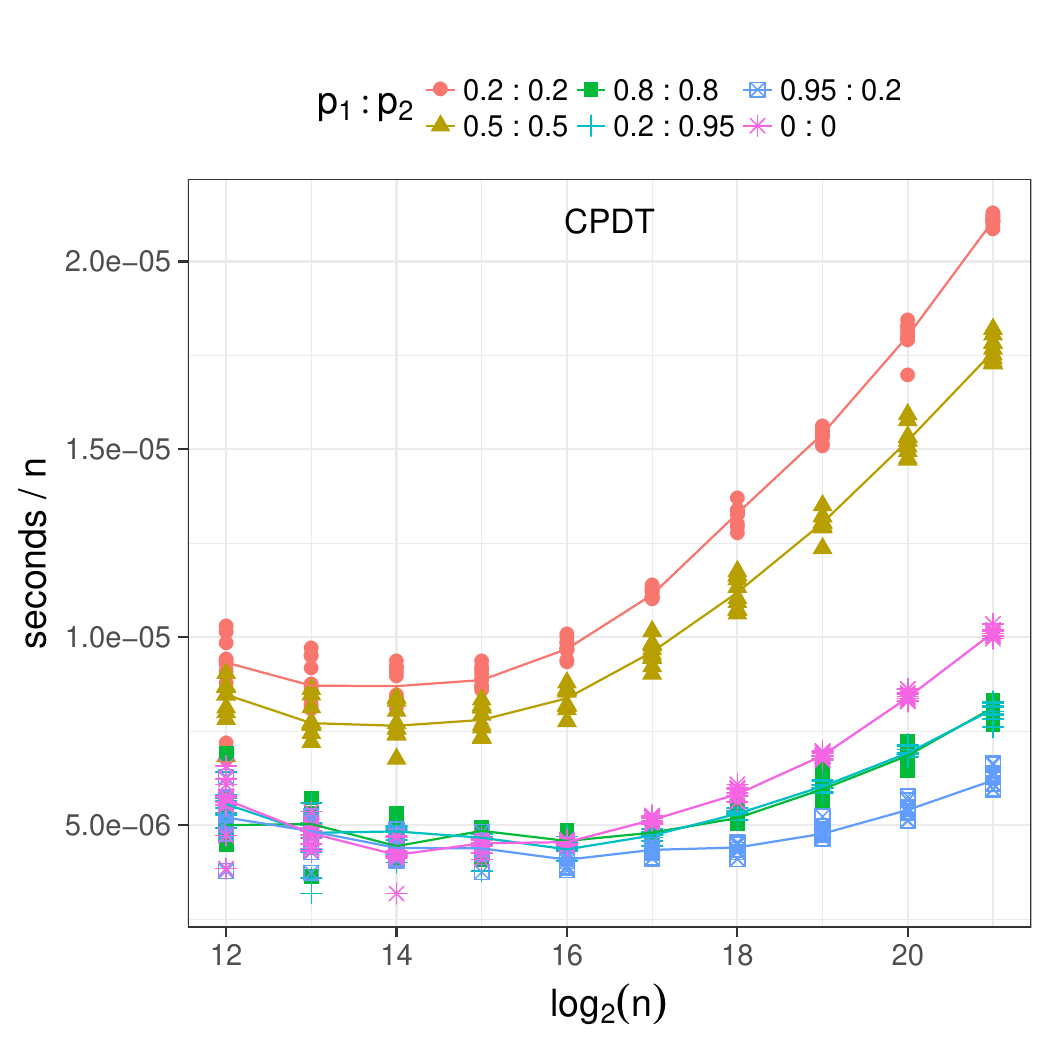}
    \end{subfigure}%
    \begin{subfigure}{0.5\textwidth}
        \centering
         \includegraphics[page=2,width=1\textwidth]{consoleFigures/experiments/random/pTime.pdf}
    \end{subfigure}%
    \\
     \begin{subfigure}{0.5\textwidth}
        \centering
         \includegraphics[page=3,width=1\textwidth]{consoleFigures/experiments/random/pTime.pdf}
    \end{subfigure}%
    \begin{subfigure}{0.5\textwidth}
        \centering
         \includegraphics[page=4,width=1\textwidth]{consoleFigures/experiments/random/pTime.pdf}
    \end{subfigure}
    \caption{Random model: how the contraction parameter affects execution time.}
    \label{fig:expPtimeSmall}
\end{figure}

\begin{figure}
\captionsetup[subfigure]{justification=centering}
    \centering
    \begin{subfigure}{0.5\textwidth}
        \centering
        \includegraphics[page=1,width=1\textwidth]{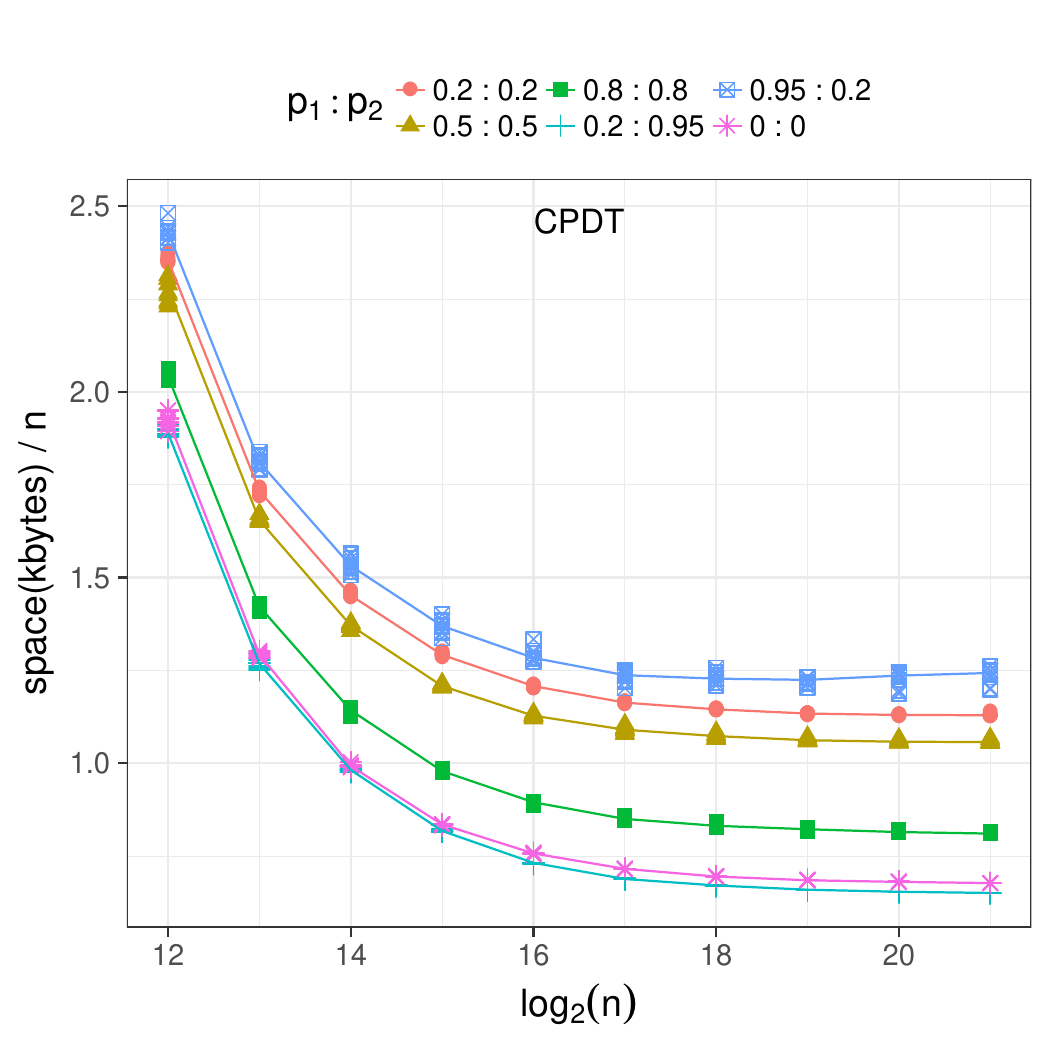}
    \end{subfigure}%
    \begin{subfigure}{0.5\textwidth}
        \centering
         \includegraphics[page=2,width=1\textwidth]{consoleFigures/experiments/random/pSpace.pdf}
    \end{subfigure}%
    \\
     \begin{subfigure}{0.5\textwidth}
        \centering
         \includegraphics[page=3,width=1\textwidth]{consoleFigures/experiments/random/pSpace.pdf}
    \end{subfigure}%
    \begin{subfigure}{0.5\textwidth}
        \centering
         \includegraphics[page=4,width=1\textwidth]{consoleFigures/experiments/random/pSpace.pdf}
    \end{subfigure}
    \caption{Random model: how the contraction parameter affects space.}
    \label{fig:expPspaceSmall}
\end{figure}

\begin{figure}
\captionsetup[subfigure]{justification=centering}
    \centering
    \begin{subfigure}{0.5\textwidth}
        \centering
        \includegraphics[page=1,width=1\textwidth]{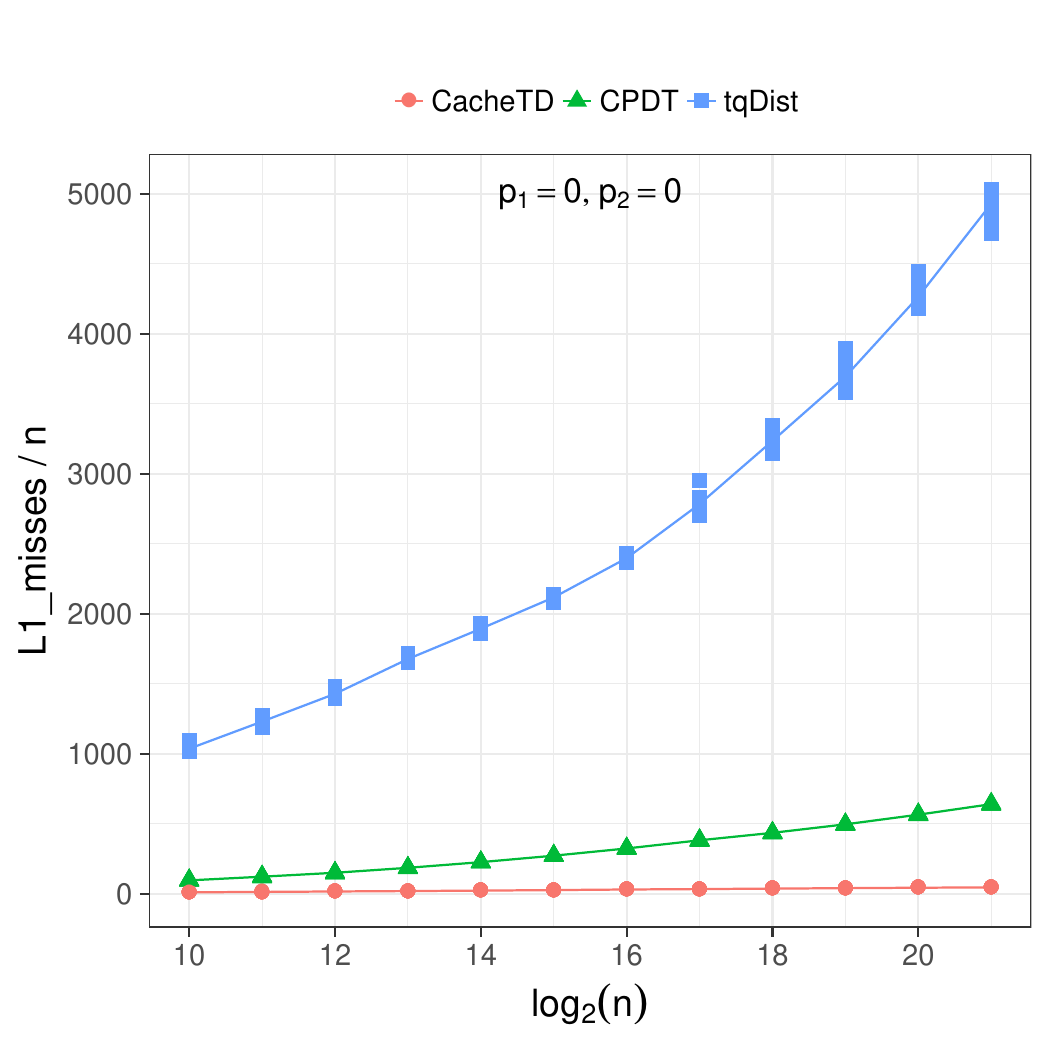}
    \end{subfigure}%
    \begin{subfigure}{0.5\textwidth}
        \centering
         \includegraphics[page=2,width=1\textwidth]{consoleFigures/experiments/random/L1misses.pdf}
    \end{subfigure}%
    \\
     \begin{subfigure}{0.5\textwidth}
        \centering
         \includegraphics[page=3,width=1\textwidth]{consoleFigures/experiments/random/L1misses.pdf}
    \end{subfigure}%
    \begin{subfigure}{0.5\textwidth}
        \centering
         \includegraphics[page=4,width=1\textwidth]{consoleFigures/experiments/random/L1misses.pdf}
    \end{subfigure}    
    \\
     \begin{subfigure}{0.5\textwidth}
        \centering
         \includegraphics[page=5,width=1\textwidth]{consoleFigures/experiments/random/L1misses.pdf}
    \end{subfigure}%
    \begin{subfigure}{0.5\textwidth}
        \centering
         \includegraphics[page=6,width=1\textwidth]{consoleFigures/experiments/random/L1misses.pdf}
    \end{subfigure}
    \caption{Random model: L1 cache misses.}
    \label{fig:l1misses}
\end{figure}

\begin{figure}
\captionsetup[subfigure]{justification=centering}
    \centering
    \begin{subfigure}{0.5\textwidth}
        \centering
        \includegraphics[page=1,width=1\textwidth]{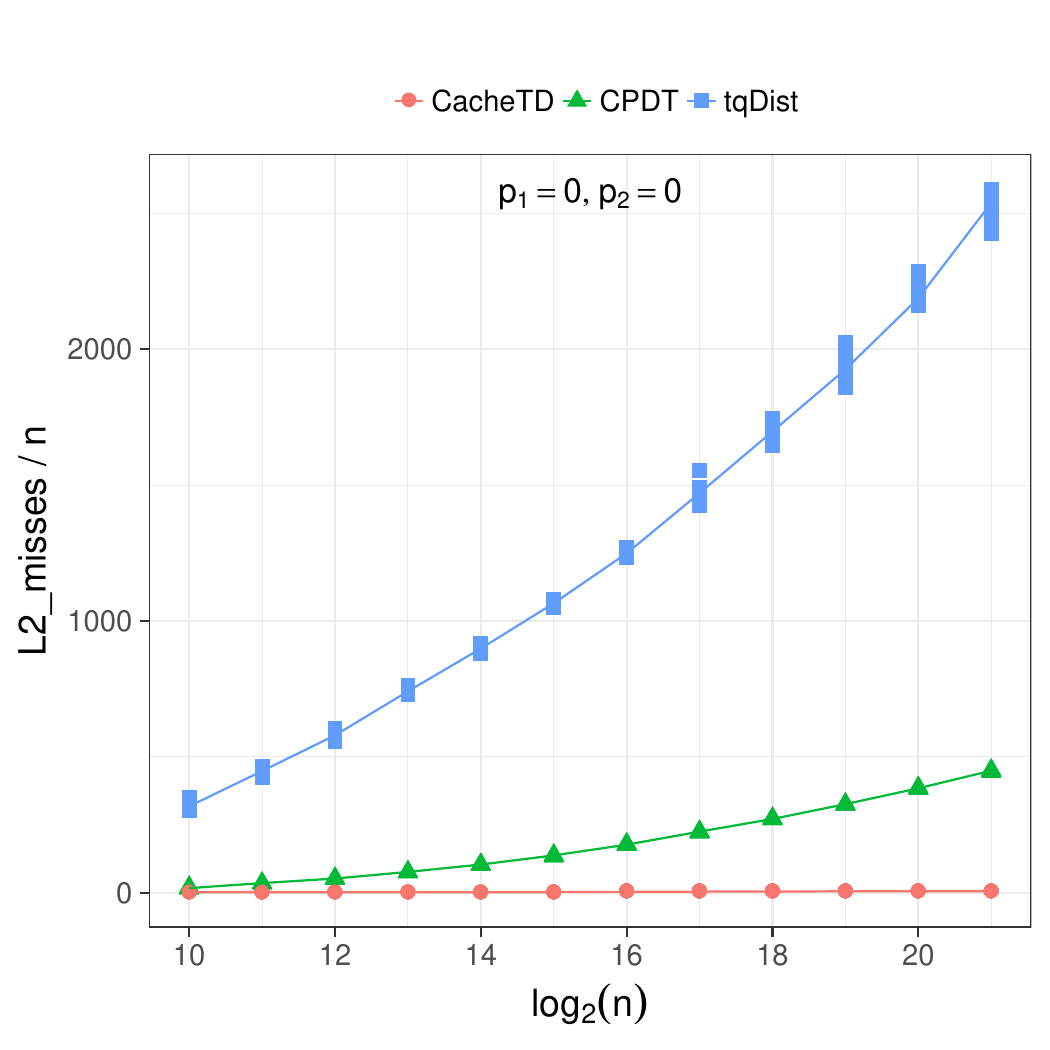}
    \end{subfigure}%
    \begin{subfigure}{0.5\textwidth}
        \centering
         \includegraphics[page=2,width=1\textwidth]{consoleFigures/experiments/random/L2misses.pdf}
    \end{subfigure}%
    \\
     \begin{subfigure}{0.5\textwidth}
        \centering
         \includegraphics[page=3,width=1\textwidth]{consoleFigures/experiments/random/L2misses.pdf}
    \end{subfigure}%
    \begin{subfigure}{0.5\textwidth}
        \centering
         \includegraphics[page=4,width=1\textwidth]{consoleFigures/experiments/random/L2misses.pdf}
    \end{subfigure}
    \\
     \begin{subfigure}{0.5\textwidth}
        \centering
         \includegraphics[page=5,width=1\textwidth]{consoleFigures/experiments/random/L2misses.pdf}
    \end{subfigure}%
    \begin{subfigure}{0.5\textwidth}
        \centering
         \includegraphics[page=6,width=1\textwidth]{consoleFigures/experiments/random/L2misses.pdf}
    \end{subfigure}
    \caption{Random model: L2 cache misses.}
    \label{fig:l2misses}
\end{figure}

\begin{figure}
\captionsetup[subfigure]{justification=centering}
    \centering
    \begin{subfigure}{0.5\textwidth}
        \centering
        \includegraphics[page=1,width=1\textwidth]{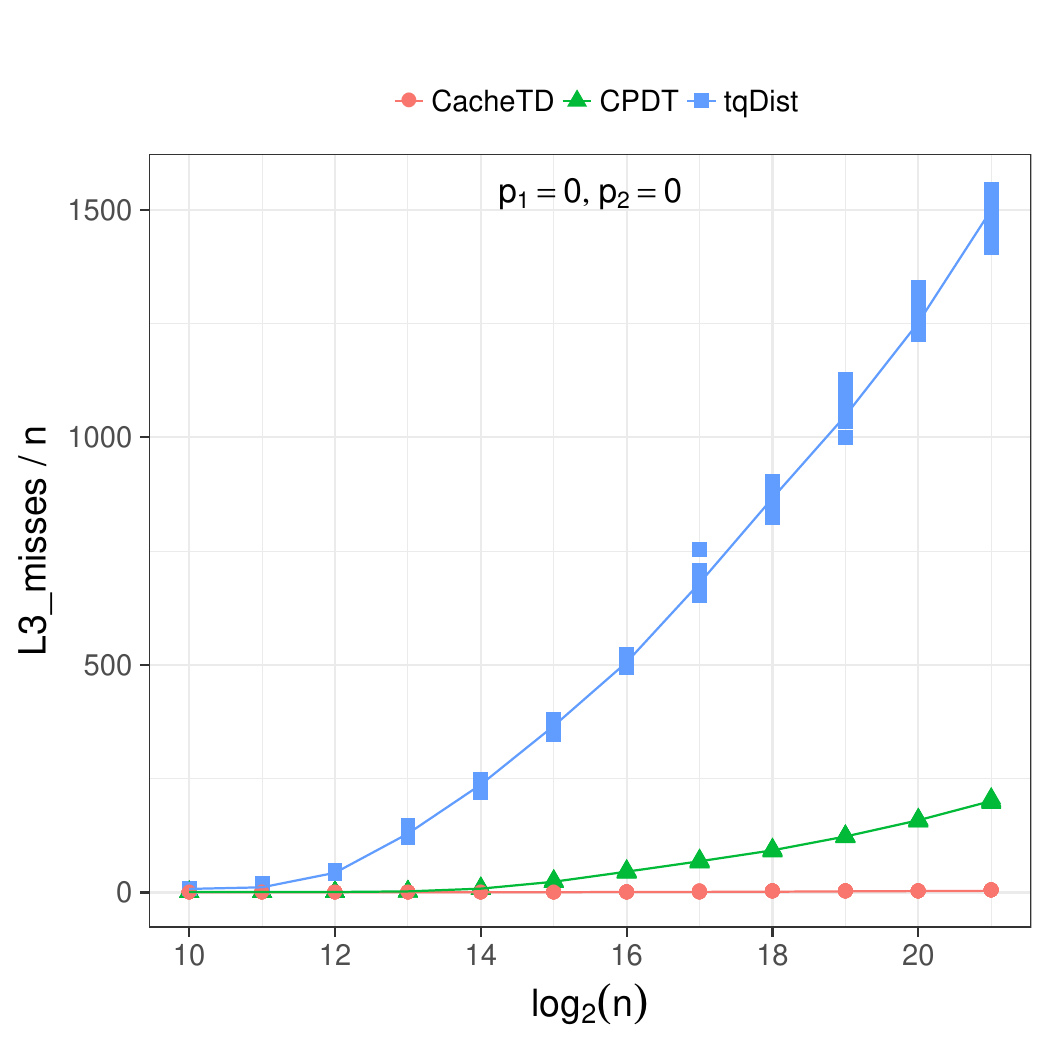}
    \end{subfigure}%
    \begin{subfigure}{0.5\textwidth}
        \centering
         \includegraphics[page=2,width=1\textwidth]{consoleFigures/experiments/random/L3misses.pdf}
    \end{subfigure}%
    \\
     \begin{subfigure}{0.5\textwidth}
        \centering
         \includegraphics[page=3,width=1\textwidth]{consoleFigures/experiments/random/L3misses.pdf}
    \end{subfigure}%
    \begin{subfigure}{0.5\textwidth}
        \centering
         \includegraphics[page=4,width=1\textwidth]{consoleFigures/experiments/random/L3misses.pdf}
    \end{subfigure}
     \\
     \begin{subfigure}{0.5\textwidth}
        \centering
         \includegraphics[page=5,width=1\textwidth]{consoleFigures/experiments/random/L3misses.pdf}
    \end{subfigure}%
    \begin{subfigure}{0.5\textwidth}
        \centering
         \includegraphics[page=6,width=1\textwidth]{consoleFigures/experiments/random/L3misses.pdf}
    \end{subfigure}
    \caption{Random model: L3 cache misses.}
    \label{fig:l3misses}
\end{figure}

\begin{figure}
\captionsetup[subfigure]{justification=centering}
    \centering
    \begin{subfigure}{0.5\textwidth}
        \centering
        \includegraphics[page=1,width=1\textwidth]{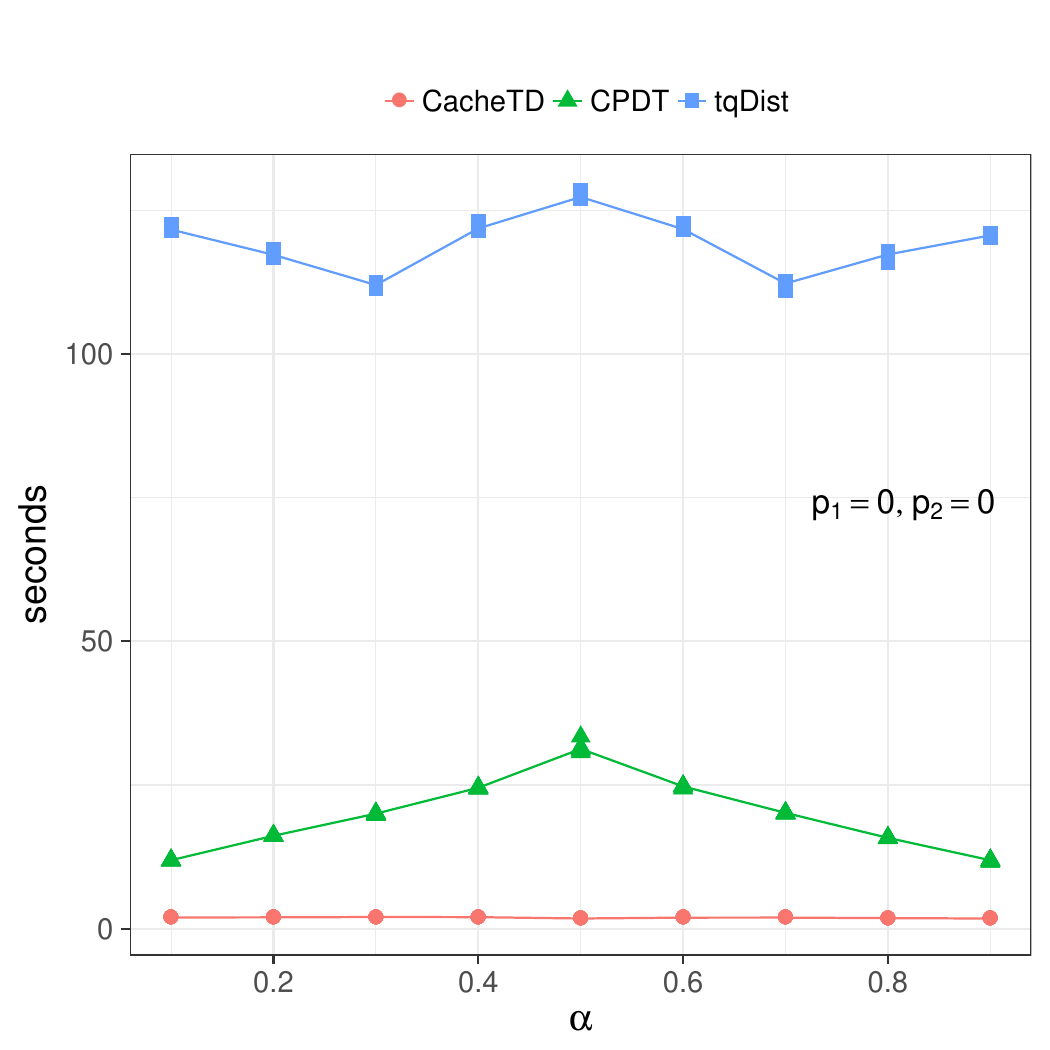}
    \end{subfigure}%
    \begin{subfigure}{0.5\textwidth}
        \centering
         \includegraphics[page=2,width=1\textwidth]{consoleFigures/experiments/skewed/time.pdf}
    \end{subfigure}%
    \\
     \begin{subfigure}{0.5\textwidth}
        \centering
         \includegraphics[page=3,width=1\textwidth]{consoleFigures/experiments/skewed/time.pdf}
    \end{subfigure}%
    \begin{subfigure}{0.5\textwidth}
        \centering
         \includegraphics[page=4,width=1\textwidth]{consoleFigures/experiments/skewed/time.pdf}
    \end{subfigure}%
    \\
     \begin{subfigure}{0.5\textwidth}
        \centering
         \includegraphics[page=5,width=1\textwidth]{consoleFigures/experiments/skewed/time.pdf}
    \end{subfigure}%
     \begin{subfigure}{0.5\textwidth}
        \centering
         \includegraphics[page=6,width=1\textwidth]{consoleFigures/experiments/skewed/time.pdf}
    \end{subfigure}
    \caption{Skewed model: running time ($n=2^{21}$).}
    \label{fig:alphaFigureTime}
\end{figure}

\begin{figure}
\captionsetup[subfigure]{justification=centering}
    \centering
    \begin{subfigure}{0.5\textwidth}
        \centering
        \includegraphics[page=1,width=1\textwidth]{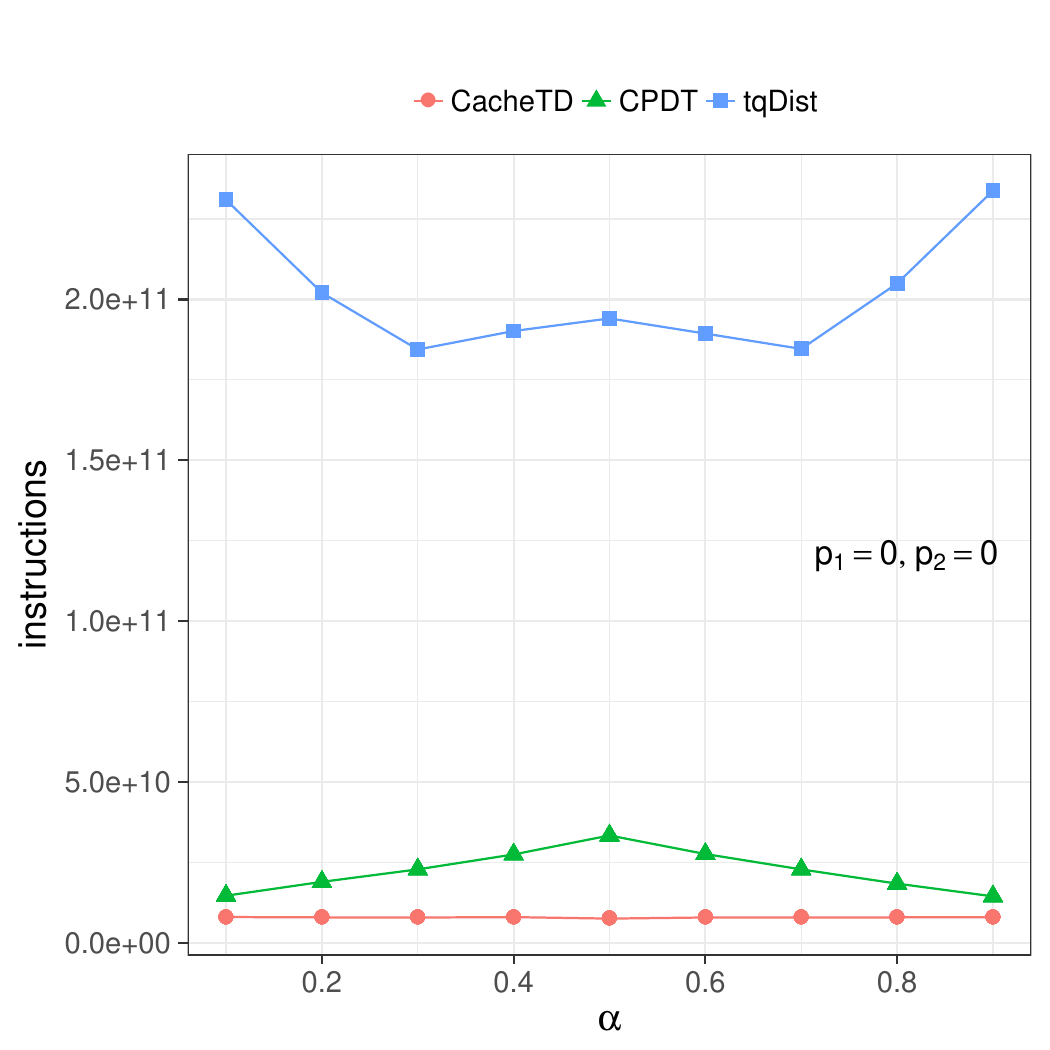}
    \end{subfigure}%
    \begin{subfigure}{0.5\textwidth}
        \centering
         \includegraphics[page=2,width=1\textwidth]{consoleFigures/experiments/skewed/instructions.pdf}
    \end{subfigure}%
    \\
     \begin{subfigure}{0.5\textwidth}
        \centering
         \includegraphics[page=3,width=1\textwidth]{consoleFigures/experiments/skewed/instructions.pdf}
    \end{subfigure}%
    \begin{subfigure}{0.5\textwidth}
        \centering
         \includegraphics[page=4,width=1\textwidth]{consoleFigures/experiments/skewed/instructions.pdf}
    \end{subfigure}%
    \\
     \begin{subfigure}{0.5\textwidth}
        \centering
         \includegraphics[page=5,width=1\textwidth]{consoleFigures/experiments/skewed/instructions.pdf}
    \end{subfigure}%
     \begin{subfigure}{0.5\textwidth}
        \centering
         \includegraphics[page=6,width=1\textwidth]{consoleFigures/experiments/skewed/instructions.pdf}
    \end{subfigure}
    \caption{Skewed model: instructions ($n=2^{21}$).}
    \label{fig:alphaFigureInstructions}
\end{figure}

\begin{figure}
\captionsetup[subfigure]{justification=centering}
    \centering
    \begin{subfigure}{0.5\textwidth}
        \centering
        \includegraphics[page=1,width=1\textwidth]{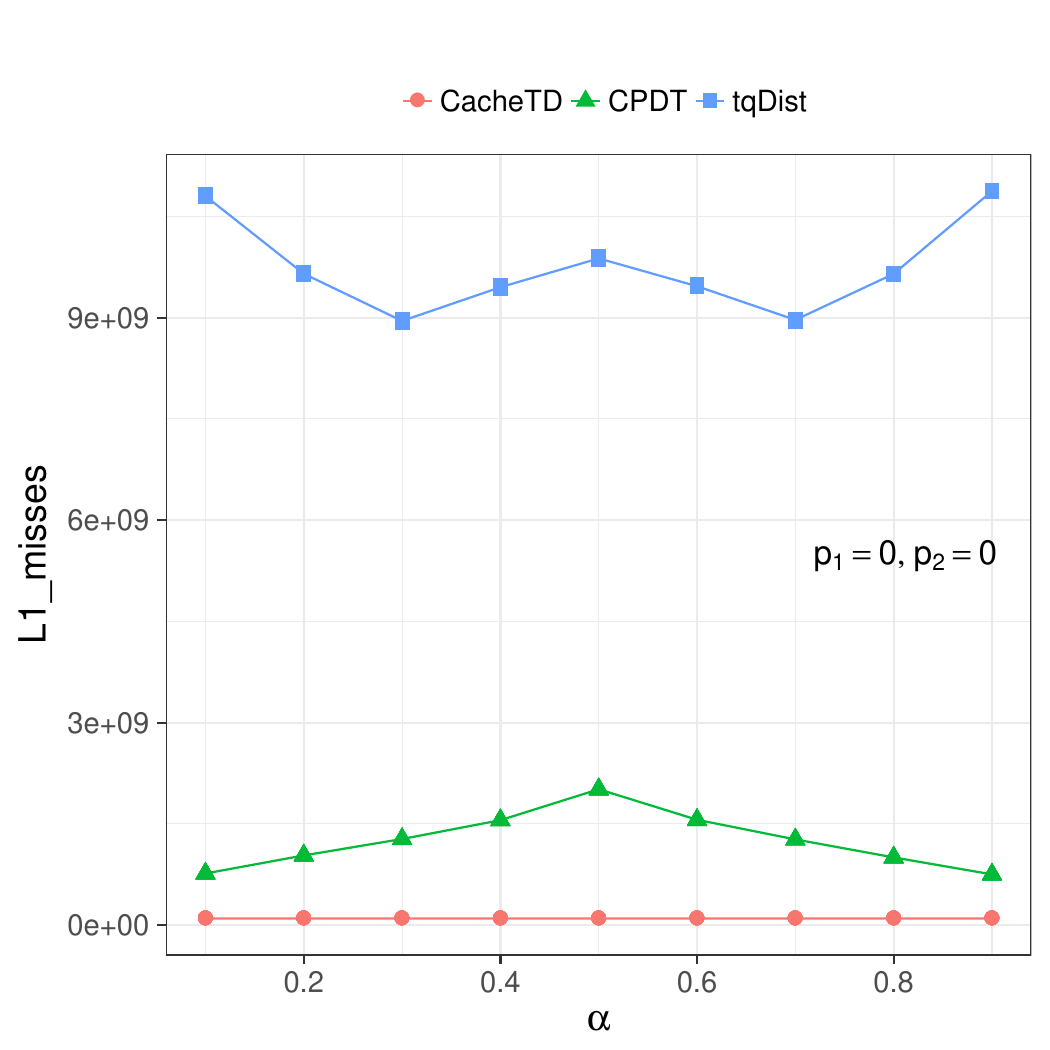}
    \end{subfigure}%
    \begin{subfigure}{0.5\textwidth}
        \centering
         \includegraphics[page=2,width=1\textwidth]{consoleFigures/experiments/skewed/l1misses.pdf}
    \end{subfigure}%
    \\
     \begin{subfigure}{0.5\textwidth}
        \centering
         \includegraphics[page=3,width=1\textwidth]{consoleFigures/experiments/skewed/l1misses.pdf}
    \end{subfigure}%
    \begin{subfigure}{0.5\textwidth}
        \centering
         \includegraphics[page=4,width=1\textwidth]{consoleFigures/experiments/skewed/l1misses.pdf}
    \end{subfigure}%
    \\
     \begin{subfigure}{0.5\textwidth}
        \centering
         \includegraphics[page=5,width=1\textwidth]{consoleFigures/experiments/skewed/l1misses.pdf}
    \end{subfigure}%
     \begin{subfigure}{0.5\textwidth}
        \centering
         \includegraphics[page=6,width=1\textwidth]{consoleFigures/experiments/skewed/l1misses.pdf}
    \end{subfigure}
    \caption{Skewed model: L1 cache misses ($n=2^{21}$).}
    \label{fig:alphaFigureL1misses}
\end{figure}

\begin{figure}
\captionsetup[subfigure]{justification=centering}
    \centering
    \begin{subfigure}{0.5\textwidth}
        \centering
        \includegraphics[page=1,width=1\textwidth]{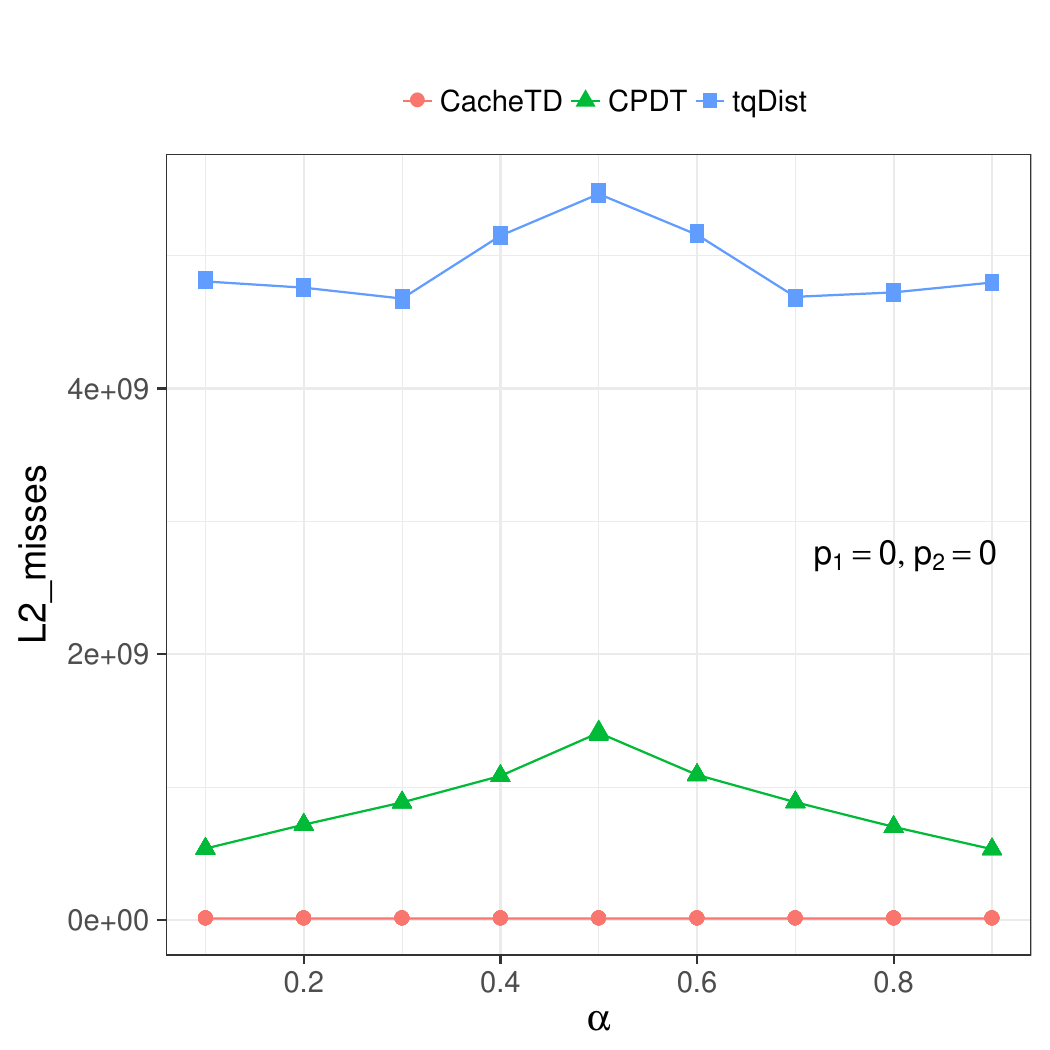}
    \end{subfigure}%
    \begin{subfigure}{0.5\textwidth}
        \centering
         \includegraphics[page=2,width=1\textwidth]{consoleFigures/experiments/skewed/l2misses.pdf}
    \end{subfigure}%
    \\
     \begin{subfigure}{0.5\textwidth}
        \centering
         \includegraphics[page=3,width=1\textwidth]{consoleFigures/experiments/skewed/l2misses.pdf}
    \end{subfigure}%
    \begin{subfigure}{0.5\textwidth}
        \centering
         \includegraphics[page=4,width=1\textwidth]{consoleFigures/experiments/skewed/l2misses.pdf}
    \end{subfigure}%
    \\
     \begin{subfigure}{0.5\textwidth}
        \centering
         \includegraphics[page=5,width=1\textwidth]{consoleFigures/experiments/skewed/l2misses.pdf}
    \end{subfigure}%
     \begin{subfigure}{0.5\textwidth}
        \centering
         \includegraphics[page=6,width=1\textwidth]{consoleFigures/experiments/skewed/l2misses.pdf}
    \end{subfigure}
    \caption{Skewed model: L2 cache misses ($n=2^{21}$).}
    \label{fig:alphaFigureL2misses}
\end{figure}

\begin{figure}
\captionsetup[subfigure]{justification=centering}
    \centering
    \begin{subfigure}{0.5\textwidth}
        \centering
        \includegraphics[page=1,width=1\textwidth]{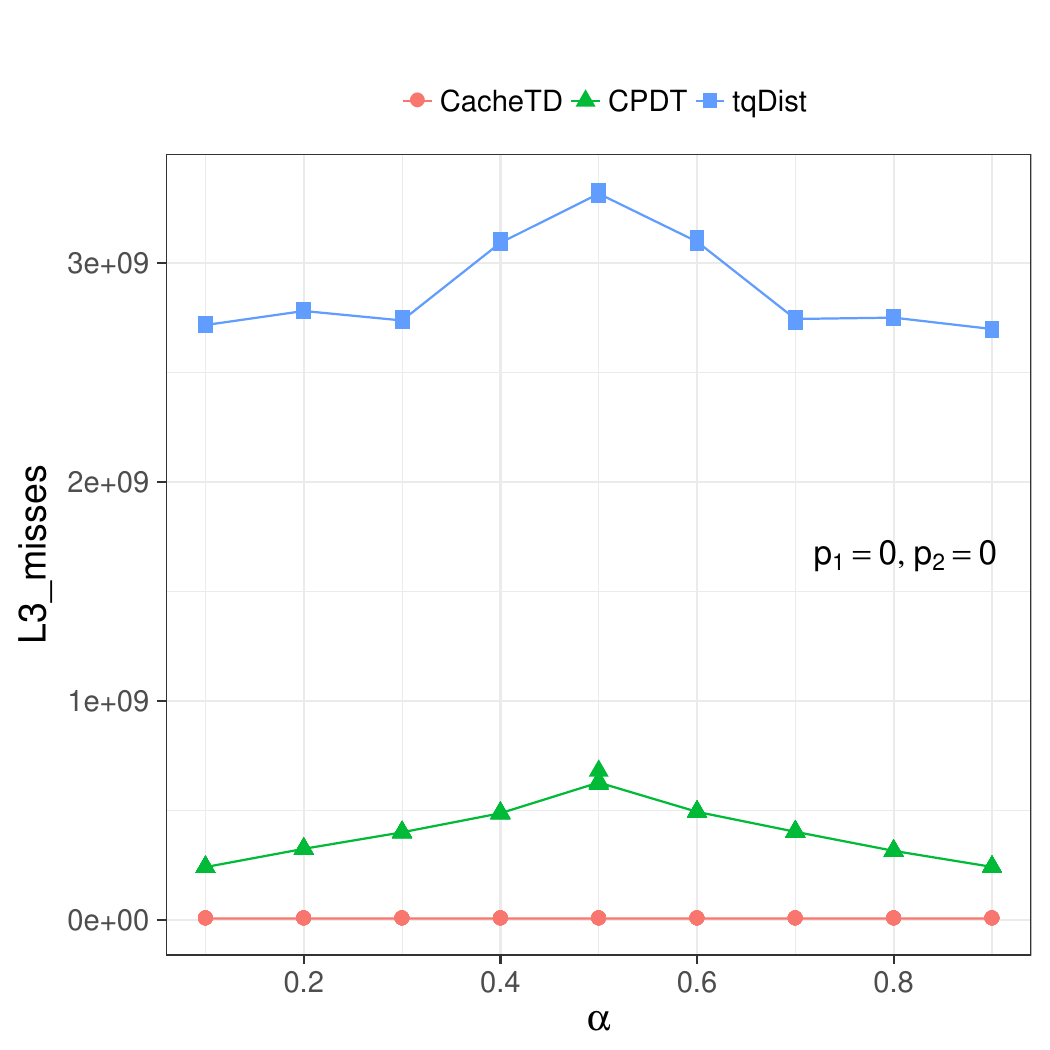}
    \end{subfigure}%
    \begin{subfigure}{0.5\textwidth}
        \centering
         \includegraphics[page=2,width=1\textwidth]{consoleFigures/experiments/skewed/l3misses.pdf}
    \end{subfigure}%
    \\
     \begin{subfigure}{0.5\textwidth}
        \centering
         \includegraphics[page=3,width=1\textwidth]{consoleFigures/experiments/skewed/l3misses.pdf}
    \end{subfigure}%
    \begin{subfigure}{0.5\textwidth}
        \centering
         \includegraphics[page=4,width=1\textwidth]{consoleFigures/experiments/skewed/l3misses.pdf}
    \end{subfigure}%
    \\
     \begin{subfigure}{0.5\textwidth}
        \centering
         \includegraphics[page=5,width=1\textwidth]{consoleFigures/experiments/skewed/l3misses.pdf}
    \end{subfigure}%
     \begin{subfigure}{0.5\textwidth}
        \centering
         \includegraphics[page=6,width=1\textwidth]{consoleFigures/experiments/skewed/l3misses.pdf}
    \end{subfigure}
    \caption{Skewed model: L3 cache misses ($n=2^{21}$).}
    \label{fig:alphaFigureL3misses}
\end{figure}

\begin{figure}
\captionsetup[subfigure]{justification=centering}
    \centering
    \begin{subfigure}{0.5\textwidth}
        \centering
        \includegraphics[page=3,width=1\textwidth]{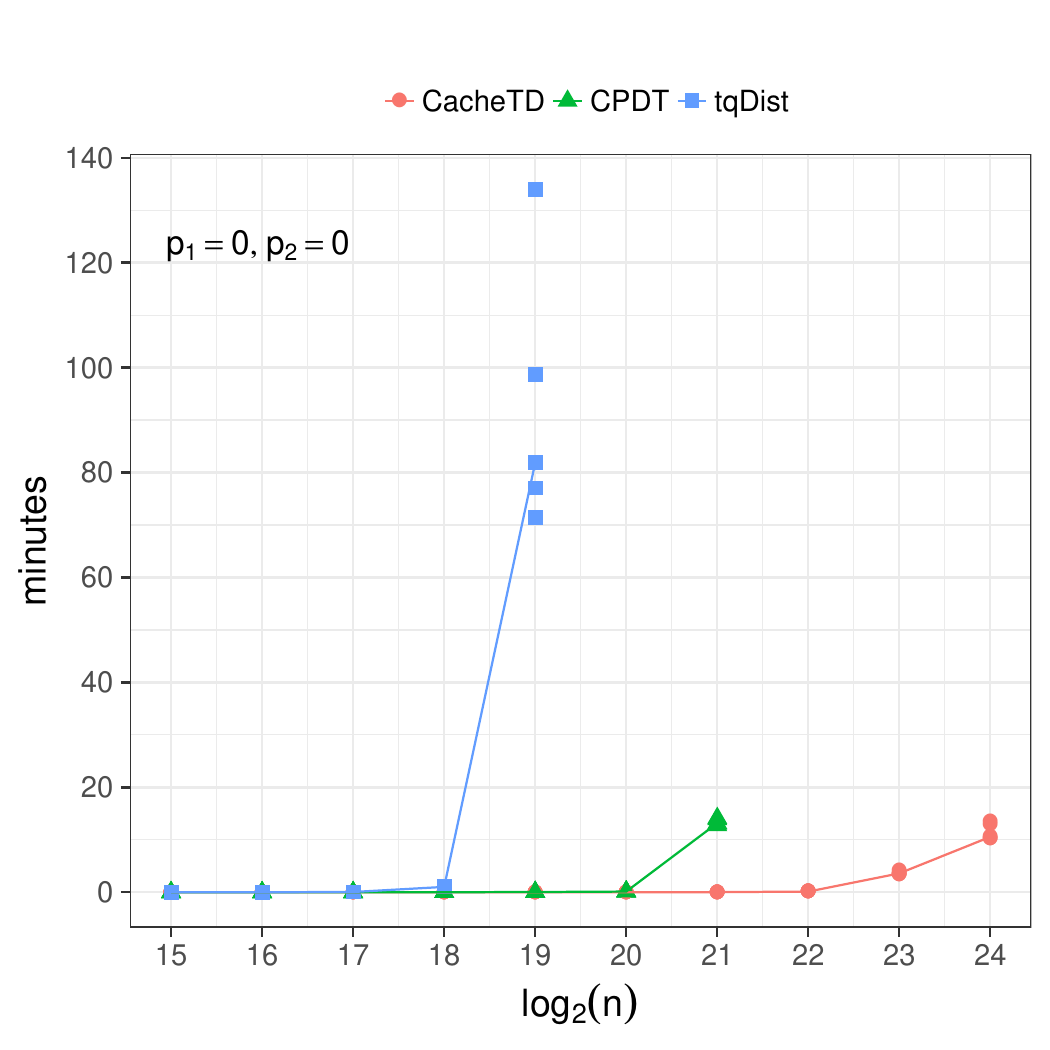}
    \end{subfigure}%
    \begin{subfigure}{0.5\textwidth}
        \centering
         \includegraphics[page=4,width=1\textwidth]{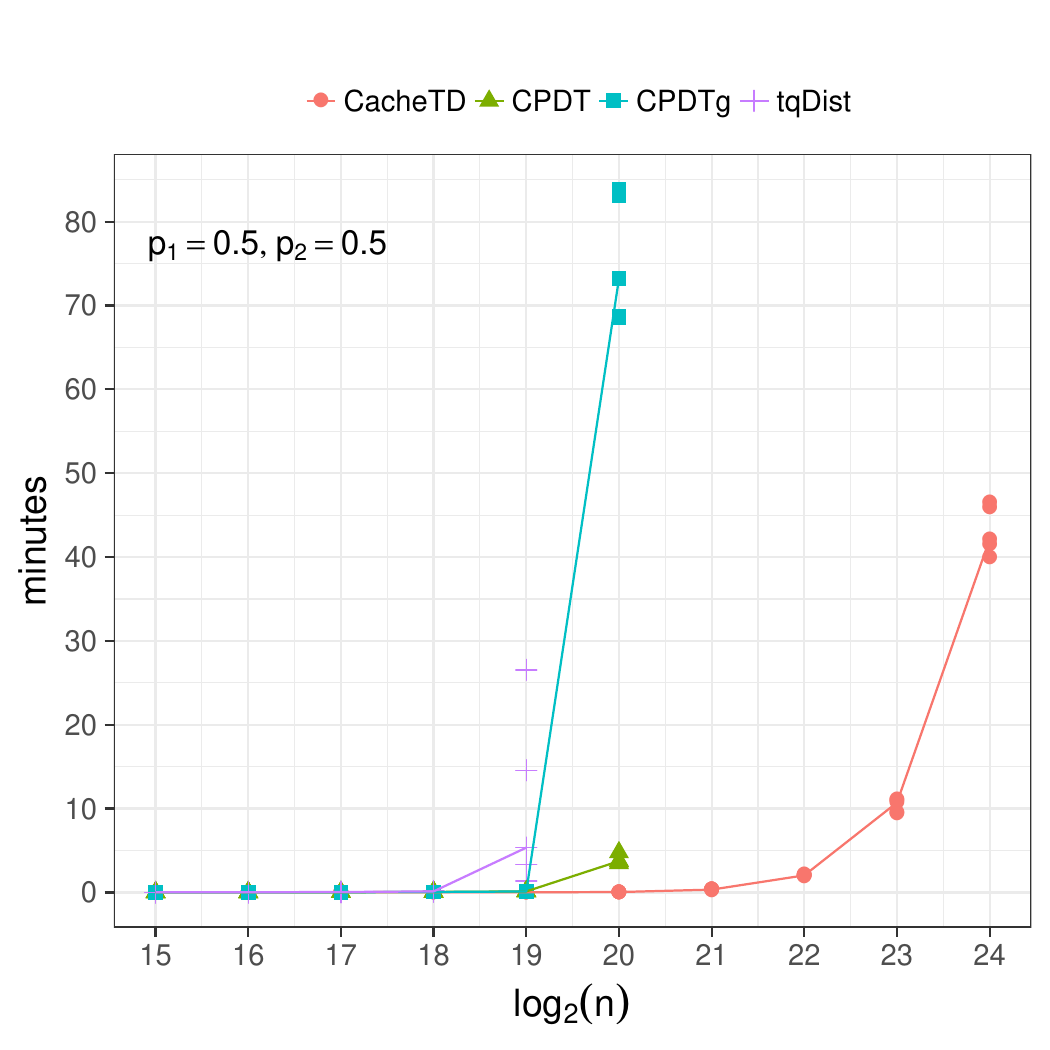}
     \end{subfigure}%
     \\
    \begin{subfigure}{0.5\textwidth}
        \centering
         \includegraphics[page=3,width=1\textwidth]{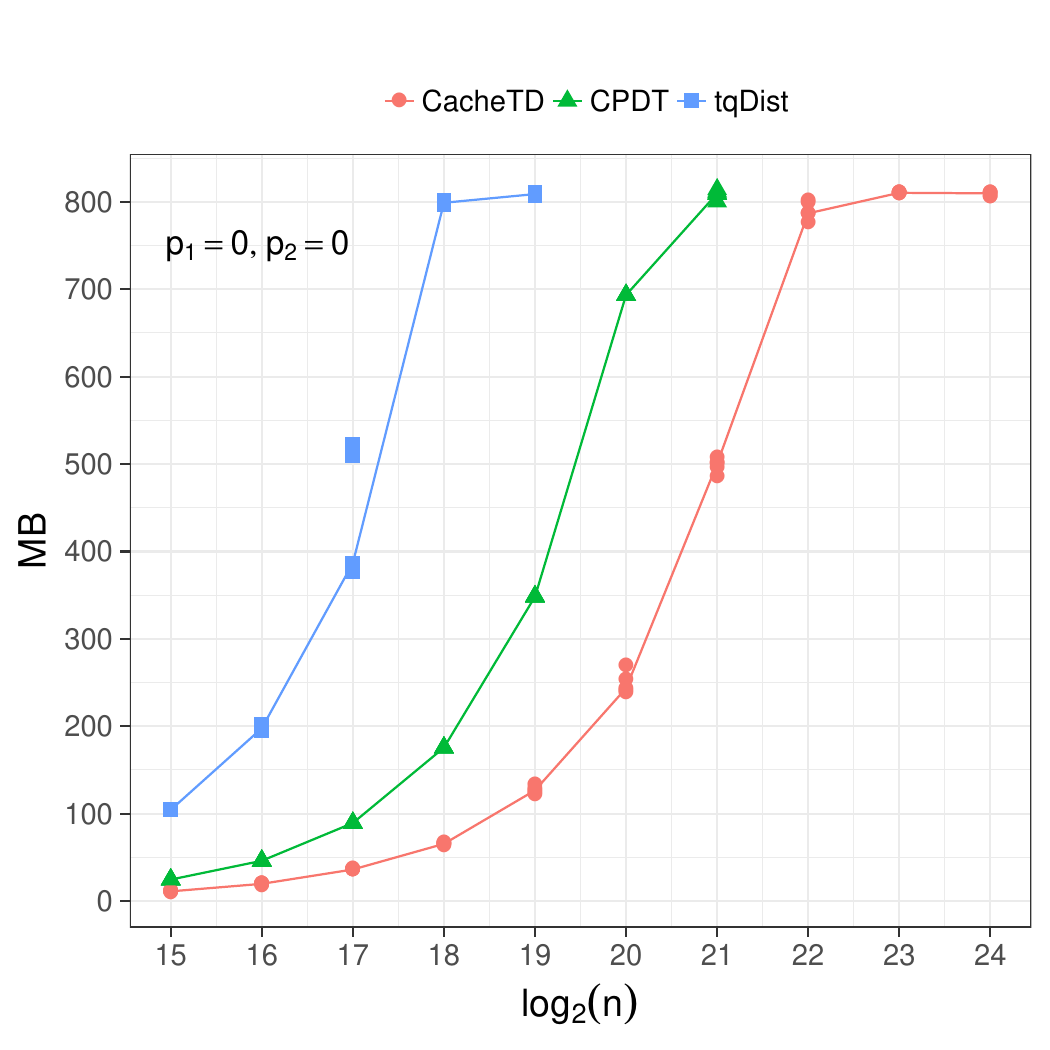}
     \end{subfigure}%
    \begin{subfigure}{0.5\textwidth}
        \centering
         \includegraphics[page=4,width=1\textwidth]{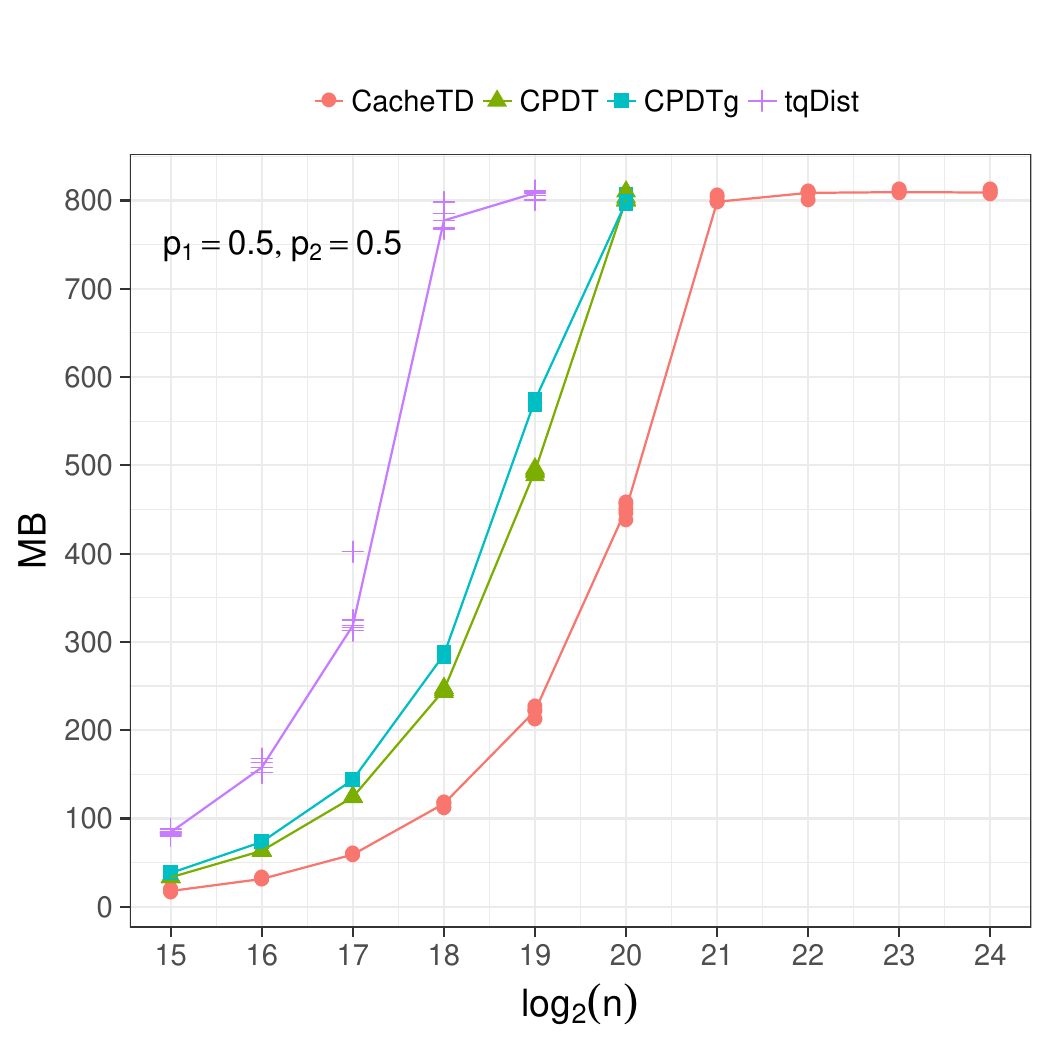}
     \end{subfigure}%
     \\
    \begin{subfigure}{0.5\textwidth}
        \centering
         \includegraphics[page=3,width=1\textwidth]{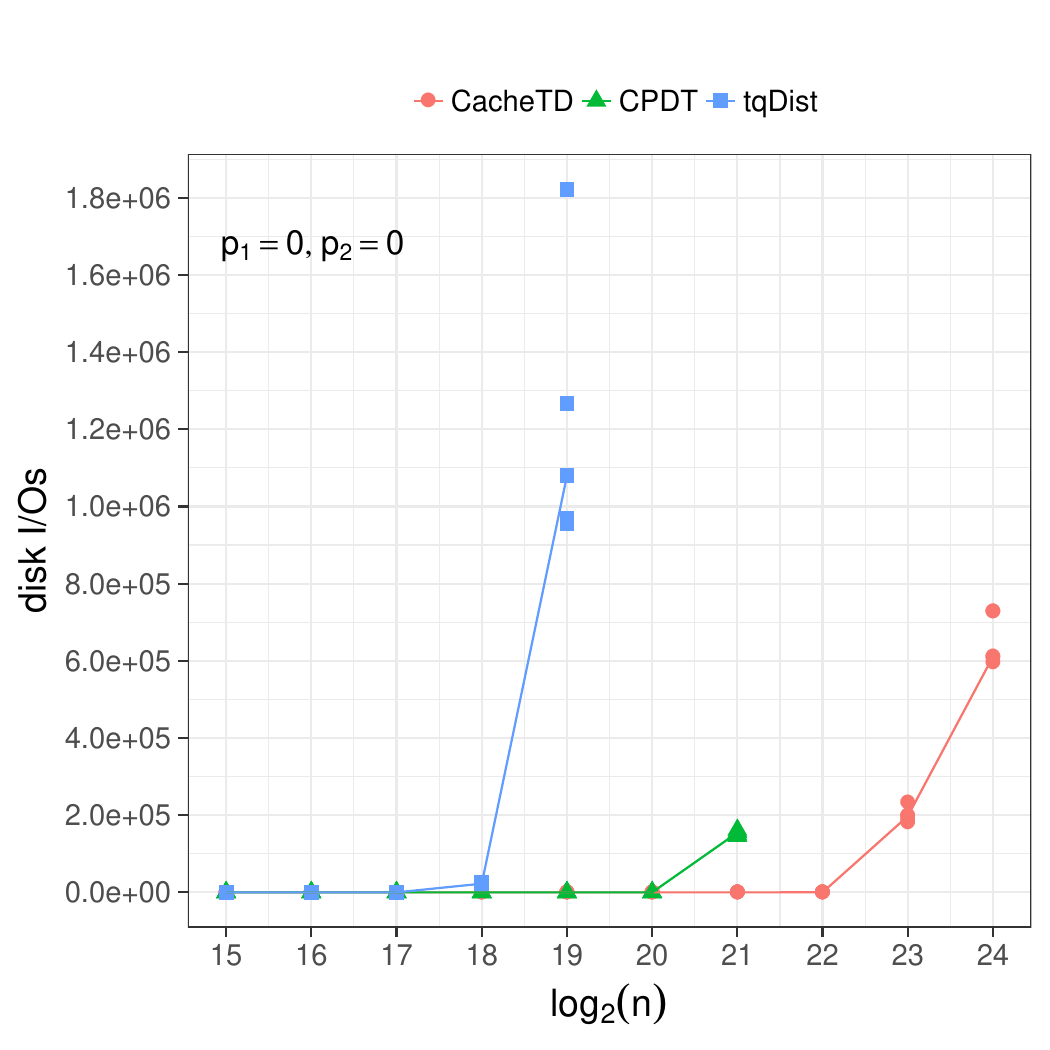}
     \end{subfigure}%
    \begin{subfigure}{0.5\textwidth}
        \centering
         \includegraphics[page=4,width=1\textwidth]{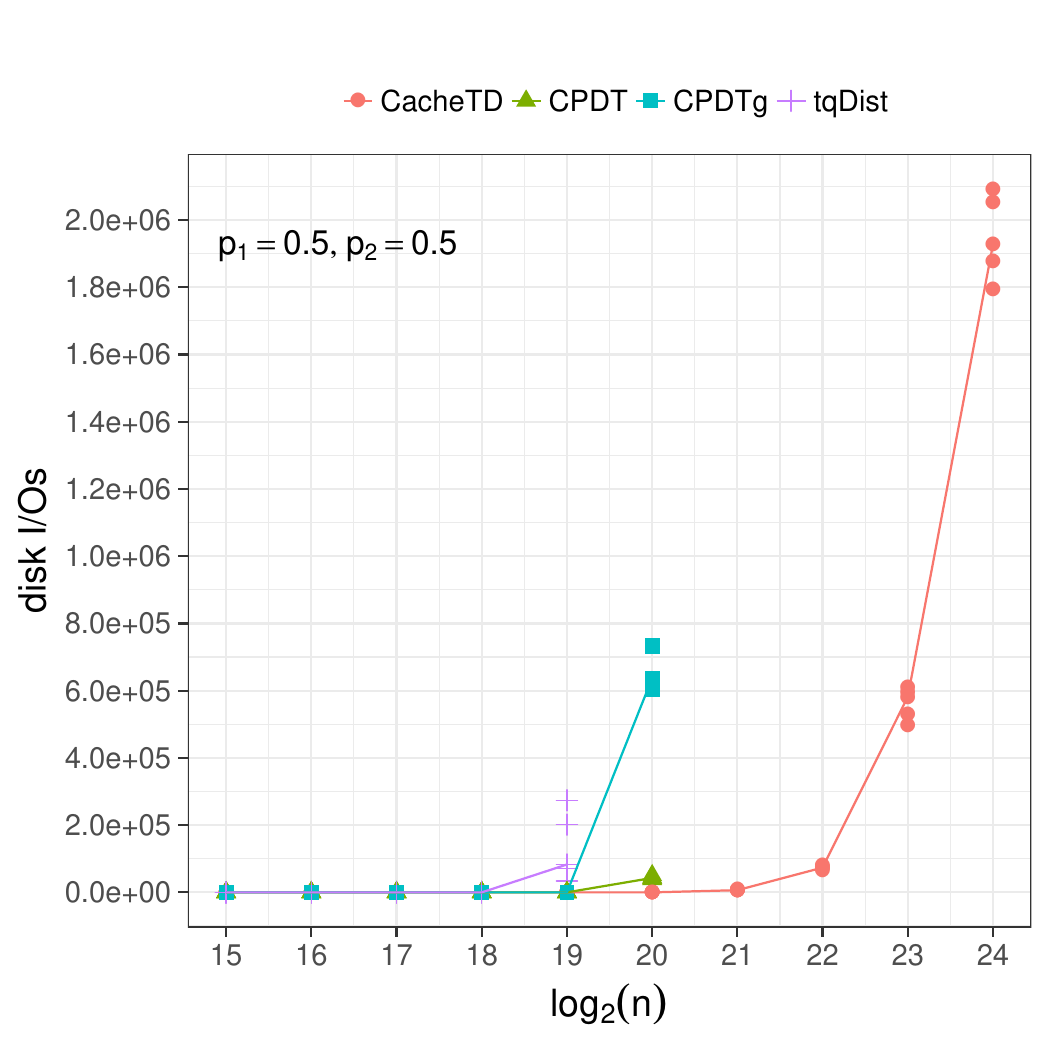}
     \end{subfigure}\\
    \caption{Random model: I/O experiments.}
    \label{fig:IOtimeRandom}
\end{figure}

\begin{figure}
\captionsetup[subfigure]{justification=centering}
    \centering
    \begin{subfigure}{0.5\textwidth}
        \centering
        \includegraphics[page=1,width=1\textwidth]{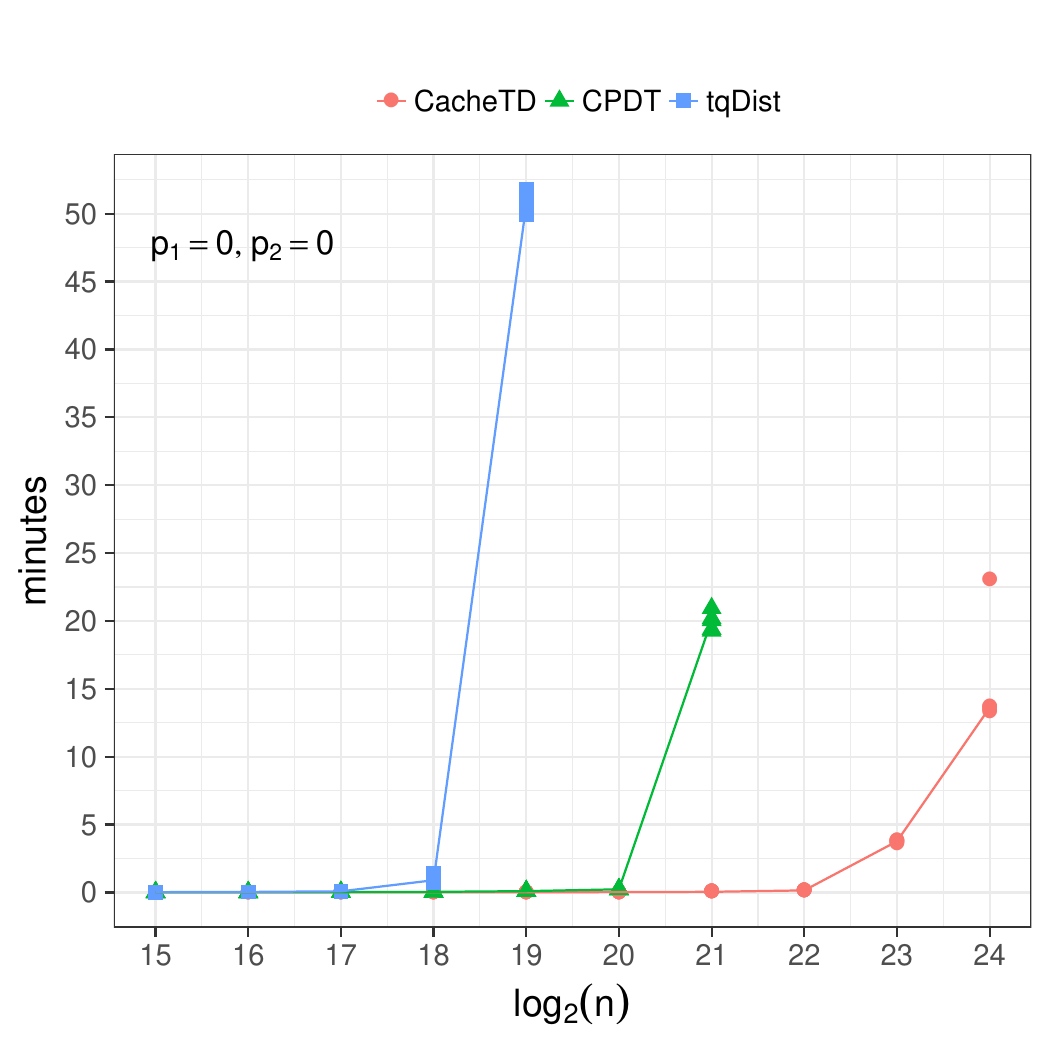}
    \end{subfigure}%
    \begin{subfigure}{0.5\textwidth}
        \centering
         \includegraphics[page=2,width=1\textwidth]{consoleFigures/experiments/io/time.pdf}
     \end{subfigure}%
     \\
    \begin{subfigure}{0.5\textwidth}
        \centering
         \includegraphics[page=1,width=1\textwidth]{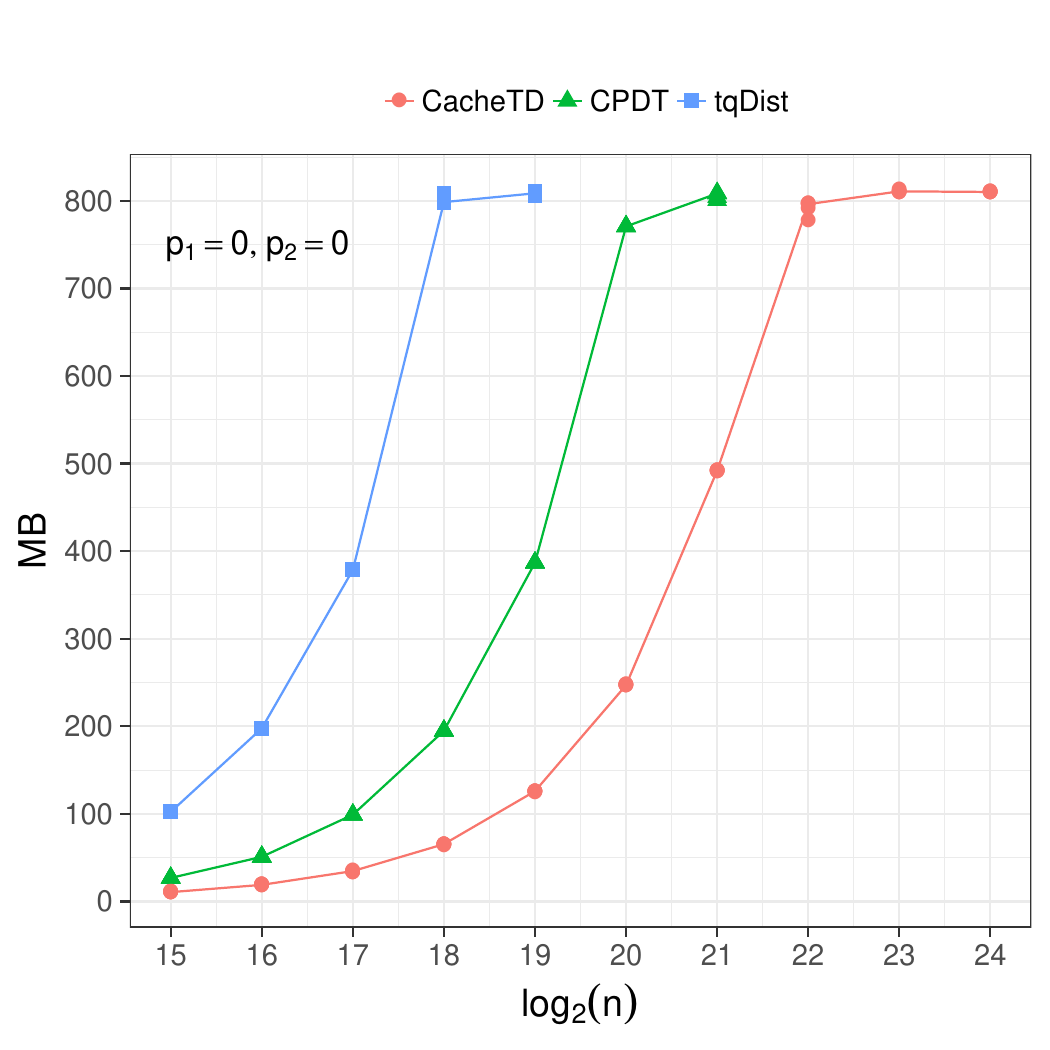}
     \end{subfigure}%
    \begin{subfigure}{0.5\textwidth}
        \centering
         \includegraphics[page=2,width=1\textwidth]{consoleFigures/experiments/io/space.pdf}
     \end{subfigure}%
     \\
    \begin{subfigure}{0.5\textwidth}
        \centering
         \includegraphics[page=1,width=1\textwidth]{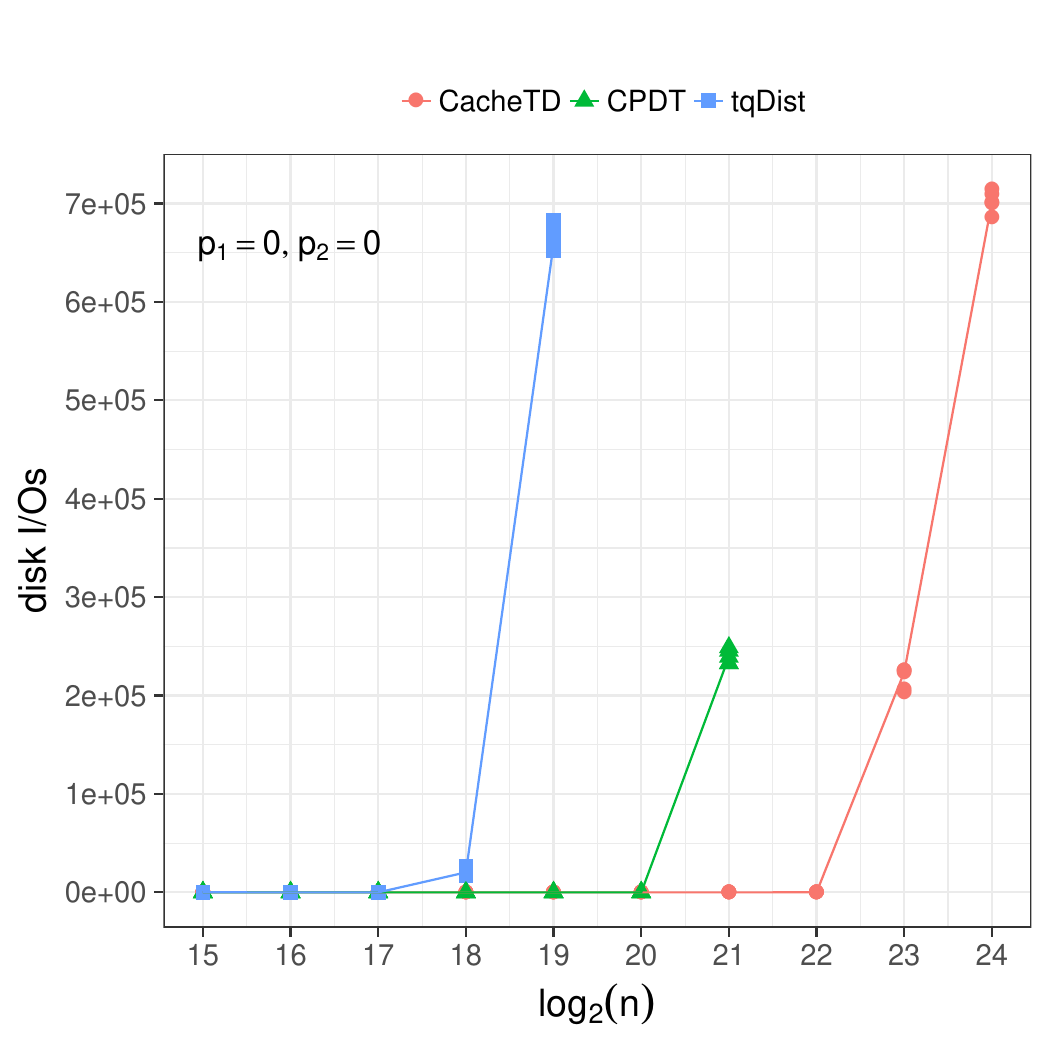}
     \end{subfigure}%
    \begin{subfigure}{0.5\textwidth}
        \centering
         \includegraphics[page=2,width=1\textwidth]{consoleFigures/experiments/io/ios.pdf}
     \end{subfigure}\\
    \caption{Skewed model: I/O experiments with $\alpha = 0.5$.}
    \label{fig:IOtimeSkewed}
\end{figure}
\end{document}